\documentclass[british,english,1p,authoryear,11pt,enabledeprecatedfontcommands]{scrartcl}
\usepackage[T1]{fontenc}
\usepackage[latin9]{inputenc}
\usepackage{geometry}
\usepackage{color}
\usepackage{babel}
\usepackage{prettyref}
\usepackage{refstyle}
\usepackage{mathrsfs}
\usepackage{mathtools}
\usepackage{amsmath}
\usepackage{amsthm}
\usepackage{amssymb}
\usepackage{graphicx}
\usepackage{esint}
\usepackage[authoryear]{natbib}
\usepackage[unicode=true,
 bookmarks=false,
 breaklinks=false,pdfborder={0 0 0},pdfborderstyle={},backref=false,colorlinks=true]
 {hyperref}

\makeatletter

\AtBeginDocument{\providecommand\lemref[1]{\ref{lem:#1}}}
\RS@ifundefined{subsecref}
  {\newref{subsec}{name = \RSsectxt}}
  {}
\RS@ifundefined{thmref}
  {\def\RSthmtxt{theorem~}\newref{thm}{name = \RSthmtxt}}
  {}
\RS@ifundefined{lemref}
  {\def\RSlemtxt{Lemma~}\newref{lem}{name = \RSlemtxt}}
  {}

\newcommand{\lyxaddress}[1]{
	\par {\raggedright #1
	\vspace{1.4em}
	\noindent\par}
}

\theoremstyle{plain}
\newtheorem{thm}{\protect\theoremname}
\theoremstyle{plain}
\newtheorem{prop}[thm]{\protect\propositionname}
\theoremstyle{remark}
\newtheorem*{rem*}{\protect\remarkname}
\theoremstyle{plain}
\newtheorem{cor}[thm]{\protect\corollaryname}
\theoremstyle{plain}
\newtheorem{lem}[thm]{Lemma}

\usepackage[geometry]{ifsym}

\usepackage[ocgcolorlinks]{ocgx2}
\definecolor{myblue}{RGB}{20,45,160}
\hypersetup{colorlinks,
            urlcolor=myblue,
            linkcolor=myblue,
            citecolor=myblue}

\newrefformat{eq}{\textup{Equation (\ref{#1})}}
\newrefformat{prop}{\textup{Proposition \ref{#1}}}
\newrefformat{fig}{\textup{Figure \ref{#1}}}
\newrefformat{cor}{\textup{Corollary \ref{#1}}}
\newrefformat{lem}{\textup{Lemma \ref{#1}}}

\DeclareSymbolFont{extraitalic}{U}{zavm}{m}{it}

\newcommand{\Var}{\mathrm{Var}}

\usepackage{xr}
\makeatother

\addto\captionsbritish{\renewcommand{\corollaryname}{Corollary}}
\addto\captionsbritish{\renewcommand{\propositionname}{Proposition}}
\addto\captionsbritish{\renewcommand{\remarkname}{Remark}}
\addto\captionsbritish{\renewcommand{\theoremname}{Theorem}}
\addto\captionsenglish{\renewcommand{\corollaryname}{Corollary}}
\addto\captionsenglish{\renewcommand{\propositionname}{Proposition}}
\addto\captionsenglish{\renewcommand{\remarkname}{Remark}}
\addto\captionsenglish{\renewcommand{\theoremname}{Theorem}}
\providecommand{\corollaryname}{Corollary}
\providecommand{\propositionname}{Proposition}
\providecommand{\remarkname}{Remark}
\providecommand{\theoremname}{Theorem}

\begin{document}
\title{\textrm{Exact Expressions for the Log-likelihood's Hessian in Multivariate
Continuous-Time Continuous-Trait Gaussian Evolution along a Phylogeny}}
\author{Woodrow Hao Chi Kiang\\
woodrow.hao.chi.kiang@liu.se, hello@hckiang.com}
\maketitle

\lyxaddress{STIMA, Department of Computer and Information Science, Linköping
University, Linköping, Sweden}

\global\long\def\children{\textrm{{Ch}}}%

\global\long\def\parent{\textrm{{Pa}}}%

\global\long\def\descendents{\textrm{{Ch}}^{*}}%

\global\long\def\ancestors{\textrm{{Pa}}^{*}}%

\global\long\def\etaone{\eta_{1}}%
 
\global\long\def\etaj{\eta_{j}}%
 
\global\long\def\etaJ{\eta_{J_{u}}}%

\global\long\def\Sig{\Sigma}%

\global\long\def\Sigx{L}%

\global\long\def\Siglr{ \begin{bmatrix} \Sig_{\etaone}  &  0  &  0  &  \cdots\\
 0  &  \Sig_{\eta_{2}}  &  0  &  \cdots\\
 0  &  0  &  \ddots &  \vdots\\
 0  &  0  &  \cdots &  \Sig_{\etaJ} 
\end{bmatrix}}%

\global\long\def\Siglrinv{ \begin{bmatrix} \Sig_{\etaone}^{-1}  &  0  &  0  &  \cdots\\
 0  &  \Sig_{\eta_{2}}^{-1}  &  0  &  \cdots\\
 0  &  0  &  \ddots &  \vdots\\
 0  &  0  &  \cdots &  \Sig_{\etaJ}^{-1} 
\end{bmatrix}}%

\global\long\def\E{\mathbb{E}}%

\global\long\def\Cov{\mathrm{Cov}}%

\global\long\def\Var{\mathrm{Var}}%

\global\long\def\R{\mathbb{R}}%

\global\long\def\sumju{\sum_{1\le j\le J_{u}}}%

\global\long\def\sumc{\sum_{1\le j\le J_{u}}c_{\etaj}}%

\global\long\def\sumgam{\sum_{1\le j\le J_{u}}\gamma_{\etaj}}%

\global\long\def\sumOmega{\sum_{1\le j\le J_{u}}\Omega_{\etaj}}%

\global\long\def\sumDelta{\sum_{1\le j\le J_{u}}\Delta_{\etaj}}%

\global\long\def\sumjo{\sum_{1\le j\le J_{o}}}%

\global\long\def\rlr{ \begin{bmatrix} r_{\etaone} \\
 \vdots\\
 r_{\etaJ} \\
 
\end{bmatrix}}%

\global\long\def\xlr{ \begin{bmatrix} x_{\etaone} \\
 \vdots\\
 x_{\etaJ} \\
 
\end{bmatrix}}%

\global\long\def\xlrT{ \begin{bmatrix} x_{\etaone}^{T}  &  \cdots &  x_{\etaJ}^{T} \end{bmatrix}}%

\global\long\def\qlr{ \begin{bmatrix} q_{\etaone} \\
 \vdots\\
 q_{\etaJ} 
\end{bmatrix}}%

\global\long\def\qlrT{ \begin{bmatrix} q_{\etaone}^{T}  &  \cdots &  q_{\etaJ}^{T} \end{bmatrix}}%

\global\long\def\Philr{ \begin{bmatrix} \Phi_{\etaone} \\
 \vdots\\
 \Phi_{\etaJ} 
\end{bmatrix}}%

\global\long\def\PhilrT{ \begin{bmatrix} \Phi_{\etaone}^{T}  &  \cdots &  \Phi_{\etaJ}^{T} \end{bmatrix}}%

\global\long\def\Psilr{ \begin{bmatrix} \Psi_{\etaone} \\
 \vdots\\
 \Psi_{\etaJ} 
\end{bmatrix}}%

\global\long\def\PsilrT{ \begin{bmatrix} \Psi_{\etaone}^{T}  &  \cdots &  \Psi_{\etaJ}^{T} \end{bmatrix}}%

\global\long\def\PsilrPhiu{ \begin{bmatrix} \Psi_{\etaone}\\
 \vdots\\
 \Psi_{\etaJ} 
\end{bmatrix}\Phi_{u}}%

\global\long\def\PsilrTPhiuT{ \Phi_{u}^{T} \begin{bmatrix} \Psi_{\etaone}^{T}  &  \cdots &  \Psi_{\etaJ}^{T} \end{bmatrix}}%

\global\long\def\tlrT{ \begin{bmatrix} t_{\etaone}^{T}  &  \cdots &  t_{\etaJ}^{T} \end{bmatrix}}%

\global\long\def\tlr{ \begin{bmatrix} t_{\etaone} \\
 \vdots\\
 t_{\etaJ} 
\end{bmatrix}}%

\global\long\def\GlrT{ \begin{bmatrix} G_{\etaone}^{T}  &  \cdots &  G_{\etaJ}^{T} \end{bmatrix}}%

\global\long\def\Glr{ \begin{bmatrix} G_{\etaone} \\
 \vdots\\
 G_{\etaJ} 
\end{bmatrix}}%

\global\long\def\HlrT{ \begin{bmatrix} H_{\etaone}^{T}  &  \cdots &  H_{\etaJ}^{T} \end{bmatrix}}%

\global\long\def\Hlr{ \begin{bmatrix} H_{\etaone} \\
 \vdots\\
 H_{\etaJ} 
\end{bmatrix}}%

\global\long\def\GLInv{\mathscr{G}_{Linv}}%

\global\long\def\setchar#1{1_{#1}}%

\global\long\def\hrooteta#1{\frac{\partial^{2}{#1}_{u_{1}}}{\partial\theta_{\eta_{1}}\partial\theta_{\eta_{2}}}}%

\global\long\def\heta#1{\frac{\partial^{2}{#1}}{\partial\theta_{\eta_{1}}\partial\theta_{\eta_{2}}}}%

\global\long\def\hgen#1#2#3{\frac{\partial^{2}{#1}_{\eta}}{\partial#2\partial#3}}%

\global\long\def\detaone#1{\frac{\partial{#1}}{\partial\theta_{\eta_{1}}}}%

\global\long\def\detatwo#1{\frac{\partial{#1}}{\partial\theta{}_{\eta_{2}}}}%

\global\long\def\deta#1{\frac{\partial{#1}}{\partial\theta{}_{\eta}}}%

\global\long\def\detat#1#2{\frac{\partial{#1}_{\eta}}{\partial#2}}%

\global\long\def\evalat#1#2{\left.#1\right|_{#2}}%

\global\long\def\symadd#1{\mathcal{S}(#1)}%

\global\long\def\lca#1#2{\mathcal{L}\left(#1,#2\right)}%

\section{Introduction}

It has long been known that phenotypes of multiple related species
should not be modelled as independent variables drawn from the same
distribution. Since \citet{felsenstein1985} introduced the independent
contrast algorithm to circumvent the nonindependence between multiple
species and \citet{hansen1997ou} described the phylogenetic Ornstein-Uhlenbeck
model, various more sophisticated SDE-type phylogenetic comparative
methods (PCM) has been developed. Recently, \citet{mitov2020} proposed
a way to calculate in linear time the likelihood of a subfamily of
these models, denoted by them as $\GLInv$, which covers a wide range
of continuous-time continuous-trait Gaussian Markov processes in which
the trait vector's variance-covariance matrix at time $t$ is invariant
with respect to the any ancestral trait values. To name a few, the
$\GLInv$ family covers the vanilla phylogenetic BM \citep{felsenstein1985}
and OU \citep{hansen1997ou} models, a class of phylogenetic regression
models described in \citet{bartoszek_phylogenetic_2012} (implemented
in the \texttt{R} package \texttt{mvSLOUCH}), the early burst \citep{harmon_early_2010}
and accelerating-decelerating \citep{blomberg_testing_2003} that
are implemented in the \texttt{R} package \texttt{geiger} \citep{harmon_geiger_2008}
and \texttt{mvMORPH} \citep{clavel_mvmorph_2015} among others. The
$\GLInv$ family is made rather general by avoiding the use of any
mathematical structures specific to the BM and OU; and as a result,
the \texttt{R} package \texttt{PCMBase}, which accompanied their paper,
is capable of computing the likelihoods a broad class of continuous-trait
models in the presence of multiple evolutionary regimes, missing or
disappeared traits, measurement errors, and evolution with jumps in
some internal lineages, among numerous other possibilities, while
avoiding the inversion the big covariance matrix of all traits of
all tips that is usually done when one calculates a Gaussian likelihood
(this is done, for example, in the \texttt{R} package \texttt{ouch}
\citep{butler_phylogenetic_2004}).

While such computation of a wide class of the PCM likelihood functions
is useful for obtaining point estimates \textit{via} maximum likelihood
estimation, uncertainty quantificaiton in PCM frequently relies on
generic numerical methods that are not specific to PCM, such as parametric
bootstrap that is implemented in \texttt{ouch} \citep{butler_phylogenetic_2004}
and Bayesian modelling such as MCMC that is implemented in \texttt{geiger}
and \texttt{OUwie} \citep{beaulieu_modeling_2012}, or \texttt{pcmabc}
\citep{Bartoszek_2019}, which implements approximate Bayesian computation.

In this paper, I will present the closed-form expressions for the
second derivatives of the $\GLInv$ family. The resulting Hessian
matrix may be used to obtain a confidence ellipse of the maximum likelihood
estimates given that the log-likelihood surface is sufficiently close
to a multivariate quadratic function. The closed-form expressions
are programmed in an R package called \texttt{glinvci}, which also
includes some out-of-box functionalities to analytically compute the
BM and OU model's confidence ellipses, as well as facilities for the
users to compute the Hessians of their own models within the family.
But perhaps more importantly, given that the family covers so many
Gaussian PCM models, the recursive mathematical structures in the
Hessian may be useful for further mathematical analysis, especially
as papers like \citet{cooper_cautionary_2016} has led to some awareness
of the statistical properties of these PCM models.

When I derive the second derivatives, a linear-time algorithm for
the $\GLInv$ likelihood is also presented. This is different from
the linear-time algorithm that \texttt{PCMBase} and \texttt{PCMFit}
\citep{mitov_automatic_2019} uses, as they expand the quadratic form
in the likelihood into six terms and recurse on the corresponding
six coefficients, while mine is based on Woodbury formulae, which
is more similar to that of \citet{ho_linear-time_2014}. The main
difference between my likelihood algorithm and that of \citet{ho_linear-time_2014}
is that theirs requires the variance-covariance matrix of the traits
to satisfy a ``three-point condition'' (see their Definition 1 and
Theorem 2) and is less general than $\GLInv$; while mine does not
require such condition and is fully multivariate, but comes with a
cost of requiring full-rank matrix update rather than their faster
rank-one covariance matrix update, although both are linear-time with
respect to number of tips of the phylogeny and both uses similar Woodbury-formula
techniques. Also, \citet{bastide2020} have used a similar technique
to obtain both a linear-time algorithm and close-form gradients, but
with a slightly different formulation and perhaps less generality
than what is presented here. Nonetheless, my main objective here is
to obtain the Hessian, and my likelihood and gradient computation
are derived in a way that produce by-product matrices that are neccessary
for the Hessian's computation.

\prettyref{sec:Model} be devoted to introducing the model, its likelihood
and the gradients. Then in \prettyref{sec:Hessian} the Hessian of
the general case will be presented, and in \prettyref{sec:The-Ornstein-Uhlenbeck-Model}
we will discuss the secord-order derivatives regarding the special
case of the widely-used Ornstein-Uhlenbeck model. \prettyref{sec:Handling-of-Missing}
will discuss how missing traits data can be handled. Finally, \prettyref{sec:Discussion}
is a discussion section. Because of their length, the mathematical
proofs are presented in Supplementary Information.

\section{Model\label{sec:Model}}

Assume that we have a known rooted phylogenetic tree representing
the evolution of $N$ species, and we have observed, for each tip
$\eta$, a multivariate real-valued trait vector $x_{\eta}\in\mathbb{R}^{k_{\eta}}$,
which is resulted from an evolutionary process along the tree. Each
branch in the phylogenetic tree is associated with a positive branch
length $t_{\eta}$ representing time, where $\eta$ denotes the node
at which the branch ends. We denote, for any node $\eta$ in a phylogenetc
tree, its direct ancestor as $\parent(\eta)$, the set of all ancestor
as $\ancestors(\eta)$, the set of direct children as $\children(\eta)$,
and the set of all descendents as $\descendents(\eta)$.

We use a convention that, if $\eta$ is an internal node, then $x_{\eta}$
denotes the joint trait vector of all tips descended from node $\eta$;
and the trait vector associated node $\eta$ itself is denoted by
$z_{\eta}\in\mathbb{R}^{k_{\eta}}$. Therefore $x_{\eta}=z_{\eta}$
if $\eta$ is a tip, but $x_{\eta}$ contains more dimensions than
$z_{\eta}$ if $\eta$ is not a tip. The trait vectors $z_{\eta}$
for each node can have different number of dimensions. The multivariate
Gaussian density is denoted as $\phi(\cdot;\mu,\Sigma)$ and the characteristic
function of a set $A$ is denoted as $\setchar A$. We use the convention
that for any function $f$, the empty sum $\sum_{y\in\varnothing}f(y)=0$,
and that $\sum_{a\le j\le b}f(y_{j})=0$ if $b<a$; similarly product
over an empty set is one. The root node is denoted as $o$ hereinafter.

The trait vectors are assumed to evolve along the phylogenetic tree
according to the $\GLInv$ model discussed in \citet{mitov2020},
in which each species' traits evolve independently of each others
after branching off from their common ancestor and for every non-root
node $\eta\in\children(u)$, there are matrices $\left(\Phi_{\eta},w_{\eta},V_{\eta}\right)$
that are independent of $z_{m}$ for any node $m$ in the tree, but
can depend on $t_{\eta}$, $t_{u}$ and other parameters, such that
\[
p(z_{\eta}|z_{u})=\phi(z_{\eta};w_{\eta}+\Phi_{\eta}z_{u},V_{\eta}),
\]

where $\Phi_{\eta}\in\mathbb{R}^{k_{\eta}\times k_{u}}$, $w_{\eta}\in\mathbb{R}^{k_{\eta}}$
and the matrix $V_{\eta}\in\mathbb{R}^{k_{\eta}\times k_{\eta}}$
is symmetric positively definite. We let $\Theta_{\eta}$ be the set
of all possible values for $\left(\Phi_{\eta},w_{\eta},V_{\eta}\right)$
and let $\Theta=\times_{\eta\in\descendents(o)}\Theta_{\eta}$. Furthermore,
we define $f(\theta)$ be the log-likelihood of the entire model given
the root trait; in other words, $f(\theta)=\ln p(x_{o}|z_{o};\theta)$,
$\theta\in\Theta$. Through out the text, we also denote the $j$th
standard Euclidean space basis by $e_{j}$. 
\begin{figure}
\centering{}\includegraphics[width=1\textwidth]{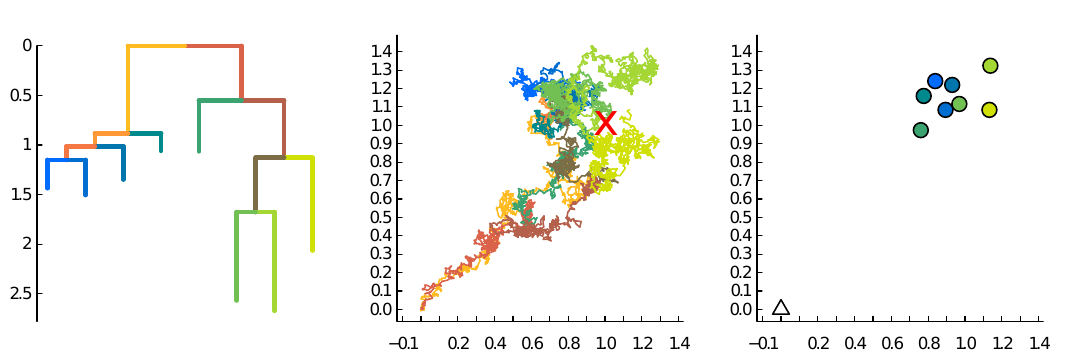}\caption{An illustration of the an multivariate OU model, which is a special
case of the $\protect\GLInv$ family. The leftmost graph shows a known
phylogeny; the middle is an realization of a 2-D OU processes along
this phylogeny, where the 'x' mark is the evolutionary optimum; the
rightmost shows what is observed, in which the triangle represent
the location of the root.}
\end{figure}

\subsection{Calculating the Likelihood and Gradient in Recursive Form}

The invariance of $\left(\Phi_{\eta},w_{\eta},V_{\eta}\right)$ with
respect to the trait vectors allows us to compute the likelihood function
in a post-order manner, and the gradient through an extra pre-order
traversal. Recently, the mentioned algorithms for the likelihood and
the gradient has been, independently of my effort, established in
\citet{bastide2020} but with a slightly different formulation and
perhaps a greater generality than in my version. Nonetheless, I will
still briefly show my formulation as a mean to lay a foundation for
discussing the second-order derivatives, as mine involves defining
some matrices and quantities that are important to the presentation
of the Hessian.

Note that for the computation of the likelihood function, the post-order
recursion shown here is very different from that of PCMBase discussed
in \citet{mitov2020}. While \citet{mitov2020} expands the quadratic
form in the Gaussian likelihood while traversing the tree, our method
uses of the Woodbury matrix inversion formula to avoid these expansion,
and more importantly, to book-keep some intermediate values for later
use in the gradient and Hessian computation.

Let $(\etaj)_{1\le j\le J_{u}}$ be all elements in $\children(u)$
and $v=\parent(u)$ with $p(z_{u}|z_{v})=\phi(z_{u};w_{u}+\Phi_{u}z_{v};V_{u})$,
where $V_{u}$ is symmetric positively definite. Also we let $(\kappa_{j})_{1\le j\le J_{o}}$
be all nodes in $\children(o)$, i.e., all direct children of the
root. We start by establishing the relationship between $p(x_{\etaj}|z_{u})$
and $p(x_{u}|z_{v})$.
\begin{prop}
\label{prop:prop1}If for some $(q_{\eta_{j}})_{1\le j\le J_{u}}$,
$(\Psi_{\eta_{j}})_{1\le j\le J_{u}}$, and $(\Sig_{\eta_{j}})_{1\le j\le J_{u}}$
the following holds
\[
p(x_{\eta_{j}}|z_{u})=\phi(x_{\eta_{j}};q_{\eta_{j}}+\Psi_{\eta_{j}}z_{u},\Sig_{\eta_{j}}),\quad\forall1\le j\le J_{u}
\]
then 
\begin{equation}
p(x_{u}|z_{v})=p\left(\xlr|z_{v}\right)=\phi(x_{u};q_{u}+\Psi_{u}z_{v};\Sig_{u})
\end{equation}
where 
\begin{align}
q_{u}= & \qlr+\Psilr w_{u},\label{eq:qu}\\
\Psi_{u} & =\Psilr\Phi_{u},\label{eq:psiu}\\
\Sig_{u}= & \Siglr+\Psilr V_{u}\PsilrT.\label{eq:qpsisig}
\end{align}
\end{prop}

The ``if'' part of the above proposition is trivially true if $\eta_{j}$
are all tips; in this case $\left(q_{\etaj},\Psi_{\etaj},\Sig_{\etaj}\right)=\left(w_{\etaj},\Phi_{\etaj},V_{\etaj}\right)$.
Note that the goal of our likelihood computation is to find $\left\{ \ln p(x_{\kappa}|z_{o})\right\} _{\kappa\in\children(o)}$,
because once these are found, the joint log-likelihood $\ln p(x_{o}|z_{o})$
is simply the sum of them by independence. Thus, if one let go of
the inefficiency of handling and inverting large matrices, the above
proposition has already enabled us to compute $p(x_{o}|z_{o})$ in
a bottom-up manner. Starting computing from $\left(q_{\eta},\Psi_{\eta},\Sig_{\eta}\right)$
for all tips $\eta$, one can successively ``merge'' clades using
\prettyref{eq:qu}-\eqref{eq:qpsisig}. Once one arrive at $\left(q_{\kappa},\Psi_{\kappa},\Sig_{\kappa}\right)$
for all $\kappa\in\children(o)$, i.e., the parameters associated
with direct children of the root, the likelihood $\prod_{\kappa\in\children(o)}p(x_{\kappa}|z_{o})$
is obviously obtained, although computing the Gaussian likelihood
this way requires one to invert $\Sig_{\kappa}$, which could be a
big matrix.

Rather than using the above naive scheme, one can avoid storing and
inverting $\Sig_{\kappa}$, using Woodbury matrix inversion update
\citep{max1950inverting} and a well-known direct result of Sylvester's
determinant theorem \citep{sylvester1883xxxix}, namely,
\begin{align*}
(A+UCV)^{-1} & =A^{-1}-A^{-1}U(C^{-1}+VA^{-1}U)^{-1}VA^{-1}\textrm{, and}\\
\det\left(A+UWV\right) & =\det(W^{-1}+VA^{-1}U)\det W\det A,
\end{align*}
 These matrix inversion formulae give rise to the following bottom-up
iteration scheme which allows an algorithm to only store matrices
of size $k_{\eta}\times k_{\eta}$ and $k_{\eta}\times k_{u}$ when
it visits a node $\eta\in\children(u)$.
\begin{prop}
For any node $n$, define 
\begin{align}
c_{n} & =(x_{n}-q_{n})^{T}\Sig_{n}^{-1}(x_{n}-q_{n});\label{eq:cu}\\
\gamma_{n} & =\Psi_{n}^{T}\Sig_{n}^{-1}(x_{n}-q_{n});\label{eq:gammau}\\
\Omega_{n} & =\Psi_{n}^{T}\Sig_{n}^{-1}\Psi_{n};\text{ and }\label{eq:Omegau}\\
\Delta_{n} & =\ln\det\Sig_{n}\label{eq:Deltau}
\end{align}
where all notations are as defined in \prettyref{prop:prop1}. It
follows that 
\begin{align}
c_{u} & =\sumc-2\sumgam^{T}w_{u}+w_{u}^{T}\left(\sumOmega\right)w_{u}-b_{u}^{T}\Lambda_{u}b_{u}\label{eq:curecur}\\
\gamma_{u} & =\Phi_{u}^{T}\left[I_{k_{u}}-\left(\sumOmega\right)\Lambda_{u}\right]b_{u}\label{eq:gammaurecur}\\
\Omega_{u} & =\Phi_{u}^{T}\left(\sumOmega\right)\left[I_{k_{u}}-\Lambda_{u}\left(\sumOmega\right)\right]\Phi_{u}\label{eq:Omegaurecur}\\
\Delta_{u} & =\sumDelta+\ln\det V_{u}+\ln\det\left(V_{u}^{-1}+\sumOmega\right)\label{eq:Deltaurecur}
\end{align}
where 
\begin{align}
\Lambda_{u} & =\left(V_{u}^{-1}+\sumOmega\right)^{-1}\label{eq:deflambda}\\
b_{u} & =\sumgam-\sumOmega w_{u}.\label{eq:defbu}
\end{align}
and $I_{k_{u}}$ is the $k_{u}\times k_{u}$ identity matrix. Furthermore,
\begin{equation}
\begin{aligned}\ln p(x_{o}|z_{o})= & -\frac{1}{2}\left[\sumjo c_{\kappa_{j}}+z_{o}^{T}\left(\sumjo\Omega_{\kappa_{j}}\right)z_{o}+\sumjo\Delta_{\kappa_{j}}\right]\\
 & \quad\quad\quad+\left(\sumjo\gamma_{\kappa_{j}}^{T}\right)z_{o}-\frac{k_{o}}{2}\ln(2\pi)
\end{aligned}
\label{eq:likmaster}
\end{equation}
where the nodes $\kappa_{j}\in\children(o)$.
\end{prop}

Using the above proposition, one can start with computing $(c,\gamma,\Omega,\Delta)$
for all tips and use \prettyref{eq:curecur}\textendash \eqref{eq:Deltaurecur}
to move upward until the reaching the direct children of the root.
Alternatively, this can also be done by a post-order traversal through
the tree, which is usually more convenient to implement if the tree
is stored in a linked-list-like data structure.

The byproducts $b_{\eta}$, $\Lambda_{\eta}$, $\Omega_{\eta}$, $\gamma_{\eta}$,
$c_{\eta}$, and $\Delta_{\eta}$ should be saved if the gradient
and Hessian is needed as they appears in both. Considering the gradient,
observe that in \prettyref{eq:likmaster}, the log-likelihood $\ln p(x_{o}|z_{o})$
comprises four elements that need to be differentiated: $\Omega_{\kappa}$,
$\gamma_{\kappa}$, $c_{\kappa}$, and $\Delta_{\kappa}$. In general,
for any node $\eta$, the derivatives $\frac{\partial\Omega_{\eta}}{\partial\theta_{\eta}}$,
$\frac{\partial\gamma_{\eta}}{\partial\theta_{\eta}}$, $\frac{\partial c_{\eta}}{\partial\theta_{\eta}}$,
and $\frac{\partial\Delta_{\eta}}{\partial\theta_{\eta}}$, where
$\theta_{\eta}$ means any element of one of $\Phi_{\eta}$, $w_{\eta}$,
and $\Sigma_{\eta}$, can be easily obtained by differentiating \prettyref{eq:curecur}\textendash \eqref{eq:Deltaurecur}
directly. Also, if the node $\eta_{2}$ is not a descendent of $\eta_{1}$
then $\frac{\partial\Omega_{\eta_{1}}}{\partial\theta_{\eta_{2}}}$,
$\frac{\partial\gamma_{\eta_{1}}}{\partial\theta_{\eta_{2}}}$, $\frac{\partial c_{\eta_{1}}}{\partial\theta_{\eta_{2}}}$,
and $\frac{\partial\Delta_{\eta_{1}}}{\partial\theta_{\eta_{2}}}$
are all zero because $\theta_{\eta_{2}}$ is absent in \prettyref{eq:curecur}\textendash \eqref{eq:Deltaurecur}.
For the case where $\eta_{2}$ is a descendent of $\eta_{1}$, the
following proposition provides a way to compute $\frac{\partial\Omega_{\eta_{1}}}{\partial\theta_{\eta_{2}}}$,
$\frac{\partial\gamma_{\eta_{1}}}{\partial\theta_{\eta_{2}}}$, $\frac{\partial c_{\eta_{1}}}{\partial\theta_{\eta_{2}}}$,
and $\frac{\partial\Delta_{\eta_{1}}}{\partial\theta_{\eta_{2}}}$
using $\frac{\partial\Omega_{\eta_{2}}}{\partial\theta_{\eta_{2}}}$,
$\frac{\partial\gamma_{\eta_{2}}}{\partial\theta_{\eta_{2}}}$, $\frac{\partial c_{\eta_{2}}}{\partial\theta_{\eta_{2}}}$,
and $\frac{\partial\Delta_{\eta_{2}}}{\partial\theta_{\eta_{2}}}$.
\begin{rem*}
By writing $\partial\theta_{\eta}$, when $\theta_{\eta}$ is an element
of the matrix $V_{\eta}$, we mean the derivatives that have already
taken into account the fact that $V_{\eta}$ is a symmetric matrix\textemdash more
formally, a directional derivative where two elements are simultaneously
perturbed unless $\theta_{\eta}$ is a diagonal element of $V_{\eta}$.
For example, when $V_{\eta}$ is $2\times2$ and $\theta_{\eta}$
is $V_{\eta}\left[1,2\right]$, by saying $\frac{\partial\Omega_{\eta}}{\partial\theta_{\eta}}$
we mean the directional derivative of $\Omega_{\eta}$ in the direction
of 
\[
\begin{bmatrix}0 & 1\\
1 & 0
\end{bmatrix}\textrm{;}
\]
simply put, it is just $\frac{\partial\Omega_{\eta}}{\partial V_{\eta}\left[1,2\right]}+\frac{\partial\Omega_{\eta}}{\partial V_{\eta}\left[2,1\right]}$.
However, when we don't write $\partial\theta_{\eta}$ (for example
when we write $\partial V_{\eta}$ or $\partial V_{ij}$ explicitly)
we mean simple partial derivatives in the usual sense. Readers should
also note that $V$ in fact contains only $k\left(k+1\right)/2$ many
free parameters in practice due to symmetricity. These ``non-symmetric
derivatives'' $\partial V_{\eta}$ are nonetheless convenient sometimes,
because they can be easily transformed into practical and symmetric
one while allowing various parameterisations to be used to ensure
symmetricity (and positive definiteness). As an example of the usefulness
of $\partial V_{\eta}$, when $\Omega$ is a function of symmetric
positive definite $V$ and we may parameterize $V$ using a lower-triangular
Cholesky decomposition $V=LL^{T}$, then for any $\ell\ge m$ we have
(proof in Section 10 of Supplementary Information)
\begin{align}
\frac{\partial\Omega}{\partial L_{\ell m}} & =\sum_{i=1}^{k}L_{im}\left(\frac{\partial\Omega}{\partial V_{\ell i}}+\frac{\partial\Omega}{\partial V_{i\ell}}\right);\label{eq:choleskychain-1}
\end{align}
while if there were more restrictions on $V$ than symmetric positive
definite then it is easy to reparametrise accordingly.

\end{rem*}
\begin{prop}
\label{prop:dchnmaster} Let $u=u_{1},u_{2},\ldots,u_{n}=\eta$ be
any chain of nodes along a branch of the phylogeny with $u_{k}=\parent(u_{k+1})$
for all $k$. Define 
\begin{align}
H_{u_{k}} & =I_{u_{k}}-\Lambda_{u_{k}}\sum_{\iota\in\children(u_{k})}\Omega_{\iota},\label{eq:Hdef}\\
a_{u_{k}} & =\Lambda_{u_{k}}b_{u_{k}}+w_{u_{k}}=H_{u_{k}}w_{u_{k}}+\Lambda_{u_{k}}\sum_{\iota\in\children(u_{k})}\gamma_{\iota},\label{eq:adef}
\end{align}
for all $k\le n-1$ and, for $k\le n$ we define, recursively, 
\begin{align}
F_{u_{k}}^{\eta} & =\begin{cases}
F_{u_{k+1}}^{\eta}H_{u_{k}}\Phi_{u_{k}}, & k\le n-1,\\
I, & k=n,
\end{cases}\label{eq:Fdef}\\
g_{u_{k}}^{\eta} & =\begin{cases}
g_{u_{k+1}}^{\eta}+F_{u_{k+1}}^{\eta}a_{u_{k}}, & k\le n-1,\\
0, & k=n,
\end{cases}\label{eq:gdef}\\
K_{u_{k}}^{\eta} & =\begin{cases}
K_{u_{k+1}}^{\eta}+F_{u_{k+1}}^{\eta}\Lambda_{u_{k}}F_{u_{k+1}}^{\eta T}, & k\le n-1,\\
0, & k=n.
\end{cases}\label{eq:kdef}
\end{align}
Then for $1\le k\le n$ it follows that 
\begin{align}
\frac{\partial\Omega_{u_{k}}}{\partial\theta_{\eta}} & =F_{u_{k}}^{\eta T}\frac{\partial\Omega_{\eta}}{\partial\theta_{\eta}}F_{u_{k}}^{\eta};\label{eq:dchnOmega}\\
\frac{\partial\gamma_{u_{k}}}{\partial\theta_{\eta}} & =F_{u_{k}}^{\eta T}\left(\frac{\partial\gamma_{\eta}}{\partial\theta_{\eta}}-\frac{\partial\Omega_{\eta}}{\partial\theta_{\eta}}g_{u_{k}}^{\eta}\right);\label{eq:dchngamma}\\
\frac{\partial c_{u_{k}}}{\partial\theta_{\eta}} & =\frac{\partial c_{\eta}}{\partial\theta_{\eta}}+g_{u_{k}}^{\eta T}\left(\frac{\partial\Omega_{\eta}}{\partial\theta_{\eta}}g_{u_{k}}^{\eta}-2\frac{\partial\gamma_{\eta}}{\partial\theta_{\eta}}\right);\text{ and }\label{eq:dchnc}\\
\frac{\partial{\Delta}_{u_{k}}}{\partial\theta_{\eta}} & =\frac{\partial\Delta_{\eta}}{\partial\theta_{\eta}}+\text{tr}\left(\frac{\partial\Omega_{\eta}}{\partial\theta_{\eta}}K_{u_{k}}^{\eta}\right).\label{eq:dchnDelta}
\end{align}
Furthermore, 
\begin{equation}
\begin{aligned}\frac{\partial\ln p(x_{o}|z_{o})}{\partial\theta_{\eta}}= & -\frac{1}{2}\left[\sumjo\frac{\partial c_{\kappa_{j}}}{\partial\theta_{\eta}}+z_{o}^{T}\left(\sumjo\frac{\partial\Omega_{\kappa_{j}}}{\partial\theta_{\eta}}\right)z_{o}+\sumjo\frac{\partial\Delta_{\kappa_{j}}}{\partial\theta_{\eta}}\right]\\
 & +\left(\sumjo\frac{\partial\gamma_{\kappa_{j}}}{\partial\theta_{\eta}}^{T}\right)z_{o},
\end{aligned}
\label{eq:dgradmaster}
\end{equation}
where, in each of the four sums in \ref{eq:dgradmaster} respectively,
only the term in which $\kappa_{j}$ is an ancestor of $\eta$ is
non-zero.
\end{prop}

\begin{prop}
\label{prop:hessselflist}The partial derivatives in the right-hand-side
of \prettyref{eq:dchnOmega}-\eqref{eq:dchnDelta} can be written
as
\begin{align}
\frac{\partial\Omega_{\eta}}{\partial\left(V_{\eta}\right)_{ij}} & =\Phi_{\eta}^{T}M_{\eta}G_{\eta}^{\left(ij\right)}M_{\eta}\Phi_{\eta};\label{eq:dodvself}\\
\frac{\partial\gamma_{\eta}}{\partial\left(V_{\eta}\right)_{ij}} & =\Phi_{\eta}^{T}M_{\eta}G_{\eta}^{\left(ij\right)}\xi_{\eta};\label{eq:dgdvself}\\
\frac{\partial c_{\eta}}{\partial\left(V_{\eta}\right)_{ij}} & =\xi_{\eta}^{T}G_{\eta}^{\left(ij\right)}\xi_{\eta};\label{eq:dcdvself}\\
\frac{\partial\Delta_{\eta}}{\partial V_{\eta}} & =D_{\eta};\label{eq:dddvself}\\
\frac{\partial\Omega_{\eta}}{\partial\left(\Phi_{\eta}\right)_{ij}} & =\Phi_{\eta}^{T}M_{\eta}Z_{\eta}e_{i}e_{j}^{T}+e_{j}e_{i}^{T}M_{\eta}Z_{\eta}\Phi_{\eta};\label{eq:dodphiself}\\
\frac{\partial\gamma_{\eta}}{\partial\left(\Phi_{\eta}\right)_{ij}} & =\left[\left(Z_{\eta}e_{i}\right)^{T}\xi_{\eta}\right]e_{j};\label{eq:dgdphiself}\\
\frac{\partial\gamma_{\eta}}{\partial\left(w_{\eta}\right)_{i}} & =-\Phi_{\eta}^{T}Z_{\eta}^{T}M_{\eta}e_{i};\label{eq:dgdwself}\\
\frac{\partial c_{\eta}}{\partial w_{\eta}} & =-2Z_{\eta}^{T}\xi_{\eta};\label{eq:dcdwself}\\
\frac{\partial c_{\eta}}{\partial\Phi_{\eta}} & =\frac{\partial\Delta_{\eta}}{\partial\Phi_{\eta}}=\frac{\partial\Omega_{\eta}}{\partial w}=\frac{\partial\Delta_{\eta}}{\partial w_{\eta}}=0,\label{eq:dxdxselfvanish}
\end{align}
where $\delta$ denotes the Kronecker delta, and
\begin{align}
G_{\eta}^{\left(ij\right)} & =\begin{cases}
-V_{\eta}^{-1}e_{i}e_{j}^{T}V_{\eta}^{-1}, & \textrm{\ensuremath{\eta} is a tip,}\\
-\Lambda_{\eta}V_{\eta}^{-1}e_{i}e_{j}^{T}V_{\eta}^{-1}\Lambda_{\eta}, & \textrm{otherwise;}
\end{cases}\\
M_{\eta} & =\begin{cases}
I, & \textrm{\ensuremath{\eta} is a tip,}\\
\sum_{\iota\in\children(\eta)}\Omega_{\iota}, & \textrm{otherwise;}
\end{cases}\\
\xi_{\eta} & =\begin{cases}
x_{\eta}-w_{\eta}, & \textrm{\ensuremath{\eta} is a tip,}\\
b_{\eta}, & \textrm{otherwise;}
\end{cases}\\
Z_{\eta} & =\begin{cases}
V_{\eta}^{-1}, & \textrm{\ensuremath{\eta} is a tip,}\\
H_{\eta}, & \textrm{otherwise;}
\end{cases}\\
D_{\eta} & =\begin{cases}
V_{\eta}^{-1}, & \textrm{\ensuremath{\eta} is a tip,}\\
V_{\eta}^{-1}-V_{\eta}^{-1}\Lambda_{\eta}V_{\eta}^{-1}, & \textrm{otherwise.}
\end{cases}
\end{align}
The partial derivatives that are not listed above are zero.
\end{prop}

Note that in \prettyref{eq:dgradmaster}, in order to compute the
gradient, we want $\frac{\partial\Omega_{\kappa_{j}}}{\partial\theta_{\eta}}$,
$\frac{\partial\gamma_{\kappa_{j}}}{\partial\theta_{\eta}}$, $\frac{\partial c_{\kappa_{j}}}{\partial\theta_{\eta}}$,
and $\frac{\partial\Delta_{\kappa_{j}}}{\partial\theta_{\eta}}$ for
all $\eta$. Using \prettyref{eq:Fdef}\textendash \eqref{eq:dchnDelta},
these partial derivatives can be computed by, starting with the trivial
$F_{\kappa_{j}}^{\kappa_{j}}$, $g_{\kappa_{j}}^{\kappa_{j}}$, and
$K_{\kappa_{j}}^{\kappa_{j}}$, recursively updating $\left(F_{\kappa_{j}}^{\eta},g_{\kappa_{j}}^{\eta},K_{\kappa_{j}}^{\eta}\right)$
and outputing $\left(\frac{\partial\Omega_{\kappa_{j}}}{\partial\theta_{\eta}},\frac{\partial\gamma_{\kappa_{j}}}{\partial\theta_{\eta}},\frac{\partial c_{\kappa_{j}}}{\partial\theta_{\eta}},\frac{\partial\Delta_{\kappa_{j}}}{\partial\theta_{\eta}}\right)$
while pre-order traversing all $\eta$ in the sub-tree rooted at $\kappa_{j}$.

\section{Hessian\label{sec:Hessian}}

\begin{figure}
\centering{}\includegraphics[width=0.48\textwidth]{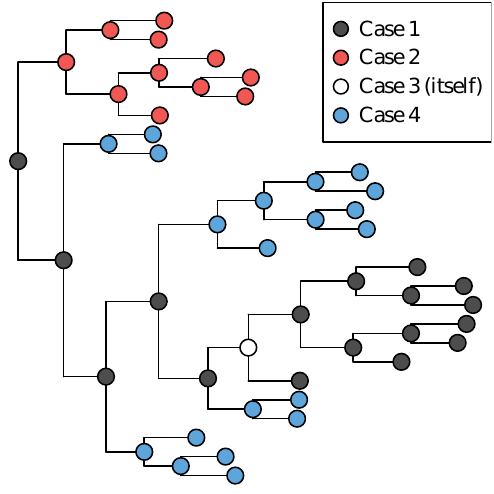}\caption{To examine the formulae for $\frac{\partial^{2}\ln p(x_{o}|z_{o})}{\partial\theta_{\eta_{1}}\partial\theta'_{\eta_{2}}}$,
the relationship between any pair of nodes $\eta_{1}$ and $\eta_{2}$
are classified into four cases. Each node in the example tree is painted
according to their relationship with the Case-3 node. The case for
the root node is in fact irrelavant because the root does not carry
any parameters. \label{fig:fourcases}}
\end{figure}

Now we turn our attention to the second-order partial derivatives
$\frac{\partial^{2}\ln p(x_{o}|z_{o})}{\partial\theta_{\eta_{1}}\partial\theta'_{\eta_{2}}}$.
We consider the following four cases seperately:
\begin{description}
\item [{Case~1}] $\eta_{1}$ is an ancestor of $\eta_{2}$;
\item [{Case~2}] the lowest common ancestor of $\eta_{1}$ and $\eta_{2}$
is the root;
\item [{Case~3}] $\eta_{1}$ and $\eta_{2}$ is the same node;
\item [{Case~4}] otherwise, that is, their lowest common ancestor is neither
the root, nor either of $\eta_{1}$ or $\eta_{2}$ themselves. These
four cases are illustrated in \prettyref{fig:fourcases}.
\end{description}
In Case 2, all the cross derivatives $\frac{\partial^{2}\ln p(x_{o}|z_{o})}{\partial\theta_{\eta_{1}}\partial\theta'_{\eta_{2}}}$
are trivially zero because the $\theta_{\eta_{2}}$ does not enter
into $\frac{\partial\ln p(x_{o}|z_{o})}{\partial\theta_{\eta_{1}}}$
at all, according to \prettyref{eq:dgradmaster} and the recursion
in \prettyref{eq:curecur}\textendash \eqref{eq:Deltaurecur}. This
is unsurprising as the trait at $\eta_{1}$ and $\eta_{2}$ should
be independent with each other if they fall into Case 2. Case 3 is
straightforward because $\left(F_{u_{j}}^{\eta},g_{u_{j}}^{\eta},K_{u_{j}}^{\eta}\right)$
is independent of $\theta_{\eta}$ by construction; therefore, taking
the derivatives of \prettyref{eq:dgradmaster} is only a matter of
finding the $\left(\frac{\partial^{2}\Omega_{\eta_{1}}}{\partial\theta_{\eta_{1}}\partial\theta'_{\eta_{1}}},\frac{\partial^{2}\gamma_{\eta_{1}}}{\partial\theta_{\eta_{1}}\partial\theta'_{\eta_{1}}},\frac{\partial^{2}c_{\eta_{1}}}{\partial\theta_{\eta_{1}}\partial\theta'_{\eta_{1}}},\frac{\partial^{2}\Delta_{\eta_{1}}}{\partial\theta_{\eta_{1}}\partial\theta'_{\eta_{1}}}\right)$
that appears in the right hand sides of \prettyref{eq:dchnOmega}\textendash \eqref{eq:dchnDelta}.
These second order derivatives are shown in the following.
\begin{prop}
(Case 3 scenario) The second derivatives $\left(\frac{\partial^{2}\Omega_{\eta}}{\partial\theta_{\eta}\partial\theta'_{\eta}},\frac{\partial^{2}\gamma_{\eta}}{\partial\theta_{\eta}\partial\theta'_{\eta}},\frac{\partial^{2}c_{\eta}}{\partial\theta_{\eta}\partial\theta'_{\eta}},\frac{\partial^{2}\Delta_{\eta}}{\partial\theta_{\eta}\partial\theta'_{\eta}}\right)$
are zero except that 
\begin{align}
\frac{\partial^{2}\Omega_{\eta}}{\partial\left(V_{\eta}\right)_{ij}\partial\left(V_{\eta}\right)_{pq}} & =\Phi_{\eta}^{T}M_{\eta}\tilde{G}_{\eta}^{\left(ijpq\right)}M_{\eta}\Phi_{\eta}\label{eq:hodvdv}\\
\frac{\partial^{2}\Omega_{\eta}}{\partial\left(V_{\eta}\right)_{ij}\partial\left(\Phi_{\eta}\right)_{pq}} & =e_{q}e_{p}^{T}M_{\eta}G_{\eta}^{\left(ij\right)}M_{\eta}\Phi_{\eta}+\Phi_{\eta}M_{\eta}G_{\eta}^{\left(ij\right)}M_{\eta}e_{p}e_{q}^{T}\label{eq:hodvdphi}\\
\frac{\partial^{2}\Omega_{\eta}}{\partial\left(\Phi_{\eta}\right)_{ij}\partial\left(\Phi_{\eta}\right)_{pq}} & =e_{q}e_{p}^{T}M_{\eta}Z_{\eta}e_{i}e_{j}^{T}+e_{j}e_{i}^{T}M_{\eta}Z_{\eta}e_{p}e_{q}^{T}\label{eq:hodphidphi}\\
\frac{\partial^{2}\gamma_{\eta}}{\partial\left(V_{\eta}\right)_{ij}\partial\left(V_{\eta}\right)_{pq}} & =\Phi_{\eta}^{T}M_{\eta}\tilde{G}_{\eta}^{\left(ijpq\right)}\xi_{\eta}\label{eq:hgamdvdv}\\
\frac{\partial^{2}\gamma_{\eta}}{\partial\left(V_{\eta}\right)_{ij}\partial\left(\Phi_{\eta}\right)_{pq}} & =e_{p}e_{q}^{T}M_{\eta}G_{\eta}^{\left(ij\right)}\xi_{\eta}\label{eq:hgamdvdphi}\\
\frac{\partial^{2}\gamma_{\eta}}{\partial\left(V_{\eta}\right)_{ij}\partial\left(w_{\eta}\right)_{p}} & =-\Phi_{\eta}M_{\eta}G_{\eta}^{\left(ij\right)}M_{\eta}e_{p}\label{eq:hgamdvdw}\\
\frac{\partial^{2}\gamma_{\eta}}{\partial\left(\Phi_{\eta}\right)_{ij}\partial\left(w_{\eta}\right)_{p}} & =-\left[\left(M_{\eta}Z_{\eta}e_{i}\right)^{T}e_{p}\right]e_{j}\label{eq:hgamdphidw}\\
\frac{\partial^{2}c_{\eta}}{\partial\left(V_{\eta}\right)_{ij}\partial\left(V_{\eta}\right)_{pq}} & =\xi_{\eta}^{T}\tilde{G}_{\eta}^{\left(ijpq\right)}\xi_{\eta}\label{eq:hcdvdv}\\
\frac{\partial^{2}c_{\eta}}{\partial\left(V_{\eta}\right)_{ij}\partial\left(w_{\eta}\right)_{p}} & =-e_{p}^{T}M_{\eta}G_{\eta}^{\left(ij\right)}\xi_{\eta}-\xi_{\eta}^{T}G_{\eta}^{\left(ij\right)}M_{\eta}e_{p}\label{eq:hcdvdw}\\
\frac{\partial^{2}\Delta_{\eta}}{\partial\left(V_{\eta}\right)_{ij}\partial\left(V_{\eta}\right)_{pq}} & =-e_{i}^{T}D_{\eta}e_{p}e_{q}^{T}D_{\eta}e_{j}\label{eq:hddvdv}
\end{align}
where 
\begin{align*}
\tilde{G}_{\eta}^{\left(ijpq\right)} & =\begin{cases}
-V_{\eta}^{-1}e_{p}e_{q}^{T}V_{\eta}^{-1}e_{i}e_{j}^{T}V_{\eta}^{-1}-V_{\eta}^{-1}e_{i}e_{j}^{T}V_{\eta}^{-1}e_{p}e_{q}^{T}V_{\eta}^{-1} & \textrm{\ensuremath{\eta} is tip;}\\
-\Lambda_{\eta}\left[V_{\eta}^{-1}e_{p}e_{q}^{T}D_{\eta}e_{i}e_{j}^{T}V_{\eta}^{-1}+V_{\eta}^{-1}e_{i}e_{j}^{T}D_{\eta}e_{p}e_{q}^{T}V_{\eta}^{-1}\right]\Lambda_{\eta} & \textrm{otherwise.}
\end{cases}
\end{align*}
\end{prop}

In Case 1 and 4, however, the second order derivative is recursive
because $\left(F_{u_{j}}^{\eta},g_{u_{j}}^{\eta},K_{u_{j}}^{\eta}\right)$
is recursive. This can be seen in the following proposition.
\begin{prop}
\label{prop:interm}In Case 1 and 4, letting $u_{1}\in\children(o)$
where $o$ is the root node, it follows that 
\begin{align}
\hrooteta{\Omega}= & \detaone{F_{u_{1}}^{\eta_{2}}}^{T}\detatwo{\Omega_{\eta_{2}}}F_{u_{1}}^{\eta_{2}}+{F_{u_{1}}^{\eta_{2}}}^{T}\detatwo{\Omega_{\eta_{2}}}\detaone{F_{u_{1}}^{\eta_{2}}};\label{eq:hodtdt}\\
\hrooteta{\gamma}= & \detaone{F_{u_{1}}^{\eta_{2}}}^{T}\left(\detatwo{\gamma_{\eta_{2}}}-\detatwo{\Omega_{\eta_{2}}}g_{u_{1}}^{\eta_{2}}\right)-{F_{u_{1}}^{\eta_{2}}}^{T}\detatwo{\Omega_{\eta_{2}}}\detaone{g_{u_{1}}^{\eta_{2}}};\label{eq:hgdtdt}\\
\hrooteta c= & \detaone{g_{u_{1}}^{\eta_{2}}}^{T}\left(\detatwo{\Omega_{\eta_{2}}}g_{u_{1}}^{\eta_{2}}-2\detatwo{\gamma_{\eta_{2}}}\right)+{g_{u_{1}}^{\eta_{2}}}^{T}\detatwo{\Omega_{\eta_{2}}}\detaone{g_{u_{1}}^{\eta_{2}}};\quad\text{and}\label{eq:hcdtdt}\\
\hrooteta{\Delta}= & \text{tr}\left(\detatwo{\Omega_{\eta_{2}}}\detaone{K_{u_{1}}^{\eta_{2}}}\right).\label{eq:hddtdt}
\end{align}
\end{prop}

Although the above is just a result of applying simple calculus product
rules to the $(c,\gamma,\Omega,\Delta)$-recursion in \prettyref{eq:dchnOmega}\textendash \eqref{eq:dchnDelta},
it informs us that the main obstacle in computing the second derivatives
in \prettyref{eq:dgradmaster} is $(\detaone{F_{u_{1}}^{\eta_{2}}},\detaone{g_{u_{1}}^{\eta_{2}}},\detaone{K_{u_{1}}^{\eta_{2}}})$.
Fortunately, the simple structure in the definition of $(F,g,K)$
allows us to derive the following interesting formulae for them.

We denote by $\lca{\eta_{1}}{\eta_{2}}$ the lowest common ancestor
of $\eta_{1}$ and $\eta_{2}$ exclusive of themselves. That is, $\lca{\eta_{1}}{\eta_{2}}$
is the lowest common ancestor in the normal sense except that if $\eta_{1}$
is an ancestor of $\eta_{2}$ then $\lca{\eta_{1}}{\eta_{2}}$ means
$\parent(\eta_{1})$, but not $\eta_{1}$.
\begin{prop}
Assume either Case 1 or 4. Let $u_{k}$ be a (potentially empty) lineage
in which $u_{k}=\parent(u_{k+1})$ for all $1\le k\le K$, with $u_{K}=\lca{\eta_{1}}{\eta_{2}}$
and $u_{1}$ being the direct child of the root that leads to $u_{K}$
. If the lineage were empty, we let $K=0$. In Case 1, let $u_{K+1}=\eta_{1}$
and $u_{K+2}$ be the child of $\eta_{1}$ that leads to $\eta_{2}$.
It follows that 
\begin{align}
\detaone{F_{u_{1}}^{\eta_{2}}}= & \left[-\left(\sum_{1\le k\le K}F_{u_{k+1}}^{\eta_{2}}\Lambda_{u_{k}}{F_{u_{k+1}}^{\eta_{1}}}^{T}\right)\detaone{\Omega_{\eta_{1}}}+\setchar{\descendents(\eta_{2})}(\eta_{1})\detaone{F_{\eta_{1}}^{\eta_{2}}}\right]F_{u_{1}}^{\eta_{1}};\label{eq:dfdtheta}\\
\begin{split}\detaone{g_{u_{1}}^{\eta_{2}}}= & \!\!\sum_{1\le k\le K-1}\left[\setchar{\descendents(\eta_{2})}(\eta_{1})\detaone{F_{\eta_{1}}^{\eta_{2}}}F_{u_{k+1}}^{\eta_{1}}-\!\!\!\!\sum_{k+1\le\ell\le K}\left(F_{u_{\ell+1}}^{\eta_{2}}\Lambda_{u_{\ell}}{F_{u_{\ell+1}}^{\eta_{1}}}^{T}\right)\detaone{\Omega_{\eta_{1}}}F_{u_{k+1}}^{\eta_{1}}\right]a_{u_{k}}\\
 & +\sum_{1\le k\le K}F_{u_{k+1}}^{\eta_{2}}\Lambda_{u_{k}}{F_{u_{k+1}}^{\eta_{1}}}^{T}\left(\detaone{\gamma_{\eta_{1}}}-\detaone{\Omega_{\eta_{1}}}g_{u_{k}}^{\eta_{1}}\right)\\
 & +\setchar{\descendents(\eta_{2})}(\eta_{1})\left(\detaone{F_{\eta_{1}}^{\eta_{2}}}a_{\parent\left(\eta_{1}\right)}+F_{u_{K+2}}^{\eta_{2}}\detaone{a_{\eta_{1}}}\right);\text{ and }
\end{split}
\label{eq:dgdtheta}\\
\detaone{K_{u_{1}}^{\eta_{2}}}= & \sum_{1\le k\le K-1}\left(A_{k}+{A_{k}}^{T}\right)+\sum_{1\le k\le K}B_{k}+\setchar{\descendents(\eta_{2})}(\eta_{1})\left(A_{K}+{A_{K}}^{T}+{B_{K+1}}\right),\label{eq:dkdtheta}
\end{align}
where
\begin{align}
A_{k}= & \left[\setchar{\descendents(\eta_{2})}(\eta_{1})\detaone{F_{\eta_{1}}^{\eta_{2}}}-\sum_{k+1\le\ell\le K}F_{u_{\ell+1}}^{\eta_{2}}\Lambda_{u_{\ell}}{F_{u_{\ell+1}}^{\eta_{1}}}^{T}\detaone{\Omega_{\eta_{1}}}\right]F_{u_{k+1}}^{\eta_{1}}\Lambda_{u_{k}}{F_{u_{k+1}}^{\eta_{2}}}^{T};\\
B_{k}= & -F_{u_{k+1}}^{\eta_{2}}\Lambda_{u_{k}}{F_{u_{k+1}}^{\eta_{1}}}^{T}\detaone{\Omega_{\eta_{1}}}F_{u_{k+1}}^{\eta_{1}}\Lambda_{u_{k}}{F_{u_{k+1}}^{\eta_{2}}}^{T},
\end{align}
and
\begin{align}
\frac{\partial F_{\eta_{1}}^{\eta_{2}}}{\partial\left(\Phi_{\eta_{1}}\right)_{pq}}= & F_{u_{K+2}}^{\eta_{2}}H_{\eta_{1}}e_{p}e_{q}^{T};\\
\frac{\partial F_{\eta_{1}}^{\eta_{2}}}{\partial\left(V_{\eta_{1}}\right)_{pq}}= & -F_{u_{K+2}}^{\eta_{2}}\Lambda_{\eta_{1}}{V_{\eta_{1}}}^{-1}e_{p}e_{q}^{T}{V_{\eta_{1}}}^{-1}\Lambda_{\eta_{1}}\sum_{\iota\in\children(\eta_{1})}\Omega_{\iota}\Phi_{\eta_{1}};\\
\frac{\partial a_{\eta_{1}}}{\partial\left(w_{\eta_{1}}\right)_{p}}= & H_{\eta_{1}}e_{p};\\
\frac{\partial a_{\eta_{1}}}{\partial\left(V_{\eta_{1}}\right)_{pq}}= & \Lambda_{\eta_{1}}V_{\eta_{1}}^{-1}e_{p}e_{q}^{T}V_{\eta_{1}}^{-1}\left(a_{\eta_{1}}-w_{\eta_{1}}\right);\\
\frac{\partial F_{\eta_{1}}^{\eta_{2}}}{\partial w_{\eta_{1}}}= & \frac{\partial a_{\eta_{1}}}{\partial\Phi_{\eta_{1}}}=\frac{\partial a_{\eta_{1}}}{\partial\Phi_{\eta_{1}}}=0;.
\end{align}
\end{prop}

It is important to see that, in \prettyref{eq:dfdtheta}, we can obtain
the partial derivative of any child of $\eta_{2}$ by pre-multiplying
$H_{\eta_{2}}\Phi_{\eta_{2}}$. This immediately leads to the following
result.
\begin{cor}
\label{cor:dfrecur}Within Case 1 and 4, for any $v\in\children(\eta_{2})$,
it follows that
\begin{align}
\detaone{F_{u_{1}}^{v}}= & H_{\eta_{2}}\Phi_{\eta_{2}}\detaone{F_{u_{1}}^{\eta_{2}}};\label{eq:dFrecurdown}\\
\detaone{g_{u_{1}}^{v}}= & H_{\eta_{2}}\Phi_{\eta_{2}}\detaone{g_{u_{1}}^{\eta_{2}}};\textrm{ and}\label{eq:dgrecurdown}\\
\detaone{K_{u_{1}}^{v}}= & H_{\eta_{2}}\Phi_{\eta_{2}}\detaone{K_{u_{1}}^{\eta_{2}}}{\Phi_{\eta_{2}}}^{T}{H_{\eta_{2}}}^{T}.\label{eq:dKrecurdown}
\end{align}
\end{cor}

These simple recursion allows us to compute the second derivatives
of the log-likelihood recursively. Recall that in order to compute
$\heta f=\heta{\ln p(x_{o}|z_{o})}$, we need $\hrooteta{\Omega}$,
$\hrooteta{\gamma}$, $\hrooteta c$, and $\hrooteta{\Delta}$, which
in turn depends on $\detaone{F_{u_{1}}^{\eta_{2}}}$, $\detaone{g_{u_{1}}^{\eta_{2}}}$,
and $\detaone{K_{u_{1}}^{\eta_{2}}}$. Computing $\detaone{F_{u_{1}}^{\eta_{2}}}$,
$\detaone{g_{u_{1}}^{\eta_{2}}}$, and $\detaone{K_{u_{1}}^{\eta_{2}}}$
for all internal nodes $\left(\eta_{1},\eta_{2}\right)$ needs a quadratic-time,
becaues from \prettyref{eq:dFrecurdown}\textendash \eqref{eq:dKrecurdown},
it is apparent that we can fix any $\eta_{1}$ and compute all $\detaone{F_{u_{1}}^{v}}$,
$\detaone{g_{u_{1}}^{v}}$, and $\detaone{K_{u_{1}}^{v}}$ for all
$v\in\children(o)$ in one pre-order tree walk (recall that $o$ is
a direct child of the root), and then one would need to loop though
all possible $\eta_{1}$. This also implies that computing all $\heta f$
would be quadratic time with respect to the number of internal nodes.
In practice, the computation can be speeded up considerably with careful
programming, for example, by storing neccessariy matrices such as
$H_{\eta_{1}}$ and parallelize this computation in multiple threads.

As a side note, despite this matrix is big, if one is only interested
in inferring about the parameters of a sub-model such as the Ornstein-Uhlenbeck
or Brownian motion model, it is unnecessary to store the entire matrix
of these $\heta f$; nor do we usually desire to store them for its
$O(n^{2}(2p^{2}+p)^{2})$ complexity, where $n$ is the number of
tips and $p$ the number of trait dimensions. To see this, let $g=(g^{(1)},\cdots,g^{(\dim\text{\ensuremath{\Theta}})})$
be the function that maps from OU model's parameter space into $\Theta$,
and let $\mathbf{H}$ and $\mathbf{J}$ be the Hessian and Jacobian
operator respectively. By the calculus chain rule it follows that
\begin{equation}
\mathbf{H}(f\circ g)=(\mathbf{J}g)^{T}[(\mathbf{H}f)\circ g](\mathbf{J}g)+\sum_{j=1}^{\dim\Theta}\left(\mathbf{J}f_{j}\right)\mathbf{H}g^{(j)},\label{eq:chainrule_simp}
\end{equation}
where $f_{j}$ denotes the restriction of $f$ to its $j$th coordinate.
In an actual implementation, because 
\begin{equation}
(\mathbf{J}g)^{T}[(\mathbf{H}f)\circ g](\mathbf{J}g)=\sum_{k,\ell}[(\mathbf{H}f)\circ g]_{k\ell}\textrm{row}_{k}(\mathbf{J}g){}^{T}\textrm{row}_{\ell}(\mathbf{J}g)
\end{equation}
 holds, to compute the first term above, one only need to perform
an on-the-fly rank-one matrix update everytime an element of $(\mathbf{H}f)\circ g$
is outputted from a tree walk, therefore avoiding the need to store
the entire $(\mathbf{H}f)\circ g$ in memory. For the second term,
we should not have any storage problems at all as long as our memory
is large enough to store $\mathbf{J}f$, which has only a $O(n^{2}p^{2})$
complexity, because the size of $\mathbf{H}g^{(j)}$ is typically
constant in $n$, such as in the OU model.

\section{The Ornstein-Uhlenbeck Model\label{sec:The-Ornstein-Uhlenbeck-Model}}

In the glinvci package, the Ornstein-Uhlenbeck model and Brownian
motion model's Hessian matrix is implemented using the calculus chain
rule in \prettyref{eq:chainrule_simp}. First, the Jacobian $\mathbf{J}g$
is computed and stored in the memory; then the algorithm iterate through
the phylogenetic tree and at the same time update the first term in
\prettyref{eq:chainrule_simp}; then the second term is computed and
added to the final result.

Consider the SDE governing the OU process 
\begin{equation}
\mathrm{d}z(t)=-H[z(t)-\mu\mathrm{]d}t+\Sigx\mathrm{d}W(t)\label{eq:ousde}
\end{equation}
where $\Sigx$ is a lower-triangular Cholesky factor. Using the same
notation as previous section, in this OU model, the function $g$
can be written as
\begin{align}
V_{\eta} & =\intop_{0}^{t_{\eta}}e^{-Hv}LL^{T}e^{-H^{T}v}\mathrm{d}v,\label{eq:VOU}\\
w_{\eta} & =\left(I-e^{-Ht_{\eta}}\right)\mu,\label{eq:wOU}\\
\Phi_{\eta} & =e^{-Ht_{\eta}}.\label{eq:PhiOU}
\end{align}

By differentiating the above expression directly, with the help of
the matrix exponential spectral formulae in \citet{najfeld_derivatives_1995},
we can obtain the explicit formulae for $\mathbf{J}g$ and $\mathbf{H}g^{(j)}$.
As the derivation of these formulae is quite lengthy, here we only
present the result but leave the tedious derivation in the supplimentary
material owing to space limitations and readibility. For convenience,
we define $\symadd A=A+A^{T}$ for any square matrix $A$, and we
let $U^{ij}=e_{i}e_{j}^{T}$.

Additionally, we define the following univariate definite integrals
that appear in the derivatives of the Ornstein-Uhlenbeck model:

\begin{align}
I_{n}\left(a,t\right) & =\intop_{0}^{t}y^{n}e^{ay}\textrm{d}y,\label{eq:intin}\\
K_{n}\left(a,b,t\right) & =\intop_{0}^{t}e^{ay}I_{n}\left(b,y\right)\textrm{d}y,\label{eq:intkn}\\
Q\left(a,b,c,t\right) & =\intop_{0}^{t}e^{ay}I_{0}\left(b,y\right)I_{0}\left(c,y\right)\textrm{d}y.\label{eq:intq}
\end{align}

As we do not restrict $a$, $b$, and $c$ to be non-zero, the integrands
may take a different form if any of them were zero (for example if
$a=0$ then $I_{0}$ becomes a polynomial integral); hence when integrating
them these cases must be seperately considered. But aside from that,
they are straightforward to integrate with standard calculus especially
for small $n$. Only $I_{0}$, $I_{1}$, $I_{2}$, $K_{0}$, $K_{1}$
and $Q$ are used hereinafter. Readers can find the explicit expressions
of these in Section 11 of the Supplementary Material.
\begin{prop}
\label{prop:oujac}If an eigen decomposition $H=P\textrm{diag}\left(\lambda\right)P^{-1}$
exists, all elements of $\mathbf{J}g$ are zero except that
\begin{align}
\detat V{H_{ij}} & =-\symadd{PMP^{T}}\textrm{,}\label{eq:dvdhij}\\
\detat V{\Sigx_{ij}} & =\evalat{V_{\eta}}{LL^{T}=\symadd{U^{ij}L^{T}}},\label{eq:dvdlij}\\
\detat w{H_{ij}} & =-\detat{\Phi}{H_{ij}}\mu\textrm{,}\label{eq:dwdhij}\\
\detat w{\mu_{i}} & =(\mathbf{I}-e^{-Ht_{n}})e_{i}\textrm{, and}\label{eq:dwdmui}\\
\detat{\Phi}{H_{ij}} & =-P\left\{ \overline{U^{ij}}\odot\left[e^{-\lambda_{i}t_{\eta}}I_{0}(\lambda_{i}-\lambda_{j},t_{\eta})\right]_{ij}\right\} P^{-1}\textrm{,}\label{eq:dphidhij}
\end{align}

where
\begin{align}
M_{mn} & =\sum_{k}\overline{U_{mk}^{ij}}\overline{\Sigma}_{kn}\big\{\left[\lambda_{m}=\lambda_{k}\right]I_{1}(-\lambda_{m}-\lambda_{n},t_{\eta})\nonumber \\
 & \quad\quad\quad\quad\quad+\left[\lambda_{m}\neq\lambda_{k}\right]\left[I_{0}(-\lambda_{k}-\lambda_{n},t_{\eta})-I_{0}(-\lambda_{m}-\lambda_{n},t_{\eta})\right]\big\}\textrm{,}\\
\overline{U^{ij}} & =P^{-1}U^{ij}P\textrm{,}\\
\overline{\Sigma} & =P^{-1}\left(\Sigx\Sigx^{T}\right)P^{-T}\textrm{,}
\end{align}

with $\left[c=0\right]$ being the Iversion bracket, $\mathbf{I}$
denoting the identity matrix, and a convention that the product of
zero and a undefined value equals to zero.
\end{prop}

Again, after directly differentiating the first derivatives shown
above, with some lengthy calculations, we arrive at the expressions
below.\begingroup\allowdisplaybreaks
\begin{prop}
\label{prop:ouhess}When $H$ has an eigen factorization $H=P\textrm{diag}\left(\lambda\right)P^{-1}$,
all second-order partial derivatives of $g$ are zero except that
\begin{align}
\hgen V{H_{mn}}{H_{ij}} & =\sum_{\alpha,\beta}\sum_{k,\ell}\left[P^{-1}U^{mn}P\right]_{\alpha\beta}\left[P^{-1}U^{ij}P\right]_{k\ell}P\mathcal{S}\Big(\nonumber \\
 & \quad\quad\quad\left[k=\beta\right]\left[\left[\gamma=\alpha\right]\mathring{L}_{k\ell\alpha\iota}^{\eta}\overline{\Sigma}_{\ell\iota}\right]_{\gamma\iota}-\left[\ell=\alpha\right]\left[\left[\gamma=k\right]\mathring{R}_{k\ell\beta\iota}^{\eta}\overline{\Sigma}_{\beta\iota}\right]_{\gamma\iota}\nonumber \\
 & \quad\quad+\left[Q\left(-\lambda_{\gamma}\!-\!\lambda_{\iota},\lambda_{\alpha}\!-\!\lambda_{\beta},\lambda_{k}\!-\!\lambda_{\ell},t_{\eta}\right)\right]_{\gamma\iota}\odot\left(U^{\alpha\beta}U^{k\ell}\overline{\Sigma}+U^{\alpha\beta}\overline{\Sigma}U^{k\ell}\right)\Big)P^{T},\label{eq:hvdhdh}\\
\hgen V{L_{mn}}{H_{ij}} & =\text{\ensuremath{\evalat{\detat V{H_{ij}}}{LL^{T}=\symadd{U^{mn}L^{T}}}}},\label{eq:hvdldh}\\
\hgen V{L_{mn}}{L_{ij}} & =\evalat{V_{\eta}}{\symadd{U^{ij}U^{nm}}},\label{eq:hvdldl}\\
\hgen w{H_{mn}}{H_{ij}} & =-\left(\hgen{\Phi}{H_{mn}}{H_{ij}}\right)\mu,\label{eq:hwdhdh}\\
\hgen w{\mu_{m}}{H_{ij}} & =-\left(\detat{\Phi}{H_{ij}}\right)e_{m}\textrm{, and}\label{eq:hwdmudh}\\
\hgen{\Phi}{H_{mn}}{H_{ij}} & =\sum_{\alpha,\beta}\sum_{k,\ell}\left[P^{-1}U^{mn}P\right]_{\alpha\beta}\left[P^{-1}U^{ij}P\right]_{k\ell}P\Big\{\nonumber \\
 & \textrm{\quad\quad\ensuremath{\left[k=\beta\right]\Big[e^{-\lambda_{\alpha}t_{\eta}}I_{0}}(\ensuremath{\lambda_{\alpha}}-\ensuremath{\lambda_{\beta}},\ensuremath{t_{\eta}})\ensuremath{I_{0}}(\ensuremath{\lambda_{k}}-\ensuremath{\lambda_{\ell}},\ensuremath{t_{\eta}})}\nonumber \\
 & \quad\quad\quad\quad\quad\quad\quad-e^{-\lambda_{\alpha}t_{\eta}}K_{0}\left(\lambda_{\alpha}-\lambda_{\ell},\lambda_{k}-\lambda_{\alpha},t_{\eta}\right)\Big]U^{\alpha\ell}\nonumber \\
 & \quad\quad\quad+[\ell=\alpha]\big[e^{-\lambda_{k}t_{\eta}}K_{0}\left(\lambda_{k}-\lambda_{\beta},\lambda_{\beta}-\lambda_{\ell},t_{\eta}\right)U^{k\beta}\Big]\Big\} P^{-1}\textrm{,}\label{eq:hphidhdh}
\end{align}

where
\begin{align}
\mathring{L}_{k\ell\alpha\iota}^{\eta} & =\left[\lambda_{k}=\lambda_{\alpha}\right]K_{1}\left(-\lambda_{\alpha}-\lambda_{\iota},\lambda_{\alpha}-\lambda_{\ell},t_{\eta}\right)\nonumber \\
 & \quad+\left[\lambda_{k}\neq\lambda_{\alpha}\right]\frac{K_{0}\left(-\lambda_{\alpha}-\lambda_{\iota},\lambda_{k}-\lambda_{\ell},t_{\eta}\right)-K_{0}\left(-\lambda_{\alpha}-\lambda_{\iota},\lambda_{\alpha}-\lambda_{\ell},t_{\eta}\right)}{\lambda_{k}-\lambda_{\alpha}}\mathrm{\,\,and}\\
\mathring{R}_{k\ell\beta\iota}^{\eta} & =\left[\lambda_{\beta}=\lambda_{\ell}\right]K_{1}\left(-\lambda_{k}-\lambda_{\iota},\lambda_{k}-\lambda_{\beta},t_{\eta}\right)\nonumber \\
 & \quad+\left[\lambda_{\beta}\neq\lambda_{\ell}\right]\frac{K_{0}\left(-\lambda_{k}-\lambda_{\iota},\lambda_{k}-\lambda_{\ell},t_{\eta}\right)-K_{0}\left(-\lambda_{k}-\lambda_{\iota},\lambda_{k}-\lambda_{\beta},t_{\eta}\right)}{\lambda_{\beta}-\lambda_{\ell}}\textrm{. }
\end{align}
\end{prop}

\endgroup

Although the results in \prettyref{prop:oujac} and \prettyref{prop:ouhess}
assume an OU process, it is actually also applicable to the Brownian
motion because the latter is just a special case of the former with
$H$ being the zero matrix. When $H$ is zero, the entire likelihood
surface is obviously invariant to $\theta$, so one can simply ignore
$\theta$ when calculating the Hessians of the log-likelihood. More
concretely, to obtain the derivatives of the Brownian motion, one
only need to evaluate the above formulae at $H_{ij}=0$ for all $i$
and $j$, and discard (or set to zero, in other words) the irrelavant
derivatives that are taken with respect to $H$ and $\mu$.

In the actual implementation in the glinvci package, the parameterization
is, in fact, slightly different: the diagonals of $L$, in the implementation,
are parameterized in its log-scale instead. This, as \citet{pinheiro_unconstrained_1996}
has pointed out, restricts the diagonals to be always positive, and
hence avoiding multiple optima in the log-likelihood caused by the
signs of the diagonals, as well as enabling the use of unrestricted
optimization when estimating the model. The Jacobian and Hessians
of the log-scale diagonal's case are simply a matter of applying a
chain rule to the formulae in \prettyref{prop:oujac} and \prettyref{prop:ouhess}
for $L_{ii}$ for all $i$.

\section{Handling of Missing Data and Lost Traits\label{sec:Handling-of-Missing}}

\begin{figure}

\begin{centering}
\includegraphics[width=0.48\textwidth]{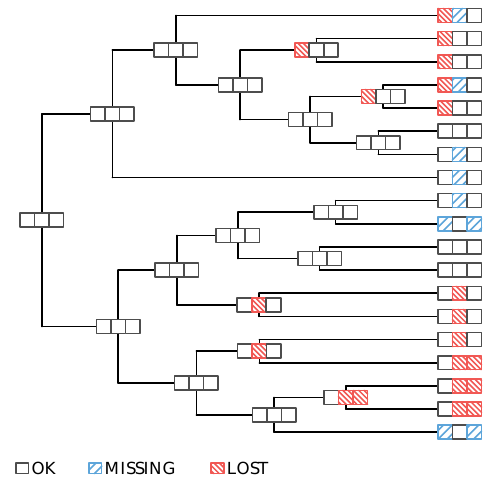}\caption{An example of missingness and existence of traits in a data set with
3-dimensional trait vector. The missingness status of the tips is
set according to the NAs and NaNs in the data, and that of the internal
nodes is inferred from the tip.}
\par\end{centering}
\end{figure}
Although the objective of this work is to present the second-order
derivatives, it is worth noting that in this formulation, the dimensionality
of the traits at different nodes can differ, and this allows the handling
of situations in which, for example, some observations are missing.

There are multiple ways that missing data can occur in PCM and they
must be modelled differently. For example, if a missing trait observation
were due to data collection problem, it should be marginalised out
in the Gaussian likelihood. But if the trait were lost completely
during the evolutionary process in a sense that it has existed at
the root of the tree but has ceased to exist at some time, simply
marginalising them would be errornous because it has to do with the
underlying evolutionary process.

The \texttt{glinvci} package has the capability to handles missing
observation, which occurs due to data collection problems, and non-existing
(lost) traits, which refers to trait dimensions which have existed
at the root of the tree but had ceased to exist during the evolutionary
process. The way \texttt{glinvci} handles missing and lost traits
is similar to PCMBase, but \texttt{glinvci} supports an additional
alternative option to control how the lost traits are handled in the
likelihood function. Similar to \texttt{PCMBase} \citet{mitov2020},
\texttt{glinvci} encodes the former situation with NA and the latter
NaN. \texttt{Glinvci} uses the follow rule to assign ``missingness
tags'' to each node. As written in my package's documentation, each
trait vector dimension $p$ of each node $\eta$ is tagged as one
of ``missing'', ``lost'' or ``available''. If the $p$th dimension
of a tip contains an NA, we tag it as ``missing'', with all other
tags in the tree left unchanged.

Thus, ``missing'' tags can only appear in the tips. The $p$th dimension
of any node $\eta$, regardless of whether or not it is an internal
node or a tip, is tagged ``lost'' if and only if the $p$th dimension
of all tips inside the clade started at $\eta$ are NaN. In other
words, this $p$th dimension is assumed to be existent and evolving
as usual in all internal nodes $\eta$ unless the dimension exists
in none of the tips descended from $\eta$. These assumption on the
traits' existence is in fact consistent with that of PCMBase.

By default, \texttt{glinvci} handles missing data the same way as
PCMBase: it marginalizes all missing dimensions at the tips of the
tree by removing the corresponding dimensions in the Gaussian distribution's
means and covariance matrices (that is, dropping the missing dimensions
in $(\Phi_{\eta},w_{\eta},V_{\eta})$), while assuming that the underlying
evolution remains unchanged. This method of handling missing measurement
is invariant to whether the user is using the OU, Brownian motion,
or other models. For non-existing traits, however, \texttt{glinvci}
provides two pre-defined methods, one of which is identical to what
PCMBase does: when a trait changes from existent at node $\eta$ to
nonexistent at node $\eta^{*}\in\children(\eta)$, the Gaussian distribution
is marginalized by dropping the lost dimension in the exact same way
as in the case of missing measurements.

The second pre-defined method for handling nonexistent traits is specific
to the OU models. In this method, we assume that, starting from the
time $t^{*}$ at which the branch leading to a node $\eta^{*}$ begins:
(i) the lost trait dimension does not evolve at all; (ii) the difference
between the lost $p$th trait at $t^{*}$ and its optimum has stopped
influencing the evolutions of any other dimensions via any means,
in particular, neither through the linear combination with the drift
matrix $H$ nor through the Wiener process covariance. (iii). the
disappearence of the traits does not alter how other existing traits
interact. From these desiderata, letting $\mathcal{P}=\left\{ p|\text{the \ensuremath{p}th dimension is lost in \ensuremath{\eta^{*}}}\right\} $,
we seek a modified OU process according to which the trait vector
should evolve since $t^{*}$, and whose parameters $(H',\mu',\Sigx')$
satisfy
\begin{align}
H'e_{p}= & 0 & \text{for \ensuremath{p\in\mathcal{P}}},\label{eq:lostconstr1}\\
H'e_{j}= & \sum_{k\notin\mathcal{P}}(He_{j})_{k}e_{k} & \text{for \ensuremath{j\notin\mathcal{P}}},\label{eq:lostconstr2}\\
\left(\Sigx'\Sigx'^{T}\right)_{jk}= & \left(\Sigx\Sigx^{T}\right)_{jk} & j\notin\mathcal{P}\text{ and }k\notin\mathcal{P}\textrm{,}\label{eq:lostconstr3}\\
\left(\Sigx'\Sigx'^{T}\right)_{jk}= & 0 & j\in\mathcal{P}\text{ or }k\in\mathcal{P}\textrm{, and}\label{eq:lostconstr4}\\
\mu'= & \mu.\label{eq:lostconstr5}
\end{align}
The first equation above means that the lost dimensions of $z(t)-\mu$
does not affect $\mathrm{d}z(t)$. The second implies that other dimensions
drift toward the optimum ``as usual'', except that these existent
dimensions' influence to the lost dimensions of $\mathrm{d}z(t)$
is zeroed out. The third and the fourth equations imply that the covariance
between the existent dimensions are not changed by the trait loss,
that the lost dimension is not correlated to any existing traits anymore,
and that the lost dimension has zero variance. It is a matter of simple
linear algebra to verify that $(H',\Sigx')$ will satisfy the above
four equations if we set $H'$ to be identical to $H$ except its
$p$th rows and columns are zero for all $p\in\mathcal{P}$; and $\Sigx'$
to be the same as $\Sigx$ except that its $p$th row is zero for
all $p\in\mathcal{P}$. The SDE is invariant of the lost dimensions
of $\mu$ because of \prettyref{eq:lostconstr1}. Using the new $(H',\mu',\Sigx')$,
\texttt{glinvci} computes $(\Phi,w,V)$ and marginalizes the lost
traits by removing the lost dimensions in $(\Phi,w,V)$.

\section{Discussion\label{sec:Discussion}}

Using the Hessian matrix $\hat{\mathcal{I}}^{-1}$ of the likelihood,
one might be interested in the $\alpha$-level asymptotic confidence
ellipse
\begin{align}
\left(\hat{\theta}-\theta\right)^{T}\hat{\mathcal{I}}^{-1}\left(\hat{\theta}-\theta\right) & \le\text{\ensuremath{\chi_{p}^{2}\left(1-\alpha\right)}}\label{eq:confellipse-1}
\end{align}
where $p$ is the number of parameters. I have set up an simple numerical
experiment to give an idea of how many tips are needed for this to
have a good coverage probability. In the experiment, for each number
of trait dimensions and each number of tips, a tree and a trait data
matrix, is, together as a pair, repeatedly simulated, with the former
generated using the ``rtree'' function taken from the R package
``ape'' \citep{apepack} and the latter according to the OU model
defined in \prettyref{eq:ousde}. The true OU parameters are the same
across repetitions, being fixed to
\begin{align}
H & =\textrm{diag}\left(\left[1-1/10,\ldots,1-p/10\right]\right),\label{eq:fixexper1}\\
\mu & =\left[-7/8,\ldots,-7/8\right]\textrm{, and}\label{eq:fixexper2}\\
LL^{T} & =\textrm{diag}\left(\left[1/2,\ldots,1/2\right]\right).\label{eq:fixexper3}
\end{align}
Then, for each simulated data set we check whether or not the confidence
ellipse defined in \prettyref{eq:confellipse-1} contains the true
parameter, and subsequently estimated an average-over-trees-and-traits
non-coverage rate for the confidence region. The result of this simulation
is shown in \prettyref{fig:coverage}.

\begin{figure}
\centering{}\includegraphics[width=0.63\textwidth]{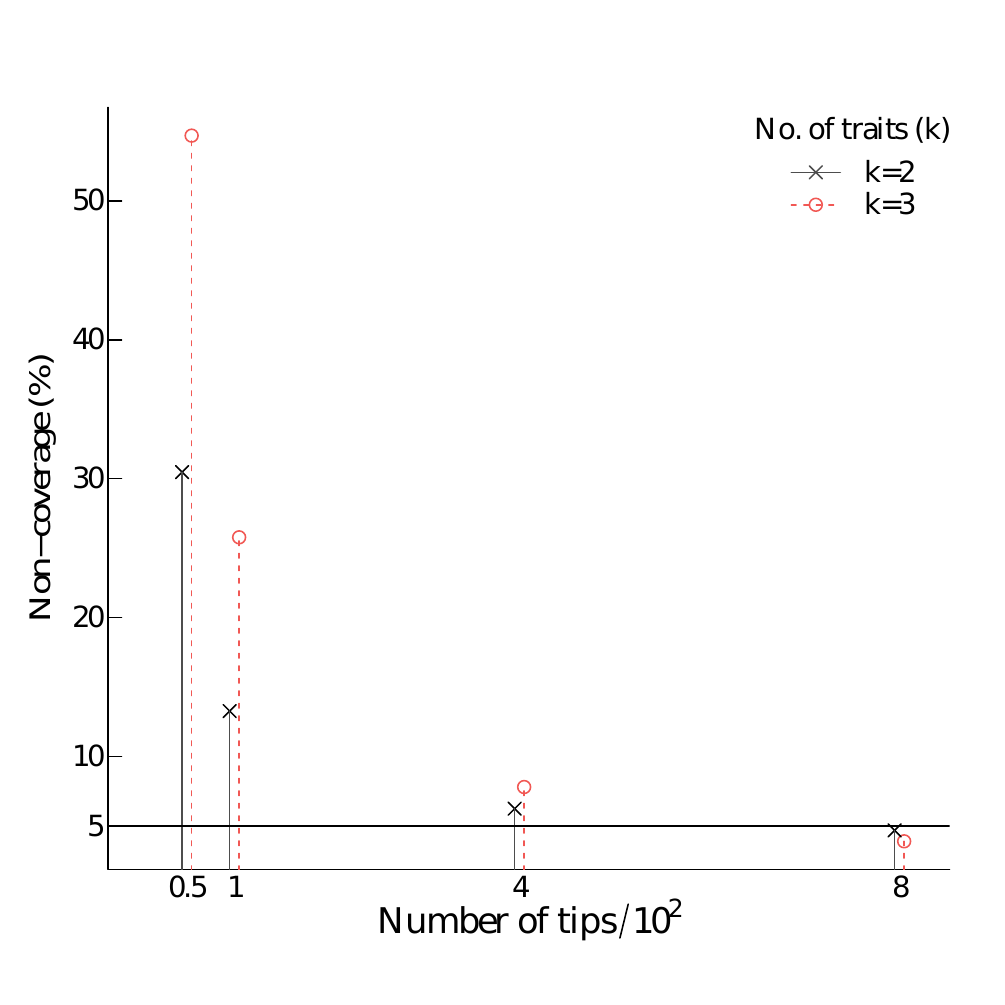}\caption{Estimated average-over-trees-and-traits non-coverage rate of the 95\%
confidence ellipse presented in \prettyref{eq:confellipse-1}. For
each pair of $(k,n_{\text{tips}})$ shown in the graph, 128 i.i.d.
sample of phylogenies with $n_{\text{tips}}$ many tips were drawn
using the \textquotedblleft rtree\textquotedblright{} function in the
R package \textquotedblleft ape\textquotedblright{} \citep{apepack}.
For each of such trees a $k$-dimensional trait data set is subsequently
drawn from the OU model according to the prescribed parameters presented
in \prettyref{eq:fixexper1}-\eqref{eq:fixexper3}; this trait data
set is then used to compute a confidence region. The non-coverage
rate in the picture is the percentage of times, among the 128 i.i.d.
sample, in which the confidence region fails to cover the truth. \label{fig:coverage}}
\end{figure}
Given that the formulae described here are so lengthy, one may wonder
whether or not there are simpler way to obtain a confidence regions.
As mentioned for example in \citet{revell_testing_2008}, if an one-species-per-row
data matrix $X$ were drawn from the simplest case of multivariate
Brownian motion model then the variance-covariance matrix of $\mathrm{vec}(X)$
can be factorized into a Kronecker product $R\otimes C$, where $R$
is a covariance matrix and $C$ is a matrix whose elements are the
total shared branch length between each pair of tips. In this formulation,
all arithmetics will be greatly simplified partly because the covariances
caused by the phylogeny is fully captured in a single matrix $C$,
which frees us from having to deal with the phylogenetic tree node-by-node.
However, the covariance structure of the data matrix will not be so
easy to work with anymore if, for example, there are multiple Brownian-motion
evolutionary regimes with differing rates. The single-regime OU model,
however, seems hopeful to be analysed exactly and analytically, as
it is not difficult to express the covariance of $X$ under this model
using a tensor-product-like structure. Some faster and \foreignlanguage{british}{more
straightforward} Bayesian formulae for the non-branching Ornstein-Uhlenbeck
process exists \citet{rajesh2018}, but it assumes equally-timed observations
during the stochastic process, which corresponds to observing the
internal nodes' trait vectors in the phylogenetic setting. Compared
to these, the second-derivatives shown in \prettyref{sec:Hessian}
do not assume, thus cannot utilize, any mathematical structures specific
to the Brownian motion or the Ornstein-Uhlenbeck process. As a consequence,
the recursion structure in \prettyref{sec:Hessian} allows the freedom
to change the function $g$, that is, \prettyref{eq:VOU}-\eqref{eq:PhiOU}
to, for instance, account for measurement error at the tips, or to
treat any internal nodes specially if one had some prior knowledge
about the nodes.

But with this level of generality, the readers are reminded that there
are generally no guarentees that the log-likelihood surface can be
sufficiently approximated by a nice quadratic function. In fact, because
$g$ is a general function, the MLE may not even exist uniquely if
the model is mis-specified or over-parameterized. While our framework
allows the users to define the function $g$ completely to their liking
and is capable of giving a Hessian for parameterization, as a cautionary
reminder, the users should ensure that their log-likelihood surface
is sufficiently approximatable by a non-singular quadratic function.
Consistency of the MLE is also only valid under some conditions\textemdash however
one defines. As an example to illustrate this subtlety, for the maximum
likelihood estimator for the ancestral state at the root (which is
not what we do here), the sufficient condition for consistency is
different depending on the definition of consistency: one might infinitely
grow the same tree by adding tips, or alternatively he might consider
a series of completely different trees whose number of tips tend to
infinity. In the former case, one might want to have a model that
describes how the tree has come to be, and how the tree will become
in future (say, the tree comes from a pure-birth tree model) in order
to reason about consistency or the asymptotic behaviour of the Hessians.
And even in the simpler latter case in which the user might just take
the tree ``as-is'', if there were multiple evolutionary regimes,
then one at least must define where the regimes are in all those theoretical
big trees, which nobody has ever observed in a real-world biological
setting; making up such a definition is in itself not an easy task
already. For this reason, in most practical scenarios, I recommend
that one should at least plot the likelihood surface to verify that
it is indeed similar to a quadratic function (they are not too bad
in most simulated scenarios I have manually inspected). Also, even
if in some cases one do not trust the Hessians due to a nonideal likelihood
surface, I still recommend using the Hessian as a diagnostic tool
to judge how good the MLE is, because any oddity in the Hessian, such
as extreme eigenvalues, is often a symptom of the MLE's ill-behaviour
or model misspecification.

From \prettyref{cor:dfrecur} we can also get some insight into what
the Hessian matrix is ``comprised of''. The corollary implies that
the ``magnitude'' of all the $H_{\eta}\Phi_{\eta}$'s throughout
the phylogeny are the key quantities that influence the numerical
values of the resulting Hessian matrix because of the way the results
of \prettyref{cor:dfrecur} is multiplied into the equations in \prettyref{prop:interm},
and then subsequently multiplied into the final Hessian matrix's entries.
One can envision that if all the $H_{\eta}\Phi_{\eta}$ throughout
the tree has large eigenvalues, then these recursion will result in
matrices with very large eigenvalues, hence then entries of the Hessian
matrix will be large, amounting to smaller uncertainty in some directions
of the resulting Hessian.

Some simplification of the algebra may be possible if one simply drop
the second term in \prettyref{eq:chainrule_simp}, hence arriving
at a delta-method-like approximation. This is also implemented in
the glinv package, in which the function for computing the Hessian
has an option to allow the user to omit this term, although anecdotal
observation suggests that this approximation can be sometimes very
good but sometimes very unreliable. Because $\mathbf{H}g^{(j)}$ becomes
more neligible when $g$ is closer to a linear function, one might
hope for a clever reparameterization to improve the delta-method approximation.
In fact, in the case of a Brownian motion, the mapping $g$ defined
by \prettyref{eq:VOU}-\eqref{eq:PhiOU} is simply a linear function
with respect to $LL^{T}$, so we should expect a much better approximation
than in the OU case. Dropping the second term might also be a good
starting point if one would like to derive asymptotic results out
of the formulae.

\section{Acknowledgement}

This work has been partially funded by the Swedish Research Council
(Vetenskapsrådet), grant no 2017-04951.

\bibliographystyle{elsarticle-harv}
\phantomsection\addcontentsline{toc}{section}{\refname}\bibliography{Manuscript}

\pagebreak
\begin{center}
\begin{hyphenrules}{nohyphenation}
\LARGE\bfseries Supplemental Materials: Exact Expressions for the Log-likelihood's Hessian in Multivariate Continuous-Time Continuous-Trait Gaussian Evolution along a Phylogeny\par
\end{hyphenrules}
\end{center}
\setcounter{equation}{0}
\setcounter{figure}{0}
\setcounter{table}{0}
\setcounter{page}{1}
\makeatletter
\renewcommand{\theequation}{S\arabic{equation}}
\renewcommand{\thefigure}{S\arabic{figure}}
\renewcommand{\bibnumfmt}[1]{[S#1]}
\renewcommand{\citenumfont}[1]{S#1}

\renewcommand\eqref[1]{\prettyref{eq:#1}}
\small

\global\long\def\etaone{\eta_{1}}%
\global\long\def\etaj{\eta_{j}}%
\global\long\def\etaJ{\eta_{J}}%
\global\long\def\rlr{ \begin{bmatrix} r_{\etaone} \\
 \vdots\\
 r_{\etaJ} \\
 
\end{bmatrix}}%
\global\long\def\xlr{ \begin{bmatrix} x_{\etaone} \\
 \vdots\\
 x_{\etaJ} \\
 
\end{bmatrix}}%
\global\long\def\xlrT{ \begin{bmatrix} x_{\etaone}^{T}  &  \cdots &  x_{\etaJ}^{T} \end{bmatrix}}%
\global\long\def\qlr{ \begin{bmatrix} q_{\etaone} \\
 \vdots\\
 q_{\etaJ} 
\end{bmatrix}}%
\global\long\def\qlrT{ \begin{bmatrix} q_{\etaone}^{T}  &  \cdots &  q_{\etaJ}^{T} \end{bmatrix}}%
\global\long\def\Philr{ \begin{bmatrix} \Phi_{\etaone} \\
 \vdots\\
 \Phi_{\etaJ} 
\end{bmatrix}}%
\global\long\def\PhilrT{ \begin{bmatrix} \Phi_{\etaone}^{T}  &  \cdots &  \Phi_{\etaJ}^{T} \end{bmatrix}}%
\global\long\def\Psilr{ \begin{bmatrix} \Psi_{\etaone} \\
 \vdots\\
 \Psi_{\etaJ} 
\end{bmatrix}}%
\global\long\def\PsilrT{ \begin{bmatrix} \Psi_{\etaone}^{T}  &  \cdots &  \Psi_{\etaJ}^{T} \end{bmatrix}}%
\global\long\def\PsilrPhiu{ \begin{bmatrix} \Psi_{\etaone}\\
 \vdots\\
 \Psi_{\etaJ} 
\end{bmatrix}\Phi_{u}}%
\global\long\def\PsilrTPhiuT{ \Phi_{u}^{T} \begin{bmatrix} \Psi_{\etaone}^{T}  &  \cdots &  \Psi_{\etaJ}^{T} \end{bmatrix}}%
\global\long\def\tlrT{ \begin{bmatrix} t_{\etaone}^{T}  &  \cdots &  t_{\etaJ}^{T} \end{bmatrix}}%
\global\long\def\tlr{ \begin{bmatrix} t_{\etaone} \\
 \vdots\\
 t_{\etaJ} 
\end{bmatrix}}%
\global\long\def\GlrT{ \begin{bmatrix} G_{\etaone}^{T}  &  \cdots &  G_{\etaJ}^{T} \end{bmatrix}}%
\global\long\def\Glr{ \begin{bmatrix} G_{\etaone} \\
 \vdots\\
 G_{\etaJ} 
\end{bmatrix}}%
\global\long\def\HlrT{ \begin{bmatrix} H_{\etaone}^{T}  &  \cdots &  H_{\etaJ}^{T} \end{bmatrix}}%
\global\long\def\Hlr{ \begin{bmatrix} H_{\etaone} \\
 \vdots\\
 H_{\etaJ} 
\end{bmatrix}}%
\global\long\def\Sig{\Sigma}%
\global\long\def\Siglr{ \begin{bmatrix} \Sig_{\etaone}  &  0  &  0  &  \cdots\\
 0  &  \Sig_{\eta_{2}}  &  0  &  \cdots\\
 0  &  0  &  \ddots &  \vdots\\
 0  &  0  &  \cdots &  \Sig_{\etaJ} 
\end{bmatrix}}%
\global\long\def\Siglrinv{ \begin{bmatrix} \Sig_{\etaone}^{-1}  &  0  &  0  &  \cdots\\
 0  &  \Sig_{\eta_{2}}^{-1}  &  0  &  \cdots\\
 0  &  0  &  \ddots &  \vdots\\
 0  &  0  &  \cdots &  \Sig_{\etaJ}^{-1} 
\end{bmatrix}}%
\global\long\def\E{\mathbb{E}}%
\global\long\def\Cov{\mathrm{Cov}}%
\global\long\def\Var{\mathrm{Var}}%
\global\long\def\R{\mathbb{R}}%
\global\long\def\sumc{\sum_{1\le j\le J}c_{\etaj}}%
\global\long\def\sumgam{\sum_{1\le j\le J}\gamma_{\etaj}}%
\global\long\def\sumOmega{\sum_{1\le j\le J}\Omega_{\etaj}}%
\global\long\def\sumDelta{\sum_{1\le j\le J}\Delta_{\etaj}}%
\global\long\def\children{\textrm{{Ch}}}%
\global\long\def\parent{\textrm{{Pa}}}%
\global\long\def\detaone#1{\frac{\partial{#1}}{\partial\theta_{\eta_{1}}}}%
\global\long\def\detatwo#1{\frac{\partial{#1}}{\partial\theta{}_{\eta_{2}}}}%
\global\long\def\deta#1{\frac{\partial{#1}}{\partial\theta{}_{\eta}}}%
\global\long\def\hrooteta#1{\frac{\partial^{2}{#1}_{u_{1}}}{\partial\theta_{\eta_{1}}\partial\theta_{\eta_{2}}}}%
\global\long\def\setchar#1{1_{#1}}%
\global\long\def\descendents{\textrm{{Ch}}^{*}}%
\global\long\def\D{\textrm{{D}}}%
\global\long\def\d{\textrm{{d}}}%
\global\long\def\vecc{\operatorname{vec}}%
\section*{Proof of Proposition 1}
\begin{proof}
We let $J=J_{u}$ for ease of notation. Let $\psi_{x_{u}|z_{v}}$
be the multivariate characteristic functions of $p(x_{u}|z_{u})$
and for convenience we denote $r=\rlr$. We have \setlength\arraycolsep{1.5pt}
\begin{align}
\psi_{x_{u}|z_{v}}(r) & =\E_{|z_{v}}\exp\left(ir^{T}\xlr\right)\\
 & =\E_{|z_{v}}\E_{|z_{u},z_{v}}\exp\left(ir^{T}\xlr\right)\\
 & =\E_{|z_{v}}\E_{|z_{u}}\exp\left(ir^{T}\xlr\right)\\
 & =\E_{|z_{v}}\E_{|z_{u}}\exp\left(i\sum_{1\le j\le J}r_{\etaj}^{T}x_{\etaj}\right)\\
 & =\E_{|z_{v}}\prod_{1\le j\le J}\E_{|z_{u}}\exp\left(ir_{\etaj}^{T}x_{\etaj}\right)\\
 & =\E_{|z_{v}}\prod_{1\le j\le J}\exp\left[i(q_{\etaj}+\Psi_{\etaj}z_{u})^{T}r_{\etaj}-\frac{1}{2}r_{\etaj}^{T}\Sig_{\etaj}r_{\etaj}\right]\\
 & =\exp\left(\sum_{1\le j\le J}iq_{\etaj}^{T}r_{\etaj}-\sum_{1\le j\le J}\frac{1}{2}r_{\etaj}^{T}\Sig_{\etaj}r_{\etaj}\right)\cdot\E_{|z_{v}}\exp\left[i\left(\sum_{1\le j\le J}r_{\etaj}^{T}\Psi_{\etaj}\right)z_{u}\right]\\
 & =\exp\left(ir^{T}\qlr-\frac{1}{2}r^{T}\Siglr r\right)\cdot\E_{|z_{v}}\exp\left(ir^{T}\Psilr z_{u}\right)\\
\end{align}
\begin{align}
 & =\exp\left(ir^{T}\qlr-\frac{1}{2}r^{T}\Siglr r\right)\\
 & \quad\quad\cdot\exp\left(ir^{T}\Psilr\left(\Phi_{u}z_{v}+w_{u}\right)-\frac{1}{2}r^{T}\Psilr V_{u}\PsilrT r\right)\\
 & =\exp\left\{ \begin{gathered}ir^{T}\left(\qlr+\Psilr\left(\Phi_{u}z_{v}+w_{u}\right)\right)\quad\quad\quad\quad\quad\quad\quad\quad\quad\quad\quad\\
-\frac{1}{2}r^{T}\left(\Siglr+\Psilr V_{u}\PsilrT\right)r
\end{gathered}
\right\} \\
 & =\exp\left(ir^{T}\left(q_{u}+\Psi_{u}z_{v}\right)-\frac{1}{2}r^{T}\Sig_{u}r\right).
\end{align}
\end{proof}

\section*{Proof of Proposition 2}

Before getting into Proposition 2, we first derive a lemma that relates
any bi-linear form of $\Sig_{u}^{-1}$ to $\Sig_{\etaj}^{-1}$ (recall
our notation that node $\eta_{j}$ is one of the direct children of
node $u$). This lemma will be used repeatedly in the proof of Proposition
2, and is the central main technique that is used to avoid storing
and inverting the all-species-all-traits covariance matrix.

\global\long\def\pl{\xl- q_{\ell}}%
 
\global\long\def\pr{\xr- q_{r}}%
 
\global\long\def\Psilwu{\Psi_{\ell}w_{u}}%
 
\global\long\def\Psirwu{\Psi_{r} w_{u}}%
 
\global\long\def\plr{\xlr- \qlr}%
 
\global\long\def\plrT{\xlrT- \qlrT}%
 
\global\long\def\Psilrwu{\Psilr w_{u}}%
 
\global\long\def\sumj{\sum_{1\le j\le J}}%

In this section we also let $J=J_{u}$ for ease of notation. 
\begin{lem}
\label{lem:recquad} For all matrices $G_{\ell},H_{\ell}\in\R^{k_{\ell}\times k_{u}}$
and $G_{r},H_{r}\in\R^{k_{r}\times k_{u}}$, 

\begin{multline*}
\GlrT\Sig_{u}^{-1}\Hlr=\sum_{1\le j\le J}G_{\etaj}^{T}\Sig_{\etaj}^{-1}H_{\etaj}\\
-\left(\sum_{1\le j\le J}G_{\etaj}^{T}\Sig_{\etaj}^{-1}\Psi_{\etaj}\right)\left(V_{u}^{-1}+\sum_{1\le j\le J}\Psi_{\etaj}^{T}\Sig_{\etaj}^{-1}\Psi_{\etaj}\right)^{-1}\left(\sum_{1\le j\le J}\Psi_{\etaj}^{T}\Sig_{\etaj}^{-1}H_{\etaj}\right)
\end{multline*}
\end{lem}

\begin{proof}
Recall that Woodbury formula allows us to write inverse of sum to
sum of inverse 
\begin{equation}
(A+UCV)^{-1}=A^{-1}-A^{-1}U(C^{-1}+VA^{-1}U)^{-1}VA^{-1}.
\end{equation}
Applying the above to \eqref{qpsisig} gives us\setlength\arraycolsep{0.3pt}{\footnotesize{}
\begin{align}
\Sig_{u}^{-1} & =\left(\Siglr+\Psilr V_{u}\PsilrT\right)^{-1}\\
 & =\Siglrinv-\Siglrinv\Psilr\\
 & \quad\quad\quad\left(V_{u}^{-1}+\PsilrT\Siglrinv\Psilr\right)^{-1}\PsilrT\Siglrinv
\end{align}
}\setlength\arraycolsep{1.5pt}Therefore, multiplying $\GlrT$ and
$\Hlr$ on the left and right respectively immediately results in 

\begin{multline*}
\GlrT\Sig_{u}^{-1}\Hlr=\sum_{1\le j\le J}G_{\etaj}^{T}\Sig_{\etaj}^{-1}H_{\etaj}\\
\quad\quad-\left(\sum_{1\le j\le J}G_{\etaj}^{T}\Sig_{\etaj}^{-1}\Psi_{\etaj}\right)\left(V_{u}^{-1}+\sum_{1\le j\le J}\Psi_{\etaj}^{T}\Sig_{\etaj}^{-1}\Psi_{\etaj}\right)^{-1}\left(\sum_{1\le j\le J}\Psi_{\etaj}^{T}\Sig_{\etaj}^{-1}H_{\etaj}\right),
\end{multline*}
which is the required equation. 
\end{proof}
Now we proceed to prove Proposition 2 in the main text.
\begin{proof}[Proof of Proposition 2.]
First we prove the equation of $c_{u}$. To ease notation, we let
$p_{\etaj}=x_{\etaj}-q_{\etaj}$, and $G_{\etaj}=H_{\etaj}=p_{\etaj}-\Psi_{\etaj}w_{u}$.
So by the definition of $c_{u}$ and an application of \lemref{recquad}
and \eqref{qu} we have

\begin{align}
c_{u} & =(x_{u}-q_{u})^{T}\Sig_{u}^{-1}(x_{u}-q_{u})\\
 & =\left(\xlr-\qlr-\Psilr w_{u}\right)^{T}\Sig_{u}^{-1}\left(\xlr-\qlr-\Psilr w_{u}\right)\\
 & =\GlrT\Sig_{u}^{-1}\Hlr\\
 & =\sum_{1\le j\le J}G_{\etaj}^{T}\Sig_{\etaj}^{-1}H_{\etaj}-\left(\sum_{1\le j\le J}G_{\etaj}^{T}\Sig_{\etaj}^{-1}\Psi_{\etaj}\right)\left(V_{u}^{-1}+\sum_{1\le j\le J}\Psi_{\etaj}^{T}\Sig_{\etaj}^{-1}\Psi_{\etaj}\right)^{-1}\left(\sum_{1\le j\le J}\Psi_{\etaj}^{T}\Sig_{\etaj}^{-1}H_{\etaj}\right)\label{eq:cuquad}
\end{align}
Now if we expand at the terms one-by-one, we have, for the first term
above,

\begin{align}
\begin{split}G_{\etaj}^{T}\Sig_{\etaj}^{-1}H_{\etaj} & =(p_{\etaj}-\Psi_{\etaj}w_{u})^{T}\Sig_{\etaj}^{-1}(p_{\etaj}-\Psi_{\etaj}w_{u})\\
 & =p_{\etaj}^{T}\Sigma_{\etaj}^{-1}p_{\etaj}-p_{\etaj}^{T}\Sig_{\etaj}^{-1}\Psi_{\etaj}w_{u}-w_{u}^{T}\Psi_{\etaj}^{T}\Sig_{\etaj}^{-1}p_{\etaj}+w_{u}^{T}\Psi_{\etaj}^{T}\Sig_{\etaj}^{-1}\Psi_{\etaj}w_{u}\\
 & =c_{\etaj}-\gamma_{\etaj}^{\dagger}w_{u}-w_{u}^{T}\gamma_{\etaj}+w_{u}^{T}\Omega_{\etaj}w_{u}
\end{split}
\end{align}
and for the second term, observe that
\begin{align*}
G_{\etaj}^{T}\Sig_{\etaj}^{-1}\Psi_{\etaj} & =(p_{\etaj}-\Psi_{\etaj}w_{u})^{T}\Sig_{\etaj}^{-1}\Psi_{\etaj}=\gamma_{\etaj}^{\dagger}-w_{u}^{T}\Omega_{\eta_{j}}w_{u}\\
\Psi_{\etaj}^{T}\Sig_{\etaj}^{-1}H_{\etaj} & =\Psi_{\etaj}^{T}\Sig_{\etaj}^{-1}(p_{\etaj}-\Psi_{\etaj}w_{u})=\gamma_{\etaj}-w_{u}^{T}\Omega_{\eta_{j}}w_{u}
\end{align*}

Substituting the above into \eqref{cuquad} immediately results in
\eqref{curecur} because 
\begin{align}
c_{u} & =\sumj\left(c_{\etaj}-2\gamma_{\etaj}^{T}w_{u}+w_{u}^{T}\Omega_{\etaj}w_{u}\right)\\
 & \quad\quad\quad\quad-\left(\sumj\gamma_{\etaj}^{T}-w_{u}^{T}\Omega_{\etaj}\right)\left(V_{u}^{-1}+\sumj\Omega_{\etaj}\right)^{-1}\left(\sumj\gamma_{\etaj}-\Omega_{\etaj}^{T}w_{u}\right)\\
 & =\sumc-2\sumgam^{T}w_{u}+w_{u}^{T}\left(\sumOmega\right)w_{u}-b_{u}^{T}\Lambda_{u}b_{u}
\end{align}

Similarly, to prove \eqref{gammaurecur}, we apply \lemref{recquad}
to the definition of $\gamma_{u}$, which leads to\setlength\arraycolsep{0.3pt}{\footnotesize{}
\begin{align}
\gamma_{u} & =\Psi_{u}^{T}\Sig_{u}^{-1}(x_{u}-q_{u})\\
 & =\Psi_{u}^{T}\Sig_{u}^{-1}\left(\xlr-\qlr-\Psilr w_{u}\right)\\
 & =\PsilrTPhiuT\Sig_{u}^{-1}\begin{bmatrix}p_{\etaone}-\Psi_{\etaone}w_{u}\\
\vdots\\
p_{\etaJ}-\Psi_{\etaJ}w_{u}
\end{bmatrix}\\
 & =\Phi_{u}^{T}\left\{ \begin{gathered}\left(\sum_{1\le j\le J}\Psi_{\etaj}^{T}\Sig_{\etaj}^{-1}\left(p_{\etaj}-\Psi_{\etaj}w_{u}\right)\right)-\left(\sum_{1\le j\le J}\Psi_{\etaj}^{T}\Sig_{\etaj}^{-1}\Psi_{\etaj}\right)\quad\quad\quad\quad\\
\quad\quad\left(V_{u}^{-1}+\sum_{1\le j\le J}\Psi_{\etaj}^{T}\Sig_{\etaj}^{-1}\Psi_{\etaj}\right)^{-1}\left(\sum_{1\le j\le J}\Psi_{\etaj}^{T}\Sig_{\etaj}^{-1}\left(p_{\etaj}-\Psi_{\etaj}w_{u}\right)\right)
\end{gathered}
\right\} \\
 & =\Phi_{u}^{T}\left\{ \left(\sum_{1\le j\le J}\gamma_{\etaj}-\Omega_{\etaj}w_{u}\right)-\left(\sum_{1\le j\le J}\Omega_{\etaj}\right)\Lambda_{u}\left(\sum_{1\le j\le J}\gamma_{\etaj}-\Omega_{\etaj}w_{u}\right)\right\} \\
 & 
\Phi_{u}^{T}\left\{ b_{u}-\left(\sum_{1\le j\le J}\Omega_{\etaj}\right)\Lambda_{u}b_{u}\right\} \\
 & =\Phi_{u}^{T}\left\{ I-\left(\sum_{1\le j\le J}\Omega_{\etaj}\right)\Lambda_{u}\right\} b_{u},
\end{align}
}{\footnotesize\par}

therefore \eqref{gammaurecur} holds.

Using the same sort of argument, applying \lemref{recquad} to the
definition of $\Omega_{u}$ we have
\begin{align}
\Omega_{u} & =\PsilrTPhiuT\Sig_{u}^{-1}\PsilrPhiu\\
 & =\Phi_{u}^{T}\left[\sumj\Psi_{\etaj}^{T}\Sig_{\etaj}^{-1}\Psi_{\etaj}-\left(\sumj\Psi_{\etaj}^{T}\Sig_{\etaj}^{-1}\Psi_{\etaj}\right)\left(V_{u}^{-1}\sumj\Psi_{\etaj}^{T}\Sig_{\etaj}^{-1}\Psi_{\etaj}\right)^{-1}\left(\sumj\Psi_{\etaj}^{T}\Sig_{\etaj}^{-1}\Psi_{\etaj}\right)\right]\Phi_{u}\\
 & =\Phi_{u}^{T}\left(\sumj\Omega_{\etaj}-\Omega_{\etaj}\Lambda_{u}\Omega_{\etaj}\right)\Phi_{u}\\
 & =\Phi_{u}^{T}\sumj\Omega_{\etaj}\left[I_{k_{u}}-\Lambda_{u}\sumj\Omega_{\etaj}\right]\Phi_{u}
\end{align}

Finally, recall that \setlength\arraycolsep{0.3pt}
\begin{equation}
\det\left(A+UWV\right)=\det(W^{-1}+VA^{-1}U)\det W\det A
\end{equation}
for any invertible matrix $A$, $W$ and matrix $U$, $V$ of suitable
dimension. Applying this to \eqref{Deltau} results in 
\begin{align}
\Delta_{u} & =\ln\det\Sig_{u}\\
 & =\ln\det\left(\Siglr+\Psilr V_{u}\PsilrT\right)\\
 & =\ln\det\left(V_{u}^{-1}+\PsilrT\Siglrinv\Psilr\right)+\ln\det V_{u}+\ln\det\Siglr\\
 & =\ln\det\left(V_{u}^{-1}+\sumj\Omega_{\etaj}\right)+\ln\det V_{u}+\sumj\Delta_{\etaj},
\end{align}
as required.\setlength\arraycolsep{1.5pt}
\end{proof}
\begin{lem}
Furthermore, the all-species-all-trait likelihood given the root node
$o$ can be expressed as 
\begin{equation}
\ln p(x_{o}|z_{o})=-\frac{1}{2}\left[\sumj c_{\kappa_{j}}+z_{o}^{T}\left(\sumj\Omega_{\kappa_{j}}\right)z_{o}+\sumj\Delta_{\kappa_{j}}\right]+\left(\sumj\gamma_{\kappa_{j}}^{T}\right)z_{o}-\frac{Jk}{2}\ln(2\pi)\label{likmaster}
\end{equation}
where the nodes $\kappa_{j}\in\children(o)$ .
\end{lem}

\begin{proof}
{\footnotesize{}
\begin{align}
\ln p(x_{o}|z_{o}) & =\sumj\ln p(x_{\kappa_{j}}|z_{o})\\
 & =\sumj\ln\phi(x_{\kappa_{j}};q_{\kappa_{j}}+\Psi_{\kappa_{j}}z_{o},\Sig_{\kappa_{j}})\\
 & =\sumj\left[\begin{aligned}-\frac{1}{2}\left(x_{\kappa_{j}}-q_{\kappa_{j}}-\Psi_{\kappa_{j}}z_{o}\right)^{T}\Sig_{\kappa_{j}}^{-1}\left(x_{\kappa_{j}}-q_{\kappa_{j}}-\Psi_{\kappa_{j}}z_{o}\right)\quad\quad\quad\quad\\
-\frac{1}{2}\sumj\ln\det{\Sig_{\kappa_{j}}}-\frac{k}{2}\ln(2\pi)
\end{aligned}
\right]\\
 & =\sumj\left[\begin{aligned}-\frac{1}{2}(x_{\kappa_{j}}-q_{\kappa_{j}})^{T}\Sig_{\kappa_{j}}^{-1}(x_{\kappa_{j}}-q_{\kappa_{j}})\quad\quad\quad\quad\quad\quad\quad\quad\quad\quad\quad\quad\quad\quad\quad\quad\\
+(x_{\kappa_{j}}-q_{\kappa_{j}})^{T}\Sig_{\kappa_{j}}^{-1}\Psi_{\kappa_{j}}z_{o}-\frac{1}{2}z_{o}^{T}\Psi_{\kappa_{j}}^{T}\Sig_{\kappa_{j}}^{-1}\Psi_{\kappa_{j}}z_{o}-\frac{1}{2}\sumj\Delta_{\kappa_{j}}
\end{aligned}
\right]-\frac{Jk}{2}\ln(2\pi)\\
 & =\sumj\left[-\frac{c_{\kappa_{j}}}{2}+\gamma_{\kappa_{j}}^{T}z_{o}-\frac{1}{2}z_{o}^{T}\Omega_{mj}z_{o}\right]-\frac{1}{2}\sumj\Delta_{\kappa_{j}}-\frac{Jk}{2}\ln(2\pi)\\
 & =-\frac{1}{2}\left[\sumj c_{\kappa_{j}}+z_{o}^{T}\left(\sumj\Omega_{\kappa_{j}}\right)z_{o}+\sumj\Delta_{\kappa_{j}}\right]+\left(\sumj\gamma_{\kappa_{j}}^{T}\right)z_{o}-\frac{Jk}{2}\ln(2\pi)
\end{align}
}{\footnotesize\par}
\end{proof}
Proposition 2 in the main text is a direct result of the above lemmata.

\section*{Proof of Proposition 3}

\begin{proof}
For each equation in the proposition, we prove by induction on $k$
starting from $k=n$ and inducing downward by showing that the equation
must hold for $k-1$ if it holds for any $k$.

First we prove that \eqref{dchnOmega} holds for all $k$. When $k=n$,
by definition immediately we have $F_{u_{k}}^{\eta}=F_{u_{n}}^{u_{n}}=I$.
Otherwise if \eqref{dchnOmega} holds for any $k$ then
\begin{align}
\frac{\partial\Omega_{u_{k-1}}}{\partial\theta_{u_{n}}} & =\frac{\partial}{\partial\theta_{u_{n}}}\left[\Phi_{u_{k-1}}^{T}\left(\sum_{\iota\in\children\left(u_{k-1}\right)}\Omega_{\iota}\right)H_{u_{k-1}}\Phi_{u_{k-1}}\right]\nonumber \\
 & =\Phi_{u_{k-1}}^{T}\left\{ \left[\frac{\partial}{\partial\theta_{u_{n}}}\left(\sum_{\iota\in\children\left(u_{k-1}\right)}\Omega_{\iota}\right)\right]H_{u_{k-1}}+\left(\sum_{\iota\in\children\left(u_{k-1}\right)}\Omega_{\iota}\right)\frac{\partial H_{u_{k-1}}}{\partial\theta_{u_{n}}}\right\} \Phi_{u_{k-1}}.\label{dodttmp-1}
\end{align}
and
\begin{align}
\frac{\partial\Lambda_{u_{k-1}}}{\partial\theta_{u_{n}}} & =\frac{\partial}{\partial\theta_{u_{n}}}\left(V_{u_{k-1}}^{-1}+\sum_{\iota\in\children\left(u_{k-1}\right)}\Omega_{\iota}\right)^{-1}\nonumber \\
 & =-\Lambda_{u_{k-1}}\frac{\partial}{\partial\theta_{u_{n}}}\left(V_{u_{k-1}}^{-1}+\sum_{\iota\in\children\left(u_{k-1}\right)}\Omega_{\iota}\right)\Lambda_{u_{k-1}}\nonumber \\
 & =-\Lambda_{u_{k-1}}\frac{\partial\Omega_{u_{k}}}{\partial\theta_{u_{n}}}\Lambda_{u_{k-1}},\label{eq:dLambdadany-1}
\end{align}
where the last equality holds because only the one and only one child
of $u_{k-1}$ that leads to $u_{n}$ has non-zero derivative with
respect to $u_{n}'s$ parameters, and this child is by defintion $u_{k}$.

Using this, we can continue the derivation in \ref{dodttmp-1}, which
leads to 
\begin{align*}
\frac{\partial\Omega_{u_{k-1}}}{\partial\theta_{u_{n}}} & =\Phi_{u_{k-1}}^{T}\left\{ \frac{\partial\Omega_{u_{k}}}{\partial\theta_{u_{n}}}H_{u_{k-1}}+\left(\sum_{\iota\in\children\left(u_{k-1}\right)}\Omega_{\iota}\right)\frac{\partial}{\partial\theta_{u_{n}}}\left(I-\Lambda_{u_{k-1}}\sum_{\iota\in\children\left(u_{k-1}\right)}\Omega_{\iota}\right)\right\} \Phi_{u_{k-1}}\\
 & =\Phi_{u_{k-1}}^{T}\left\{ \frac{\partial\Omega_{u_{k}}}{\partial\theta_{u_{n}}}H_{u_{k-1}}+\left(\sum_{\iota\in\children\left(u_{k-1}\right)}\Omega_{\iota}\right)\left[\Lambda_{u_{k-1}}\frac{\partial\Omega_{u_{k}}}{\partial\theta_{u_{n}}}\Lambda_{u_{k-1}}\sum_{\iota\in\children\left(u_{k-1}\right)}\Omega_{\iota}-\Lambda_{u_{k-1}}\frac{\partial\Omega_{u_{k}}}{\partial\theta_{u_{n}}}\right]\right\} \Phi_{u_{k-1}}\\
 & =\Phi_{u_{k-1}}^{T}\left\{ \frac{\partial\Omega_{u_{k}}}{\partial\theta_{u_{n}}}H_{u_{k-1}}+\left(\sum_{\iota\in\children\left(u_{k-1}\right)}\Omega_{\iota}\right)\Lambda_{u_{k-1}}\frac{\partial\Omega_{u_{k}}}{\partial\theta_{u_{n}}}\left(\Lambda_{u_{k-1}}\sum_{\iota\in\children\left(u_{k-1}\right)}\Omega_{\iota}-I\right)\right\} \Phi_{u_{k-1}}\\
 & =\Phi_{u_{k-1}}^{T}\left\{ \frac{\partial\Omega_{u_{k}}}{\partial\theta_{u_{n}}}H_{u_{k-1}}-\left(\sum_{\iota\in\children\left(u_{k-1}\right)}\Omega_{\iota}\right)\Lambda_{u_{k-1}}\frac{\partial\Omega_{u_{k}}}{\partial\theta_{u_{n}}}H_{u_{k-1}}\right\} \Phi_{u_{k-1}}\\
 & =\Phi_{u_{k-1}}^{T}\left\{ I-\left(\sum_{\iota\in\children\left(u_{k-1}\right)}\Omega_{\iota}\right)\Lambda_{u_{k-1}}\right\} \frac{\partial\Omega_{u_{k}}}{\partial\theta_{u_{n}}}H_{u_{k-1}}\Phi_{u_{k-1}}\\
 & =\Phi_{u_{k-1}}^{T}H_{u_{k-1}}^{T}\frac{\partial\Omega_{u_{k}}}{\partial\theta_{u_{n}}}H_{u_{k-1}}\Phi_{u_{k-1}}\\
 & =\Phi_{u_{k-1}}^{T}H_{u_{k-1}}^{T}F_{u_{k}}^{u_{n}T}\frac{\partial\Omega_{u_{n}}}{\partial\theta_{u_{n}}}F_{u_{k}}^{u_{n}}H_{u_{k-1}}\Phi_{u_{k-1}}\\
 & =F_{u_{k-1}}^{u_{n}T}\frac{\partial\Omega_{u_{n}}}{\partial\theta_{u_{n}}}F_{u_{k-1}}^{u_{n}},
\end{align*}
as required to complete the induction proof of \eqref{dchnOmega}
for all $k=1\cdots n$.

We will now prove \eqref{dchngamma} holds for all $k=1\cdots n$
using the same induction. When $k=n$, the equation is trivially true
because $F_{u_{n}}^{u_{n}}=I$ and $g_{u_{n}}^{u_{n}}=0$ by definition.
Notice that
\begin{align*}
\frac{\partial b_{u_{k-1}}}{\partial\theta_{u_{n}}} & =\frac{\partial}{\partial\theta_{u_{n}}}\left(\sum_{\iota\in\children\left(u_{k-1}\right)}\gamma_{\iota}-\sum_{\iota\in\children\left(u_{k-1}\right)}\Omega_{\iota}w_{u_{k-1}}\right)\\
 & =\frac{\partial\gamma_{u_{k}}}{\partial\theta_{u_{n}}}-\frac{\partial\Omega_{u_{k}}}{\partial\theta_{u_{n}}}w_{u_{k-1}}.
\end{align*}
Using this, suppose \eqref{dchngamma} were true for any $k$, we
have
\begin{align*}
\frac{\partial\gamma_{u_{k-1}}}{\partial\theta_{u_{n}}} & =\frac{\partial}{\partial\theta_{u_{n}}}\left\{ \Phi_{u_{k-1}}^{T}\left[I-\left(\sum_{\iota\in\children\left(u_{k-1}\right)}\Omega_{\iota}\right)\Lambda_{u_{k-1}}\right]b_{u_{k-1}}\right\} \\
 & =\Phi_{u_{k-1}}^{T}\left\{ \left[\frac{\partial}{\partial\theta_{u_{n}}}\left[I-\left(\sum_{\iota\in\children\left(u_{k-1}\right)}\Omega_{\iota}\right)\Lambda_{u_{k-1}}\right]\right]b_{u_{k-1}}+\left[I-\left(\sum_{\iota\in\children\left(u_{k-1}\right)}\Omega_{\iota}\right)\Lambda_{u_{k-1}}\right]\frac{\partial}{\partial\theta_{u_{n}}}b_{u_{k-1}}\right\} \\
 & =\Phi_{u_{k-1}}^{T}\left\{ \begin{aligned}\left[-\frac{\partial\Omega_{u_{k}}}{\partial\theta_{u_{n}}}\Lambda_{u_{k-1}}-\left(\sum_{\iota\in\children\left(u_{k-1}\right)}\Omega_{\iota}\right)\left(-\Lambda_{u_{k-1}}\frac{\partial\Omega_{u_{k}}}{\partial\theta_{u_{n}}}\Lambda_{u_{k-1}}\right)\right]b_{u_{k-1}}\quad\quad\\
+H_{u_{k-1}}^{T}\left(\frac{\partial\gamma_{u_{k}}}{\partial\theta_{u_{n}}}-\frac{\partial\Omega_{u_{k}}}{\partial\theta_{u_{n}}}w_{u_{k-1}}\right)
\end{aligned}
\right\} \\
 & =\Phi_{u_{k-1}}^{T}\left\{ \begin{aligned}-\left[I-\left(\sum_{\iota\in\children\left(u_{k-1}\right)}\Omega_{\iota}\right)\Lambda_{u_{k-1}}\right]\left(\frac{\partial\Omega_{u_{k}}}{\partial\theta_{u_{n}}}\Lambda_{u_{k-1}}\right)b_{u_{k-1}}\quad\quad\quad\quad\quad\quad\quad\quad\\
+H_{u_{k-1}}^{T}\left[F_{u_{k}}^{u_{n}T}\left(\frac{\partial\gamma_{u_{n}}}{\partial\theta_{u_{n}}}-\frac{\partial\Omega_{u_{n}}}{\partial\theta_{u_{n}}}g_{u_{k}}^{u_{n}}\right)-F_{u_{k}}^{u_{n}T}\frac{\partial\Omega_{u_{n}}}{\partial\theta_{u_{n}}}F_{u_{k}}^{u_{n}}w_{u_{k-1}}\right]
\end{aligned}
\right\} \\
 & =\Phi_{u_{k-1}}^{T}\left\{ \begin{aligned}-H_{u_{k-1}}^{T}F_{u_{k}}^{u_{n}T}\left(\frac{\partial\Omega_{u_{n}}}{\partial\theta_{u_{n}}}\right)F_{u_{k}}^{u_{n}}\Lambda_{u_{k-1}}b_{u_{k-1}}\quad\quad\quad\quad\quad\quad\quad\quad\quad\quad\quad\quad\quad\quad\quad\\
+H_{u_{k-1}}^{T}\left[F_{u_{k}}^{u_{n}T}\left(\frac{\partial\gamma_{u_{n}}}{\partial\theta_{u_{n}}}-\frac{\partial\Omega_{u_{n}}}{\partial\theta_{u_{n}}}g_{u_{k}}^{u_{n}}\right)-F_{u_{k}}^{u_{n}T}\frac{\partial\Omega_{u_{n}}}{\partial\theta_{u_{n}}}F_{u_{k}}^{u_{n}}w_{u_{k-1}}\right]
\end{aligned}
\right\} \\
 & =\Phi_{u_{k-1}}^{T}H_{u_{k-1}}^{T}F_{u_{k}}^{u_{n}T}\left\{ \frac{\partial\gamma_{u_{n}}}{\partial\theta_{u_{n}}}-\left(\frac{\partial\Omega_{u_{n}}}{\partial\theta_{u_{n}}}\right)\left[g_{u_{k}}^{u_{n}}+F_{u_{k}}^{u_{n}}\left(\Lambda_{u_{k-1}}b_{u_{k-1}}+w_{u_{k-1}}\right)\right]\right\} \\
 & =F_{u_{k-1}}^{u_{n}T}\left(\frac{\partial\gamma_{u_{n}}}{\partial\theta_{u_{n}}}-\left(\frac{\partial\Omega_{u_{n}}}{\partial\theta_{u_{n}}}\right)g_{u_{k-1}}^{u_{n}}\right);
\end{align*}
therefore by induction \eqref{dchngamma} has been proven.

To prove \eqref{dchnc}, first, to slightly simplify the notation,
we let 
\[
M_{ij}=a_{u_{i}}^{T}F_{u_{i+1}}^{u_{n}T}\frac{\partial\Omega_{u_{n}}}{\partial\theta_{u_{n}}}F_{u_{j+1}}^{u_{n}}a_{u_{j}}.
\]
Then we directly take derivatives from the definition of $c$, which
gives

\begin{align*}
\frac{\partial c_{u_{\ell}}}{\partial\theta_{u_{n}}} & =\frac{\partial c_{u_{\ell+1}}}{\partial\theta_{u_{n}}}-2w_{u_{\ell}}^{T}\frac{\partial\gamma_{u_{\ell+1}}}{\partial\theta_{u_{n}}}+w_{u_{\ell}}^{T}\frac{\partial\Omega_{u_{\ell+1}}}{\partial\theta_{u_{n}}}w_{u_{\ell}}-\frac{\partial}{\partial\theta_{u_{n}}}\left(b_{u_{\ell}}^{T}\Lambda_{u_{\ell}}b_{u_{\ell}}\right)\\
 & =\frac{\partial c_{u_{\ell+1}}}{\partial\theta_{u_{n}}}-2w_{u_{\ell}}^{T}\frac{\partial\gamma_{u_{\ell+1}}}{\partial\theta_{u_{n}}}+w_{u_{\ell}}^{T}\frac{\partial\Omega_{u_{\ell+1}}}{\partial\theta_{u_{n}}}w_{u_{\ell}}-\left(\frac{\partial\gamma_{u_{\ell+1}}}{\partial\theta_{u_{n}}}-\frac{\partial\Omega_{u_{\ell+1}}}{\partial\theta_{u_{n}}}w_{u_{\ell}}\right)^{T}\Lambda_{u_{\ell}}b_{u_{\ell}}\\
 & \quad\quad\quad\quad-b_{u_{\ell}}^{T}\frac{\partial\Lambda_{u_{\ell}}}{\partial\theta_{u_{n}}}b_{u_{\ell}}-b_{u_{\ell}}^{T}\Lambda_{u_{\ell}}\left(\frac{\partial\gamma_{u_{\ell+1}}}{\partial\theta_{u_{n}}}-\frac{\partial\Omega_{u_{\ell+1}}}{\partial\theta_{u_{n}}}w_{u_{\ell}}\right)\\
 & =\frac{\partial c_{u_{\ell+1}}}{\partial\theta_{u_{n}}}-\left(a_{u_{\ell}}^{T}\frac{\partial\gamma_{u_{\ell+1}}}{\partial\theta_{u_{n}}}+\frac{\partial\gamma_{u_{\ell+1}}}{\partial\theta_{u_{n}}}^{T}a_{u_{\ell}}\right)+w_{u_{\ell}}^{T}\frac{\partial\Omega_{u_{\ell+1}}}{\partial\theta_{u_{n}}}w_{u_{\ell}}+w_{u_{\ell}}^{T}\frac{\partial\Omega_{u_{\ell+1}}}{\partial\theta_{u_{n}}}^{T}\Lambda_{u_{\ell}}b_{u_{\ell}}\\
 & \quad\quad\quad\quad+b_{u_{\ell}}^{T}\Lambda_{u_{\ell}}\frac{\partial\Omega_{u_{\ell+1}}}{\partial\theta_{u_{n}}}w_{u_{\ell}}+b_{u_{\ell}}^{T}\Lambda_{u_{\ell}}\frac{\partial\Omega_{u_{\ell+1}}}{\partial\theta_{u_{n}}}\Lambda_{u_{\ell}}b_{u_{\ell}}\\
 & =\frac{\partial c_{u_{\ell+1}}}{\partial\theta_{u_{n}}}-2a_{u_{\ell}}^{T}\frac{\partial\gamma_{u_{\ell+1}}}{\partial\theta_{u_{n}}}+a_{u_{\ell}}^{T}\frac{\partial\Omega_{u_{\ell+1}}}{\partial\theta_{u_{n}}}a_{u_{\ell}}\\
 & =\frac{\partial c_{u_{\ell+1}}}{\partial\theta_{u_{n}}}-2a_{u_{\ell}}^{T}F_{u_{\ell+1}}^{u_{n}T}\left(\frac{\partial\gamma_{u_{n}}}{\partial\theta_{u_{n}}}-\frac{\partial\Omega_{u_{n}}}{\partial\theta_{u_{n}}}g_{u_{\ell+1}}^{u_{n}}\right)+a_{u_{\ell}}^{T}F_{u_{\ell+1}}^{u_{n}T}\frac{\partial\Omega_{u_{n}}}{\partial\theta_{u_{n}}}F_{u_{\ell+1}}^{u_{n}}a_{u_{\ell}}\\
 & =\frac{\partial c_{u_{\ell+1}}}{\partial\theta_{u_{n}}}-a_{u_{\ell}}^{T}F_{u_{\ell+1}}^{u_{n}T}\left(2\frac{\partial\gamma_{u_{n}}}{\partial\theta_{u_{n}}}-2\frac{\partial\Omega_{u_{n}}}{\partial\theta_{u_{n}}}g_{u_{\ell+1}}^{u_{n}}-\frac{\partial\Omega_{u_{n}}}{\partial\theta_{u_{n}}}F_{u_{\ell+1}}^{u_{n}}a_{u_{\ell}}\right)\\
 & =\frac{\partial c_{u_{\ell+1}}}{\partial\theta_{u_{n}}}-\left(g_{u_{\ell}}^{u_{n}T}-g_{u_{\ell+1}}^{u_{n}T}\right)\left[2\frac{\partial\gamma_{u_{n}}}{\partial\theta_{u_{n}}}-\frac{\partial\Omega_{u_{n}}}{\partial\theta_{u_{n}}}\left(g_{u_{\ell}}^{u_{n}}+g_{u_{\ell+1}}^{u_{n}}\right)\right].
\end{align*}

Moving the first term to the left hand side, we have an expression
for the difference between $\frac{\partial c_{u_{\ell}}}{\partial\theta_{u_{n}}}$
and $\frac{\partial c_{u_{\ell+1}}}{\partial\theta_{u_{n}}}$, i.e.,
for any $\ell$, it follows that
\begin{align*}
\frac{\partial c_{u_{\ell}}}{\partial\theta_{u_{n}}}-\frac{\partial c_{u_{\ell+1}}}{\partial\theta_{u_{n}}} & =\left(g_{u_{\ell}}^{u_{n}T}-g_{u_{\ell+1}}^{u_{n}T}\right)\frac{\partial\Omega_{u_{n}}}{\partial\theta_{u_{n}}}\left(g_{u_{\ell}}^{u_{n}}+g_{u_{\ell+1}}^{u_{n}}\right)-2\left(g_{u_{\ell}}^{u_{n}T}-g_{u_{\ell+1}}^{u_{n}T}\right)\frac{\partial\gamma_{u_{n}}}{\partial\theta_{u_{n}}}\\
 & =g_{u_{\ell}}^{u_{n}T}\frac{\partial\Omega_{u_{n}}}{\partial\theta_{u_{n}}}g_{u_{\ell}}^{u_{n}}-g_{u_{\ell+1}}^{u_{n}T}\frac{\partial\Omega_{u_{n}}}{\partial\theta_{u_{n}}}g_{u_{\ell+1}}^{u_{n}}-2\left(g_{u_{\ell}}^{u_{n}T}-g_{u_{\ell+1}}^{u_{n}T}\right)\frac{\partial\gamma_{u_{n}}}{\partial\theta_{u_{n}}}\\
 & =\sum_{i=\ell}^{n-1}\sum_{j=\ell}^{n-1}a_{u_{i}}^{T}F_{u_{i+1}}^{u_{n}T}\frac{\partial\Omega_{u_{n}}}{\partial\theta_{u_{n}}}F_{u_{j+1}}^{u_{n}}a_{u_{j}}-\sum_{i=\ell+1}^{n-1}\sum_{j=\ell+1}^{n-1}a_{u_{i}}^{T}F_{u_{i+1}}^{u_{n}T}\frac{\partial\Omega_{u_{n}}}{\partial\theta_{u_{n}}}F_{u_{j+1}}^{u_{n}}a_{u_{j}}\\
 & \quad\quad\quad\quad-2\left(g_{u_{\ell}}^{u_{n}T}-g_{u_{\ell+1}}^{u_{n}T}\right)\frac{\partial\gamma_{u_{n}}}{\partial\theta_{u_{n}}}\\
 & =\sum_{j=\ell}^{n-1}a_{u_{\ell}}^{T}F_{u_{\ell+1}}^{u_{n}T}\frac{\partial\Omega_{u_{n}}}{\partial\theta_{u_{n}}}F_{u_{j+1}}^{u_{n}}a_{u_{j}}+\sum_{i=\ell+1}^{n-1}a_{u_{i}}^{T}F_{u_{i+1}}^{u_{n}T}\frac{\partial\Omega_{u_{n}}}{\partial\theta_{u_{n}}}F_{u_{\ell+1}}^{u_{n}}a_{u_{\ell}}-2a_{u_{\ell}}^{T}F_{u_{\ell+1}}^{u_{n}T}\frac{\partial\gamma_{u_{n}}}{\partial\theta_{u_{n}}}\\
 & =\sum_{j=\ell}^{n-1}M_{\ell j}+\sum_{i=\ell+1}^{n-1}M_{i\ell}-2a_{u_{\ell}}^{T}F_{u_{\ell+1}}^{u_{n}T}\frac{\partial\gamma_{u_{n}}}{\partial\theta_{u_{n}}}
\end{align*}
which leads to
\begin{align*}
\frac{\partial c_{u_{k}}}{\partial\theta_{u_{n}}}-\frac{\partial c_{u_{n}}}{\partial\theta_{u_{n}}} & =\sum_{\ell=k}^{n-1}\left(\frac{\partial c_{u_{\ell}}}{\partial\theta_{u_{n}}}-\frac{\partial c_{u_{\ell+1}}}{\partial\theta_{u_{n}}}\right)\\
 & =\sum_{\ell=k}^{n-1}\sum_{j=\ell}^{n-1}M_{\ell j}+\sum_{\ell=k}^{n-1}\sum_{i=\ell+1}^{n-1}M_{i\ell}-2\sum_{\ell=k}^{n-1}a_{u_{\ell}}^{T}F_{u_{\ell+1}}^{u_{n}T}\frac{\partial\gamma_{u_{n}}}{\partial\theta_{u_{n}}}\\
 & =\sum_{\ell=k}^{n-1}\sum_{j=\ell}^{n-1}M_{\ell j}+\sum_{j=k}^{n-1}\sum_{\ell=j+1}^{n-1}M_{\ell j}-2g_{u_{k}}^{u_{n}T}\frac{\partial\gamma_{u_{n}}}{\partial\theta_{u_{n}}}
\end{align*}
In the above, the first term sums over the index set $I_{1}=\left\{ \left(\ell,j\right)\in\mathbb{N}:k\le\ell\le j\le n-1\right\} $
and the second term sums over the set $I_{2}=\left\{ \left(\ell,j\right)\in\mathbb{N}:n-1\ge\ell>j\ge k\right\} $.
Here $I_{1}$ and $I_{2}$ are disjoint and their union covers the
entire set $\left\{ k\cdots n-1\right\} \times\left\{ k\cdots n-1\right\} $
because either one of $\ell\le j$ or $\ell>j$ must be true. Therefore
the first two terms can be rewritten as 
\[
\sum_{\ell=k}^{n-1}\sum_{j=\ell}^{n-1}M_{\ell j}+\sum_{j=k}^{n-1}\sum_{\ell=j+1}^{n-1}M_{\ell j}=\sum_{i=k}^{n-1}\sum_{j=k}^{n-1}M_{ij}=\left(\sum_{i=k}^{n-1}a_{u_{i}}^{T}F_{u_{i+1}}^{u_{n}T}\right)\frac{\partial\Omega_{u_{n}}}{\partial\theta_{u_{n}}}\left(\sum_{j=k}^{n-1}F_{u_{j+1}}^{u_{n}}a_{u_{j}}\right)=g_{u_{k}}^{u_{n}T}\frac{\partial\Omega_{u_{n}}}{\partial\theta_{u_{n}}}g_{u_{k}}^{u_{n}}.
\]
So we could conclude that
\begin{align*}
\frac{\partial c_{u_{k}}}{\partial\theta_{u_{n}}} & =\frac{\partial c_{u_{n}}}{\partial\theta_{u_{n}}}+g_{u_{k}}^{u_{n}T}\frac{\partial\Omega_{u_{n}}}{\partial\theta_{u_{n}}}g_{u_{k}}^{u_{n}}-2g_{u_{k}}^{u_{n}T}\frac{\partial\gamma_{u_{n}}}{\partial\theta_{u_{n}}},
\end{align*}
hence \eqref{dchnc} is proved.

Finally, to prove \eqref{dchnDelta}, we differentiate the definition
directly, leading to
\begin{align*}
\frac{\partial\Delta_{u_{\ell}}}{\partial\theta_{u_{n}}} & =\frac{\partial\Delta_{u_{\ell+1}}}{\partial\theta_{u_{n}}}-\frac{\partial}{\partial\theta_{u_{n}}}\ln\det\Lambda_{u_{\ell}}\\
 & =\frac{\partial\Delta_{u_{\ell+1}}}{\partial\theta_{u_{n}}}-\mathrm{tr}\left\{ \Lambda_{\ell}^{-T}\frac{\partial\Delta_{u_{\ell}}}{\partial\theta_{u_{n}}}\right\} \\
 & =\frac{\partial\Delta_{u_{\ell+1}}}{\partial\theta_{u_{n}}}-\mathrm{tr}\left\{ \Lambda_{\ell}^{-1}\left[-\Lambda_{u_{\ell}}\frac{\partial\Omega_{u_{\ell+1}}}{\partial\theta_{u_{n}}}\Lambda_{u_{\ell}}\right]\right\} \\
 & =\frac{\partial\Delta_{u_{\ell+1}}}{\partial\theta_{u_{n}}}+\mathrm{tr}\left\{ F_{u_{\ell+1}}^{u_{n}T}\frac{\partial\Omega_{u_{n}}}{\partial\theta_{u_{n}}}F_{u_{\ell+1}}^{u_{n}}\Lambda_{u_{\ell}}\right\} \\
 & =\frac{\partial\Delta_{u_{\ell+1}}}{\partial\theta_{u_{n}}}+\mathrm{tr}\left\{ \frac{\partial\Omega_{u_{n}}}{\partial\theta_{u_{n}}}F_{u_{\ell+1}}^{u_{n}}\Lambda_{u_{\ell}}F_{u_{\ell+1}}^{u_{n}T}\right\} 
\end{align*}
for any $1\le\ell\le n-1$. The above allows us to write $\frac{\partial\Delta_{u_{k}}}{\partial\theta_{u_{n}}}$
as a sum:
\[
\frac{\partial\Delta_{u_{k}}}{\partial\theta_{u_{n}}}-\frac{\partial\Delta_{u_{n}}}{\partial\theta_{u_{n}}}=\sum_{\ell=k}^{n-1}\left(\frac{\partial\Delta_{u_{\ell}}}{\partial\theta_{u_{n}}}-\frac{\partial\Delta_{u_{\ell+1}}}{\partial\theta_{u_{n}}}\right)=\mathrm{tr}\left\{ \frac{\partial\Omega_{u_{n}}}{\partial\theta_{u_{n}}}\sum_{\ell=k}^{n-1}F_{u_{\ell+1}}^{u_{n}}\Lambda_{u_{\ell}}F_{u_{\ell+1}}^{u_{n}T}\right\} =\mathrm{tr}\left\{ \frac{\partial\Omega_{u_{n}}}{\partial\theta_{u_{n}}}K_{u_{k}}^{u_{n}T}\right\} .
\]
Adding $\frac{\partial\Delta_{u_{n}}}{\partial\theta_{u_{n}}}$ to
both sides results in the required equation.
\end{proof}

\section*{Proof of Proposition 4}
\begin{proof}
Proposition 2 implies that $\Omega_{\eta}$ does not depend on $w_{\eta}$
and $c_{\eta}$ does not depend on $\Phi_{\eta}$, nor does $\Delta_{\eta}$
on $\Phi_{\eta}$ and $w_{\eta}$; thus these partial derivatives
vanish, i.e., \eqref{dxdxselfvanish} holds. To prove the rest of
the equations we take divide the proof into two cases: when $\eta$
is a tip and when it is not.

\textbf{When $\eta$ is a tip.} First, we prove the case where $\eta$
is a tip by directly differentiating $\Omega_{\eta}$, $\gamma_{\eta}$,
$c_{\eta}$ and $\Delta_{\eta}$ from definition (\eqref{cu}-(\ref{eq:Deltau})).
Note that now in \eqref{cu}-(\ref{eq:Deltau}) we have $\Psi_{\eta}=\Phi_{\eta}$,
$q_{\eta}=w_{\eta}$ and $\Sigma_{\eta}=V_{\eta}$.

Because 
\[
\frac{\partial V_{\eta}^{-1}}{\partial\left(V_{\eta}\right)_{ij}}=-V_{\eta}^{-1}e_{i}e_{j}^{T}V_{\eta}^{-1}=G_{\eta}^{\left(ij\right)}=M_{\eta}G_{\eta}^{\left(ij\right)}M_{\eta},
\]
\eqref{dodvself}-(\ref{eq:dcdvself}) can be easily seen as follows
\begin{gather*}
\frac{\partial\Omega_{\eta}}{\partial\left(V_{\eta}\right)_{ij}}=\Phi_{\eta}^{T}\left(\frac{\partial V_{\eta}^{-1}}{\partial\left(V_{\eta}\right)_{ij}}\right)\Phi_{\eta}=\Phi_{\eta}^{T}M_{\eta}G_{\eta}^{\left(ij\right)}M_{\eta}\Phi_{\eta};\\
\frac{\partial\gamma_{\eta}}{\partial\left(V_{\eta}\right)_{ij}}=\left[\Phi_{\eta}^{T}\frac{\partial V_{\eta}^{-1}}{\partial\left(V_{\eta}\right)_{ij}}\left(x_{\eta}-w_{\eta}\right)\right]=\Phi_{\eta}^{T}M_{\eta}G_{\eta}^{\left(ij\right)}\xi_{\eta};\mathrm{~and}\\
\frac{\partial c_{\eta}}{\partial\left(V_{\eta}\right)_{ij}}=\left(x_{\eta}-w_{\eta}\right)^{T}\frac{\partial V_{\eta}^{-1}}{\partial\left(V_{\eta}\right)_{ij}}\left(x_{\eta}-w_{\eta}\right)=\xi_{\eta}^{T}G_{\eta}^{\left(ij\right)}\xi_{\eta}.
\end{gather*}

\eqref{dddvself}-(\ref{eq:dgdwself}) can be derived as follows.
\begin{gather*}
\frac{\partial\Delta_{\eta}}{\partial V_{\eta}}=\frac{\partial}{\partial V_{\eta}}\ln\det V_{\eta}=\left(V_{\eta}^{-1}\right)^{T}=V_{\eta}^{-1}=D_{\eta}\\
\frac{\partial\Omega_{\eta}}{\partial\left(\Phi_{\eta}\right)_{ij}}=\Phi_{\eta}^{T}V_{\eta}^{-1}\frac{\partial\Phi_{\eta}}{\partial\left(\Phi_{\eta}\right)_{ij}}+\left[\frac{\partial\Phi_{\eta}}{\partial\left(\Phi_{\eta}\right)_{ij}}\right]^{T}V_{\eta}^{-1}\Phi_{\eta}=\Phi_{\eta}^{T}M_{\eta}Z_{\eta}e_{i}e_{j}^{T}+e_{j}^{T}e_{i}M_{\eta}Z_{\eta}\Phi_{\eta}\\
\frac{\partial\gamma_{\eta}}{\partial\left(\Phi_{\eta}\right)_{ij}}=\frac{\partial\Phi_{\eta}^{T}}{\partial\left(\Phi_{\eta}\right)_{ij}}V_{\eta}^{-1}\left(x_{\eta}-w_{\eta}\right)=e_{j}\left(e_{i}^{T}Z_{\eta}\xi_{\eta}\right)=\left[\left(Z_{\eta}e_{i}\right)^{T}\xi_{\eta}\right]e_{j}\\
\frac{\partial\gamma_{\eta}}{\partial\left(w_{\eta}\right)_{i}}=\Phi_{\eta}^{T}V_{\eta}^{-1}\left(\frac{\partial x_{\eta}}{\partial\left(w_{\eta}\right)_{i}}-\frac{\partial w_{\eta}}{\partial\left(w_{\eta}\right)_{i}}\right)=\Phi_{\eta}^{T}Z_{\eta}^{T}M_{\eta}\left(0-e_{i}\right)=-\Phi_{\eta}^{T}Z_{\eta}^{T}M_{\eta}e_{i}
\end{gather*}

Finally, \eqref{dcdwself} holds because, due to symmetricity of $V_{\eta}^{-1}$,
\[
\frac{\partial c_{\eta}}{\partial w_{\eta}}=\frac{\partial}{\partial w_{\eta}}\left[\left(x_{\eta}-w_{\eta}\right)^{T}V_{\eta}^{-1}\left(x_{\eta}-w_{\eta}\right)\right]=-2V_{\eta}^{-1}\left(x_{\eta}-w_{\eta}\right)=-2Z_{\eta}\xi_{\eta}.
\]

\textbf{When $\eta$ is not a tip.} It is useful to notice that by
\eqref{deflambda} we have
\begin{equation}
\frac{\partial\Lambda_{\eta}}{\partial\left(V_{\eta}\right)_{ij}}=\frac{\partial\left(V_{\eta}^{-1}+\sum_{\iota\in\children\left(\eta\right)}\Omega_{\iota}\right)^{-1}}{\partial\left(V_{\eta}\right)_{ij}}=-\Lambda_{\eta}\frac{\partial V_{\eta}^{-1}}{\partial\left(V_{\eta}\right)_{ij}}\Lambda_{\eta}=-\Lambda_{\eta}\frac{\partial V_{\eta}^{-1}}{\partial\left(V_{\eta}\right)_{ij}}\Lambda_{\eta}=\Lambda_{\eta}V_{\eta}^{-1}e_{i}e_{j}^{T}V_{\eta}^{-1}\Lambda_{\eta}=-G_{\eta}^{\left(ij\right)}.\label{eq:dLambdadv}
\end{equation}
Furthermore, using \eqref{adef} we have
\begin{align*}
\frac{\partial H_{\eta}}{\partial\left(V_{\eta}\right)_{ij}} & =\frac{\partial I_{k_{\eta}}}{\partial\left(V_{\eta}\right)_{ij}}-\frac{\partial\Lambda_{\eta}}{\partial\left(V_{\eta}\right)_{ij}}\sum_{\iota\in\children(\eta)}\Omega_{\iota}=0-\left(-G_{\eta}^{\left(ij\right)}M_{\eta}\right)=G_{\eta}^{\left(ij\right)}M_{\eta}\mathrm{\ and}
\end{align*}

Now we differentiate \eqref{curecur}-(\ref{eq:Deltaurecur}) with
respect to $V_{\eta}$, which gives us
\begin{gather}
\frac{\partial\Omega_{\eta}}{\partial\left(V_{\eta}\right)_{ij}}=\Phi_{\eta}^{T}M_{\eta}\frac{\partial H_{\eta}}{\partial\left(V_{\eta}\right)_{ij}}\Phi_{\eta}=\Phi_{\eta}^{T}M_{\eta}G_{\eta}^{\left(ij\right)}M_{\eta}\Phi_{\eta};\nonumber \\
\frac{\partial\gamma_{\eta}}{\partial\left(V_{\eta}\right)_{ij}}=\Phi_{\eta}^{T}\frac{\partial H_{\eta}^{T}}{\partial\left(V_{\eta}\right)_{ij}}b_{\eta}=\Phi_{\eta}^{T}M_{\eta}G_{\eta}^{\left(ij\right)}\xi_{\eta};\nonumber \\
\frac{\partial c_{\eta}}{\partial\left(V_{\eta}\right)_{ij}}=0-2\cdot0+0-b_{\eta}^{T}\frac{\partial\Lambda_{\eta}}{\partial\left(V_{\eta}\right)_{ij}}b_{\eta}=\xi_{\eta}^{T}G_{\eta}^{\left(ij\right)}\xi_{\eta};\label{eq:sidcdv}
\end{gather}
and
\begin{align*}
\frac{\partial\Delta_{\eta}}{\partial\left(V_{\eta}\right)_{ij}} & =0+\frac{\partial\ln\det V_{\eta}}{\partial\left(V_{\eta}\right)_{ij}}+\frac{\partial\ln\det\Lambda_{\eta}^{-1}}{\partial\left(V_{\eta}\right)_{ij}}\\
 & =\left(V_{\eta}^{-1}\right)_{ij}+\mathrm{tr}\left[\Lambda_{\eta}\frac{\partial\Lambda_{\eta}^{-1}}{\partial\left(V_{\eta}\right)_{ij}}\right]\\
 & =\left(V_{\eta}^{-1}\right)_{ij}+\mathrm{tr}\left\{ \Lambda_{\eta}\left[-\Lambda_{\eta}^{-1}\left(-\Lambda_{\eta}V_{\eta}^{-1}e_{i}e_{j}^{T}V_{\eta}^{-1}\Lambda_{\eta}\right)\Lambda_{\eta}^{-1}\right]\right\} \\
 & =\left(V_{\eta}^{-1}\right)_{ij}+\mathrm{tr}\left(\Lambda_{\eta}V_{\eta}^{-1}e_{i}e_{j}^{T}V_{\eta}^{-1}\right)\\
 & =\left(V_{\eta}^{-1}\right)_{ij}+\mathrm{tr}\left(e_{j}^{T}V_{\eta}^{-1}\Lambda_{\eta}V_{\eta}^{-1}e_{i}\right)\\
 & =\left[V_{\eta}^{-1}+\left(V_{\eta}^{-1}\Lambda_{\eta}V_{\eta}^{-1}\right)^{T}\right]_{ij}\\
 & =\left[V_{\eta}^{-1}+V_{\eta}^{-1}\Lambda_{\eta}V_{\eta}^{-1}\right]_{ij}\\
 & =\left[D_{\eta}\right]_{ij}.
\end{align*}
where the last equality holds because both $V_{\eta}$ and $\Lambda_{\eta}$
are symmetric (the latter is true because of \eqref{deflambda} and
(\ref{eq:Omegau})). Finally, the followings show \eqref{dodphiself}-(\ref{eq:dcdwself}):
\begin{align}
\frac{\partial\Omega_{\eta}}{\partial\left(\Phi_{\eta}\right)_{ij}} & =\Phi_{\eta}^{T}M_{\eta}H_{\eta}\frac{\partial\Phi_{\eta}}{\partial\left(\Phi_{\eta}\right)_{ij}}+\frac{\partial\Phi_{\eta}^{T}}{\partial\left(\Phi_{\eta}\right)_{ij}}M_{\eta}H_{\eta}\Phi_{\eta}=\Phi_{\eta}^{T}M_{\eta}Z_{\eta}e_{i}e_{j}^{T}+e_{j}e_{i}^{T}M_{\eta}Z_{\eta}\Phi_{\eta};\nonumber \\
\frac{\partial\gamma_{\eta}}{\partial\left(\Phi_{\eta}\right)_{ij}} & =\frac{\partial\Phi_{\eta}^{T}}{\partial\left(\Phi_{\eta}\right)_{ij}}H_{\eta}^{T}b_{\eta}=e_{j}e_{i}^{T}H_{\eta}^{T}b_{\eta}=e_{j}\left(e_{i}^{T}Z_{\eta}^{T}\xi_{\eta}\right)=\left[\left(Z_{\eta}^{T}e_{i}\right)^{T}\xi_{\eta}\right]e_{j};\nonumber \\
\frac{\partial\gamma_{\eta}}{\partial\left(w_{\eta}\right)_{i}} & =\Phi_{\eta}^{T}H_{\eta}^{T}\frac{\partial b_{\eta}}{\partial\left(w_{\eta}\right)_{i}}=\Phi_{\eta}^{T}H_{\eta}^{T}\left[-\left(\sum_{\iota\in\children\left(\eta\right)}\Omega_{\iota}\right)e_{i}\right]=-\Phi_{\eta}^{T}Z_{\eta}^{T}M_{\eta}e_{i};\nonumber \\
\frac{\partial c_{\eta}}{\partial w_{\eta}} & =\frac{\partial}{\partial\left(w_{\eta}\right)_{i}}\left[\sum_{\iota\in\children\left(\eta\right)}c_{\iota}-2\sum_{\iota\in\children\left(\eta\right)}\gamma_{\iota}^{T}w_{\eta}+w_{\eta}^{T}\left(\sum_{\iota\in\children\left(\eta\right)}\Omega_{\iota}\right)w_{\eta}-b_{\eta}^{T}\Lambda_{\eta}b_{\eta}\right]\label{eq:sidcdw1}\\
 & =0-2\left(\sum_{\iota\in\children\left(\eta\right)}\gamma_{\iota}\right)+2\left(\sum_{\iota\in\children\left(\eta\right)}\Omega_{\iota}\right)w_{\eta}\label{eq:sidcdw2}\\
 & \quad\quad\quad-\left[-2\left(\sum_{\iota\in\children\left(\eta\right)}\Omega_{\iota}\right)^{T}\Lambda_{\eta}\left(\sum_{\iota\in\children\left(\eta\right)}\gamma_{\iota}-\left(\sum_{\iota\in\children\left(\eta\right)}\Omega_{\iota}\right)w_{\eta}\right)\right]\\
 & =2\left[-I+\left(\sum_{\iota\in\children\left(\eta\right)}\Omega_{\iota}\right)\Lambda_{\eta}\right]\left(\sum_{\iota\in\children\left(\eta\right)}\gamma_{\iota}\right)+2\left[I-\left(\sum_{\iota\in\children\left(\eta\right)}\Omega_{\iota}\right)\Lambda_{\eta}\right]\left(\sum_{\iota\in\children\left(\eta\right)}\Omega_{\iota}\right)w_{\eta}\nonumber \\
 & =2H_{\eta}^{T}\left[\left(\sum_{\iota\in\children\left(\eta\right)}\Omega_{\iota}\right)w_{\eta}-\sum_{\iota\in\children\left(\eta\right)}\gamma_{\iota}\right]\nonumber \\
 & =-2H_{\eta}^{T}\xi_{\eta}\label{eq:sidcdw}\\
 & =-2Z_{\eta}^{T}\xi_{\eta}.\nonumber 
\end{align}
\end{proof}

\section*{Proof of Proposition 5}
\begin{proof}
It is trivial to verify that the second-order partial derivatives
not mentioned in Proposition 5 are all zero. Notice that regardless
of whether $\eta$ is a tip or not, we have 
\begin{align*}
\frac{\partial^{2}G_{\eta}^{\left(ij\right)}}{\partial\left(V_{\eta}\right)_{ij}\partial\left(V_{\eta}\right)_{pq}} & =\tilde{G}_{\eta}^{\left(ijpq\right)},
\end{align*}
because
\begin{align*}
\frac{\partial^{2}G_{\eta}^{\left(ij\right)}}{\partial\left(V_{\eta}\right)_{ij}\partial\left(V_{\eta}\right)_{pq}} & =\frac{\partial V_{\eta}^{-1}}{\partial\left(V_{\eta}\right)_{pq}}e_{i}e_{j}^{T}V_{\eta}^{-1}+V_{\eta}^{-1}e_{i}e_{j}^{T}\frac{\partial V_{\eta}^{-1}}{\partial\left(V_{\eta}\right)_{pq}}\\
 & =-V_{\eta}^{-1}e_{p}e_{q}^{T}V_{\eta}^{-1}e_{i}e_{j}^{T}V_{\eta}^{-1}-V_{\eta}^{-1}e_{i}e_{j}^{T}V_{\eta}^{-1}e_{p}e_{q}^{T}V_{\eta}^{-1}\\
 & =\tilde{G}_{\eta}^{\left(ijpq\right)}
\end{align*}
if $\eta$ were a tip, otherwise
\begin{align*}
\frac{\partial G_{\eta}^{\left(ij\right)}}{\partial\left(V_{\eta}\right)_{pq}} & =\frac{\partial\Lambda_{\eta}}{\partial\left(V_{\eta}\right)_{pq}}V_{\eta}^{-1}e_{i}e_{j}^{T}V_{\eta}^{-1}\Lambda_{\eta}+\Lambda_{\eta}\frac{\partial V_{\eta}^{-1}}{\partial\left(V_{\eta}\right)_{pq}}e_{i}e_{j}^{T}V_{\eta}^{-1}\Lambda_{\eta}\\
 & \quad\quad\quad\quad+\Lambda_{\eta}V_{\eta}^{-1}e_{i}e_{j}^{T}\frac{\partial V_{\eta}^{-1}}{\partial\left(V_{\eta}\right)_{pq}}\Lambda_{\eta}+\Lambda_{\eta}V_{\eta}^{-1}e_{i}e_{j}^{T}V_{\eta}^{-1}\frac{\partial\Lambda_{\eta}}{\partial\left(V_{\eta}\right)_{pq}}\\
 & =\Lambda_{\eta}V_{\eta}^{-1}e_{p}e_{q}^{T}V_{\eta}^{-1}\Lambda_{\eta}V_{\eta}^{-1}e_{i}e_{j}^{T}V_{\eta}^{-1}\Lambda_{\eta}-\Lambda_{\eta}V_{\eta}^{-1}e_{p}e_{q}^{T}V_{\eta}^{-1}e_{i}e_{j}^{T}V_{\eta}^{-1}\Lambda_{\eta}\\
 & \quad\quad\quad\quad-\Lambda_{\eta}V_{\eta}^{-1}e_{i}e_{j}^{T}V_{\eta}^{-1}e_{p}e_{q}^{T}V_{\eta}^{-1}\Lambda_{\eta}-\Lambda_{\eta}V_{\eta}^{-1}e_{i}e_{j}^{T}V_{\eta}^{-1}\Lambda_{\eta}V_{\eta}^{-1}e_{p}e_{q}^{T}V_{\eta}^{-1}\Lambda_{\eta}\\
 & =\Lambda_{\eta}V_{\eta}^{-1}e_{p}e_{q}^{T}\left(V_{\eta}^{-1}\Lambda_{\eta}V_{\eta}^{-1}-V_{\eta}^{-1}\right)e_{i}e_{j}^{T}V_{\eta}^{-1}\Lambda_{\eta}\\
 & \quad\quad\quad\quad+\Lambda_{\eta}V_{\eta}^{-1}e_{i}e_{j}^{T}\left(V_{\eta}^{-1}\Lambda_{\eta}V_{\eta}^{-1}-V_{\eta}^{-1}\right)e_{p}e_{q}^{T}V_{\eta}^{-1}\Lambda_{\eta}\\
 & =\tilde{G}_{\eta}^{\left(ijpq\right)}.
\end{align*}
Also it is trivial to verify that 
\[
\frac{\partial\xi_{\eta}}{\partial\left(w_{\eta}\right)_{p}}=-M_{\eta}e_{p}
\]
regardless of whether $\eta$ is a tip. Therefore in either case we
have
\begin{align}
\frac{\partial^{2}\Omega_{\eta}}{\partial\left(V_{\eta}\right)_{ij}\partial\left(V_{\eta}\right)_{pq}} & =\Phi_{\eta}^{T}M_{\eta}\frac{\partial G_{\eta}^{\left(ij\right)}}{\partial\left(V_{\eta}\right)_{pq}}M_{\eta}\Phi_{\eta}=\Phi_{\eta}^{T}M_{\eta}\tilde{G}_{\eta}^{\left(ijpq\right)}M_{\eta}\Phi_{\eta}.\nonumber \\
\frac{\partial^{2}\Omega_{\eta}}{\partial\left(V_{\eta}\right)_{ij}\partial\left(\Phi_{\eta}\right)_{pq}} & =\frac{\partial\Phi_{\eta}^{T}}{\partial\left(\Phi_{\eta}\right)_{pq}}M_{\eta}G_{\eta}^{\left(ij\right)}M_{\eta}\Phi_{\eta}+\Phi_{\eta}M_{\eta}G_{\eta}^{\left(ij\right)}M_{\eta}\frac{\partial\Phi_{\eta}}{\partial\left(\Phi_{\eta}\right)_{pq}}\nonumber \\
 & =e_{q}e_{p}^{T}M_{\eta}G_{\eta}^{\left(ij\right)}M_{\eta}\Phi_{\eta}+\Phi_{\eta}M_{\eta}G_{\eta}^{\left(ij\right)}M_{\eta}e_{p}e_{q}^{T}\\
\frac{\partial^{2}\Omega_{\eta}}{\partial\left(\Phi_{\eta}\right)_{ij}\partial\left(\Phi_{\eta}\right)_{pq}} & =\frac{\partial\Phi_{\eta}^{T}}{\partial\left(\Phi_{\eta}\right)_{pq}}M_{\eta}Z_{\eta}e_{i}e_{j}^{T}+e_{j}e_{i}^{T}M_{\eta}Z_{\eta}\frac{\partial\Phi_{\eta}}{\partial\left(\Phi_{\eta}\right)_{pq}}=e_{q}e_{p}^{T}M_{\eta}Z_{\eta}e_{i}e_{j}^{T}+e_{j}e_{i}^{T}M_{\eta}Z_{\eta}e_{p}e_{q}^{T}\nonumber \\
\frac{\partial^{2}\gamma_{\eta}}{\partial\left(V_{\eta}\right)_{ij}\partial\left(V_{\eta}\right)_{pq}} & =\Phi_{\eta}^{T}M_{\eta}\frac{\partial G_{\eta}^{\left(ij\right)}}{\partial\left(V_{\eta}\right)_{pq}}\xi_{\eta}=\Phi_{\eta}^{T}M_{\eta}\tilde{G}_{\eta}^{\left(ijpq\right)}\xi_{\eta}\nonumber \\
\frac{\partial^{2}\gamma_{\eta}}{\partial\left(V_{\eta}\right)_{ij}\partial\left(\Phi_{\eta}\right)_{pq}} & =\frac{\partial\Phi_{\eta}}{\partial\left(\Phi_{\eta}\right)_{pq}}M_{\eta}G_{\eta}^{\left(ij\right)}\xi_{\eta}=e_{p}e_{q}^{T}M_{\eta}G_{\eta}^{\left(ij\right)}\xi_{\eta}\nonumber \\
\frac{\partial^{2}\gamma_{\eta}}{\partial\left(V_{\eta}\right)_{ij}\partial\left(w_{\eta}\right)_{p}} & =\Phi_{\eta}M_{\eta}G_{\eta}^{\left(ij\right)}\frac{\partial\xi_{\eta}}{\partial\left(w_{\eta}\right)_{p}}=-\Phi_{\eta}M_{\eta}G_{\eta}^{\left(ij\right)}M_{\eta}e_{p}\nonumber \\
\frac{\partial^{2}\gamma_{\eta}}{\partial\left(\Phi_{\eta}\right)_{ij}\partial\left(w_{\eta}\right)_{p}} & =-\frac{\partial\Phi_{\eta}^{T}}{\partial\left(\Phi_{\eta}\right)_{ij}}Z_{\eta}^{T}M_{\eta}e_{p}=-e_{j}e_{i}^{T}Z_{\eta}^{T}M_{\eta}e_{p}=-\left[\left(M_{\eta}Z_{\eta}e_{i}\right)^{T}e_{p}\right]e_{j}\nonumber \\
\frac{\partial^{2}c_{\eta}}{\partial\left(V_{\eta}\right)_{ij}\partial\left(V_{\eta}\right)_{pq}} & =\xi_{\eta}^{T}\frac{\partial G_{\eta}^{\left(ij\right)}}{\partial\left(V_{\eta}\right)_{pq}}\xi_{\eta}=\xi_{\eta}^{T}\tilde{G}_{\eta}^{\left(ijpq\right)}\xi_{\eta}\nonumber \\
\frac{\partial^{2}c_{\eta}}{\partial\left(V_{\eta}\right)_{ij}\partial\left(w_{\eta}\right)_{p}} & =\frac{\partial\xi_{\eta}^{T}G_{\eta}^{\left(ij\right)}\xi_{\eta}}{\partial\left(w_{\eta}\right)_{p}}=\frac{\partial\xi_{\eta}^{T}}{\partial\left(w_{\eta}\right)_{p}}G_{\eta}^{\left(ij\right)}\xi_{\eta}+\xi_{\eta}^{T}G_{\eta}^{\left(ij\right)}\frac{\partial\xi_{\eta}}{\partial\left(w_{\eta}\right)_{p}}=-e_{p}^{T}M_{\eta}G_{\eta}^{\left(ij\right)}\xi_{\eta}-\xi_{\eta}^{T}G_{\eta}^{\left(ij\right)}M_{\eta}e_{p}\label{eq:sihcdvdw}
\end{align}
Finally, \eqref{hddvdv} holds because if $\eta$ is a tip then
\[
\frac{\partial^{2}\Delta_{\eta}}{\partial\left(V_{\eta}\right)_{ij}\partial\left(V_{\eta}\right)_{pq}}=e_{i}^{T}\frac{\partial V_{\eta}^{-1}}{\partial\left(V_{\eta}\right)_{pq}}e_{j}=-e_{i}^{T}V_{\eta}^{-1}e_{p}e_{q}^{T}V_{\eta}^{-1}e_{j}=-e_{i}^{T}D_{\eta}e_{p}e_{q}^{T}D_{\eta}e_{j};
\]
otherwise 
\begin{align*}
\frac{\partial^{2}\Delta_{\eta}}{\partial\left(V_{\eta}\right)_{ij}\partial\left(V_{\eta}\right)_{pq}} & =e_{i}^{T}\left[\frac{\partial V_{\eta}^{-1}}{\partial\left(V_{\eta}\right)_{pq}}-\frac{\partial V_{\eta}^{-1}}{\partial\left(V_{\eta}\right)_{pq}}\Lambda_{\eta}V_{\eta}^{-1}-V_{\eta}^{-1}\frac{\partial\Lambda_{\eta}}{\partial\left(V_{\eta}\right)_{pq}}V_{\eta}^{-1}-V_{\eta}^{-1}\Lambda_{\eta}\frac{\partial V_{\eta}^{-1}}{\partial\left(V_{\eta}\right)_{pq}}\right]e_{j}\\
 & =e_{i}^{T}\left(\begin{aligned}-V_{\eta}^{-1}e_{p}e_{q}^{T}V_{\eta}^{-1}+V_{\eta}^{-1}e_{p}e_{q}^{T}V_{\eta}^{-1}\Lambda_{\eta}V_{\eta}^{-1}\quad\quad\quad\quad\quad\quad\quad\quad\\
-V_{\eta}^{-1}\Lambda_{\eta}V_{\eta}^{-1}e_{p}e_{q}^{T}V_{\eta}^{-1}\Lambda_{\eta}V_{\eta}^{-1}+V_{\eta}^{-1}\Lambda_{\eta}V_{\eta}^{-1}e_{p}e_{q}^{T}V_{\eta}^{-1}
\end{aligned}
\right)e_{j}\\
 & =e_{i}^{T}\left[\begin{aligned}\left(V_{\eta}^{-1}\Lambda_{\eta}V_{\eta}^{-1}-V_{\eta}^{-1}\right)e_{p}e_{q}^{T}V_{\eta}^{-1}+V_{\eta}^{-1}e_{p}e_{q}^{T}\left(V_{\eta}^{-1}-D_{\eta}\right)\quad\quad\quad\\
-\left(V_{\eta}^{-1}-D_{\eta}\right)e_{p}e_{q}^{T}\left(V_{\eta}^{-1}-D_{\eta}\right)
\end{aligned}
\right]e_{j}\\
 & =e_{i}^{T}\left[-D_{\eta}e_{p}e_{q}^{T}V_{\eta}^{-1}+V_{\eta}^{-1}e_{p}e_{q}^{T}\left(V_{\eta}^{-1}-D_{\eta}\right)-\left(V_{\eta}^{-1}-D_{\eta}\right)e_{p}e_{q}^{T}\left(V_{\eta}^{-1}-D_{\eta}\right)\right]e_{j}\\
 & =e_{i}^{T}\left[-D_{\eta}e_{p}e_{q}^{T}V_{\eta}^{-1}+D_{\eta}e_{p}e_{q}^{T}\left(V_{\eta}^{-1}-D_{\eta}\right)\right]e_{j}\\
 & =e_{i}^{T}\left[D_{\eta}e_{p}e_{q}^{T}\left(-V_{\eta}^{-1}+V_{\eta}^{-1}-D_{\eta}\right)\right]e_{j}\\
 & =-e_{i}^{T}D_{\eta}e_{p}e_{q}^{T}D_{\eta}e_{j}.
\end{align*}

\end{proof}

\section*{Proof of Proposition 6}
\begin{proof}
If we substitute $u_{k}=u_{1}$ and $\eta=\eta_{2}$ into \eqref{dchnOmega}-(\ref{eq:dchnDelta})
and apply $\partial/\partial\theta_{\eta_{1}}$ to both sides of the
equations we will have
\begin{align*}
\hrooteta{\Omega}= & \detaone{F_{u_{1}}^{\eta_{2}}}^{T}\detatwo{\Omega_{\eta_{2}}}F_{u_{1}}^{\eta_{2}}+{F_{u_{1}}^{\eta_{2}}}^{T}\detatwo{\Omega_{\eta_{2}}}\detaone{F_{u_{1}}^{\eta_{2}}}+{F_{u_{1}}^{\eta_{2}}}^{T}\frac{\partial^{2}\Omega_{\eta_{2}}}{\partial\theta_{\eta_{1}}\partial\theta_{\eta_{2}}}\detaone{F_{u_{1}}^{\eta_{2}}}\\
\hrooteta{\gamma}= & \detaone{F_{u_{1}}^{\eta_{2}}}^{T}\left(\detatwo{\gamma_{\eta_{2}}}-\detatwo{\Omega_{\eta_{2}}}g_{u_{1}}^{\eta_{2}}\right)+{F_{u_{1}}^{\eta_{2}}}^{T}\left[\frac{\partial^{2}\gamma_{\eta_{2}}}{\partial\theta_{\eta_{1}}\theta_{\eta_{2}}}-\detatwo{\Omega_{\eta_{2}}}\detaone{g_{u_{1}}^{\eta_{2}}}-\frac{\partial^{2}\Omega_{\eta_{2}}}{\partial\theta_{\eta_{1}}\partial\theta_{\eta_{2}}}g_{u_{1}}^{\eta_{2}}\right]\\
\hrooteta c= & \frac{\partial^{2}c_{\eta_{2}}}{\partial\theta_{\eta_{1}}\partial\theta_{\eta_{2}}}+\detaone{g_{u_{1}}^{\eta_{2}}}^{T}\left(\detatwo{\Omega_{\eta_{2}}}g_{u_{1}}^{\eta_{2}}-2\detatwo{\gamma_{\eta_{2}}}\right)\\
 & \quad\quad\quad\quad+{g_{u_{1}}^{\eta_{2}}}^{T}\left(\detatwo{\Omega_{\eta_{2}}}\detaone{g_{u_{1}}^{\eta_{2}}}+\frac{\partial^{2}\Omega_{\eta_{2}}}{\partial\theta_{\eta_{1}}\partial\theta_{\eta_{2}}}g_{u_{1}}^{\eta_{2}}-2\frac{\partial^{2}\gamma_{\eta_{2}}}{\partial\theta_{\eta_{1}}\theta_{\eta_{2}}}\right);\quad\text{and}\\
\hrooteta{\Delta}= & \frac{\partial^{2}\Delta_{\eta_{2}}}{\partial\theta_{\eta_{1}}\partial\theta_{\eta_{2}}}+\text{tr}\left(\detatwo{\Omega_{\eta_{2}}}\detaone{K_{u_{1}}^{\eta_{2}}}+\frac{\partial^{2}\Omega_{\eta_{2}}}{\partial\theta_{\eta_{1}}\partial\theta_{\eta_{2}}}K_{u_{1}}^{\eta_{2}}\right).
\end{align*}
Now notice that $\frac{\partial^{2}\Omega_{\eta_{2}}}{\partial\theta_{\eta_{1}}\partial\theta_{\eta_{2}}}$,
$\frac{\partial^{2}\gamma_{\eta_{2}}}{\partial\theta_{\eta_{1}}\theta_{\eta_{2}}}$,
$\frac{\partial^{2}c_{\eta_{2}}}{\partial\theta_{\eta_{1}}\partial\theta_{\eta_{2}}}$,
and $\frac{\partial^{2}\Delta_{\eta_{2}}}{\partial\theta_{\eta_{1}}\partial\theta_{\eta_{2}}}$
vanishes no matter whether the relationship between $\eta_{1}$ and
$\eta_{2}$ is in Case 1 or 4, because according to \eqref{curecur}-(\ref{eq:Deltaurecur})
$\left(\Omega_{\eta_{2}},\gamma_{\eta_{2}},c_{\eta_{2}},\Delta_{\eta_{2}}\right)$
can only contain $\theta_{\eta_{1}}$ if $\eta_{1}$ were either $\eta_{2}$
itself or a descendant of it, which is a configuration that is excluded
in Case 1 nor 4 by definition. After setting them to zero \eqref{hodtdt}-(\ref{eq:hddtdt})
follows.
\end{proof}

\section*{Proof of Proposition 7}
\begin{proof}
We additionally let $u_{K+1}$ to be further nodes down the lineage
that goes from the root to $\eta_{2}$ in Case 4, if they exists.
Notice that in Case 1, both $u_{K+1}$ and $u_{K+2}$ must exists
by defintion while in Case 4 there is only $u_{K+1}$, which could
be just $\eta_{2}$ if the lowest common ancestor $u_{K}$ of $\eta_{1}$
and $\eta_{2}$ were $\eta_{2}$'s mother.

Notice that we can write $F_{u_{1}}^{\eta_{2}}=F_{u_{K+1}}^{\eta_{2}}F_{u_{1}}^{u_{K+1}}$
by definition. In Case 4, $F_{u_{K+1}}^{\eta_{2}}$ cannot be dependent
on $\theta_{\eta_{1}}$ as $u_{K+1}$ is not a ancestor of $\eta_{2}$;
while in Case 1 this is not true, as $u_{K+1}=\eta_{1}$ always hold
in Case 1 by definition. Therefore by product rule of calculus we
have
\begin{align}
\frac{\partial F_{u_{1}}^{\eta_{2}}}{\partial\theta_{\eta_{1}}} & =1_{\textrm{Ch}^{*}\left(\eta_{2}\right)}\left(\eta_{1}\right)\left(\frac{\partial F_{\eta_{1}}^{\eta_{2}}}{\partial\theta_{\eta_{1}}}\right)F_{u_{1}}^{\eta_{1}}+F_{u_{K+1}}^{\eta_{2}}\frac{\partial F_{u_{1}}^{u_{K+1}}}{\partial\theta_{\eta_{1}}}\label{eq:dfdanything}
\end{align}
where $1_{\textrm{Ch}^{*}\left(\eta_{2}\right)}\left(\eta_{1}\right)$
is one in Case 1 and zero in Case 4. We can expand $\frac{\partial F_{u_{1}}^{u_{K+1}}}{\partial\theta_{\eta_{1}}}$
using the facts that
\begin{align}
\frac{\partial H_{u_{i}}}{\partial\theta_{u_{j}}} & =\frac{\partial}{\partial\theta_{u_{j}}}\left[I-\Lambda_{u_{i}}\sum_{\eta\in\children\left(u_{i}\right)}\Omega_{\eta}\right]\nonumber \\
 & =-\left(\frac{\partial\Lambda_{u_{i}}}{\partial\theta_{u_{j}}}\right)\sum_{\eta\in\children\left(u_{i}\right)}\Omega_{\eta}-\Lambda_{u_{i}}\frac{\partial}{\partial\theta_{u_{j}}}\sum_{\eta\in\children\left(u_{i}\right)}\Omega_{\eta}\nonumber \\
 & =\Lambda_{u_{i}}\frac{\partial\Omega_{u_{i+1}}}{\partial\theta_{u_{j}}}\Lambda_{u_{i}}\sum_{\eta\in\children\left(u_{i}\right)}\Omega_{\eta}-\Lambda_{u_{i}}\frac{\partial\Omega_{u_{i+1}}}{\partial\theta_{u_{j}}}\nonumber \\
 & =\Lambda_{u_{i}}F_{u_{i+1}}^{u_{j}T}\frac{\partial\Omega_{u_{j}}}{\partial\theta_{u_{j}}}F_{u_{i+1}}^{u_{j}}\left(\Lambda_{u_{i}}\sum_{\eta\in\children\left(u_{i}\right)}\Omega_{\eta}-I\right)\nonumber \\
 & =-\Lambda_{u_{i}}F_{u_{i+1}}^{u_{j}T}\frac{\partial\Omega_{u_{j}}}{\partial\theta_{u_{j}}}F_{u_{i+1}}^{u_{j}}H_{u_{i}}\label{eq:dH}
\end{align}
as follows
\begin{align*}
\frac{\partial F_{u_{1}}^{u_{K+1}}}{\partial\theta_{\eta_{1}}} & =\frac{\partial}{\partial\theta_{\eta_{1}}}\left(H_{u_{K}}\Phi_{u_{K}}H_{u_{K-1}}\Phi_{u_{K-1}}\cdots H_{u_{1}}\Phi_{u_{1}}\right)\\
 & =\sum_{1\le k\le K}F_{u_{k+1}}^{u_{K+1}}\frac{\partial H_{u_{k}}}{\partial\theta_{\eta_{1}}}\Phi_{u_{k}}F_{u_{1}}^{u_{k}}\\
 & =-\sum_{1\le k\le K}F_{u_{k+1}}^{u_{K+1}}\Lambda_{u_{k}}F_{u_{k+1}}^{\eta_{1}T}\frac{\partial\Omega_{\eta_{1}}}{\partial\theta_{\eta_{1}}}F_{u_{k+1}}^{u_{j}}H_{u_{k}}\Phi_{u_{k}}F_{u_{1}}^{u_{k}}\\
 & =-\left(\sum_{1\le k\le K}F_{u_{k+1}}^{u_{K+1}}\Lambda_{u_{k}}F_{u_{k+1}}^{\eta_{1}T}\right)\frac{\partial\Omega_{\eta_{1}}}{\partial\theta_{\eta_{1}}}F_{u_{1}}^{\eta_{1}}.
\end{align*}
Plugging these back into \eqref{dfdanything}, we will have
\begin{align}
\frac{\partial F_{u_{1}}^{\eta_{2}}}{\partial\theta_{\eta_{1}}} & =1_{\textrm{Ch}^{*}\left(\eta_{2}\right)}\left(\eta_{1}\right)\left(\frac{\partial F_{\eta_{1}}^{\eta_{2}}}{\partial\theta_{\eta_{1}}}\right)F_{u_{1}}^{\eta_{1}}-\left(\sum_{1\le k\le K}F_{u_{k+1}}^{\eta_{2}}\Lambda_{u_{k}}F_{u_{k+1}}^{\eta_{1}T}\right)\frac{\partial\Omega_{\eta_{1}}}{\partial\theta_{\eta_{1}}}F_{u_{1}}^{\eta_{1}};\label{eq:dfdtu1}
\end{align}
therefore (\ref{eq:dfdtheta}) has been proven.

In addition because during the above derivation we have never assumed
$u_{1}$ to be actually among one of the direct children of the root
node, nothing would change had we replace $u_{1}$ with $u_{i}$ in
\eqref{dfdtu1} and amend the lower bound of summation to $i$ accordingly.
This will help us derive the formulae for $\detaone{g_{u_{1}}^{\eta_{2}}}$,
which we will discuss immediately below.

Consider differentiating \eqref{gdef} term by term to obtain $\detaone{g_{u_{1}}^{\eta_{2}}}$.
Regardless of whether we are in Case 1 or Case 4, by definition (\eqref{gdef}),
the vector $g_{u_{1}}^{\eta_{2}}$ is the sum of $F_{u_{k+1}}^{\eta_{2}}a_{u_{k}}$
along the lineage leading from $u_{1}$ to $\eta_{2}$. But as before,
we know that $F_{u_{k+1}}^{\eta_{2}}$ does not depend on $\theta_{\eta_{1}}$
unless $k\le K$ and similarly $a_{u_{k}}$ does not depend on $\theta_{\eta_{1}}$
unless $k\le K+1$ (in fact, $\frac{\partial a_{u_{k}}}{\partial\theta_{\eta_{1}}}$
also vanishes when $k=K+1$ in Case 4 although not in Case 1); this
means that we only need to consider part of the sum $\sum_{1\le k\le K+1}F_{u_{k+1}}^{\eta_{2}}a_{u_{k}}$
when differentiating, as the differentials of all other terms down
the lineage toward $\eta_{2}$ would vanish. These allow us to write
\begin{align}
\frac{\partial g_{u_{1}}^{\eta_{2}}}{\partial\theta_{\eta_{1}}} & =\sum_{1\le k\le K}\frac{\partial F_{u_{k+1}}^{\eta_{2}}}{\partial\theta_{\eta_{1}}}a_{u_{k}}+\sum_{1\le k\le K+1}F_{u_{k+1}}^{\eta_{2}}\frac{\partial a_{u_{k}}}{\partial\theta_{\eta_{1}}}\nonumber \\
 & =\sum_{1\le k\le K-1}\frac{\partial F_{u_{k+1}}^{\eta_{2}}}{\partial\theta_{\eta_{1}}}a_{u_{k}}+1_{\textrm{Ch}^{*}\left(\eta_{2}\right)}\left(\eta_{1}\right)\frac{\partial F_{\eta_{1}}^{\eta_{2}}}{\partial\theta_{\eta_{1}}}a_{u_{K}}+\sum_{1\le k\le K}F_{u_{k+1}}^{\eta_{2}}\frac{\partial a_{u_{k}}}{\partial\theta_{\eta_{1}}}+1_{\textrm{Ch}^{*}\left(\eta_{2}\right)}\left(\eta_{1}\right)F_{u_{K+2}}^{\eta_{2}}\frac{\partial a_{\eta_{1}}}{\partial\theta_{\eta_{1}}}\nonumber \\
 & =1_{\textrm{Ch}^{*}\left(\eta_{2}\right)}\left(\eta_{1}\right)\left[\frac{\partial F_{\eta_{1}}^{\eta_{2}}}{\partial\theta_{\eta_{1}}}a_{\parent\left(\eta_{1}\right)}+F_{u_{K+2}}^{\eta_{2}}\frac{\partial a_{\eta_{1}}}{\partial\theta_{\eta_{1}}}\right]\\
 & \quad\quad+\sum_{1\le k\le K-1}\left[1_{\textrm{Ch}^{*}\left(\eta_{2}\right)}\left(\eta_{1}\right)\frac{\partial F_{\eta_{1}}^{\eta_{2}}}{\partial\theta_{\eta_{1}}}F_{u_{k+1}}^{\eta_{1}}-\left(\sum_{k+1\le\ell\le K}F_{u_{\ell+1}}^{\eta_{2}}\Lambda_{u_{\ell}}F_{u_{\ell+1}}^{\eta_{1}T}\right)\frac{\partial\Omega_{\eta_{1}}}{\partial\theta_{\eta_{1}}}F_{u_{k+1}}^{\eta_{1}}\right]a_{u_{k}}\label{eq:dgdttmp01-1}\\
 & \quad\quad+\sum_{1\le k\le K}F_{u_{k+1}}^{\eta_{2}}\frac{\partial a_{u_{k}}}{\partial\theta_{\eta_{1}}}.
\end{align}
The only thing remaining that has not been solved for is $F_{u_{k+1}}^{\eta_{2}}\frac{\partial a_{u_{k}}}{\partial\theta_{\eta_{1}}}$
in the last term. Because $u_{k}$ is never the same node as $\eta_{1}$
in Case 1 and 4 by definition, nor can it be one of the latter's ancestors,
by plugging in \eqref{dLambdadany-1}, (\ref{eq:dH}), and (\ref{eq:dchngamma})
into \eqref{adef} we can obtain the expression for $F_{u_{k+1}}^{\eta_{2}}\frac{\partial a_{u_{k}}}{\partial\theta_{\eta_{1}}}$
as follows.
\begin{align*}
F_{u_{k+1}}^{\eta_{2}}\frac{\partial a_{u_{k}}}{\partial\theta_{\eta_{1}}} & =F_{u_{k+1}}^{\eta_{2}}\left[\frac{\partial H_{u_{k}}}{\partial\theta_{\eta_{1}}}w_{u_{k}}+\frac{\partial\Lambda_{u_{k}}}{\partial\theta_{\eta_{1}}}\sum_{\iota\in\children\left(u_{k}\right)}\gamma_{\iota}+\Lambda_{u_{k}}\frac{\partial}{\partial\theta_{\eta_{1}}}\sum_{\iota\in\children\left(u_{k}\right)}\gamma_{\iota}\right]\\
 & =-F_{u_{k+1}}^{\eta_{2}}\left[\begin{aligned}\Lambda_{u_{k}}F_{u_{k+1}}^{\eta_{1}T}\frac{\partial\Omega_{\eta}}{\partial\theta_{\eta_{1}}}F_{u_{k+1}}^{\eta_{1}}H_{u_{k}}w_{u_{k}}+\Lambda_{u_{k}}F_{u_{k+1}}^{u_{n}T}\frac{\partial\Omega_{u_{n}}}{\partial\theta_{u_{n}}}F_{u_{k+1}}^{u_{n}}\Lambda_{u_{k}}\sum_{\iota\in\children\left(u_{k}\right)}\gamma_{\iota}\quad\quad\quad\\
+\Lambda_{u_{k}}F_{u_{k+1}}^{\eta_{1}T}\left(\frac{\partial\gamma_{\eta_{1}}}{\partial\theta_{\eta_{1}}}-\frac{\partial\Omega_{\eta_{1}}}{\partial\theta_{\eta_{1}}}g_{u_{k+1}}^{\eta}\right)
\end{aligned}
\right]\\
 & =-F_{u_{k+1}}^{\eta_{2}}\Lambda_{u_{k}}F_{u_{k+1}}^{\eta_{1}T}\left\{ \frac{\partial\Omega_{\eta}}{\partial\theta_{\eta_{1}}}\left[F_{u_{k+1}}^{\eta_{1}}\left(H_{u_{k}}w_{u_{k}}+\Lambda_{u_{k}}\sum_{\iota\in\children\left(u_{k}\right)}\gamma_{\iota}\right)+g_{u_{k+1}}^{\eta_{1}}\right]-\frac{\partial\gamma_{\eta_{1}}}{\partial\theta_{\eta_{1}}}\right\} \\
 & =-F_{u_{k+1}}^{\eta_{2}}\Lambda_{u_{k}}F_{u_{k+1}}^{\eta_{1}T}\left[\frac{\partial\Omega_{\eta}}{\partial\theta_{\eta_{1}}}\left(F_{u_{k+1}}^{\eta_{1}}a_{u_{k}}+g_{u_{k+1}}^{\eta_{1}}\right)-\frac{\partial\gamma_{\eta_{1}}}{\partial\theta_{\eta_{1}}}\right]\\
 & =F_{u_{k+1}}^{\eta_{2}}\Lambda_{u_{k}}F_{u_{k+1}}^{\eta_{1}T}\left[\frac{\partial\gamma_{\eta_{1}}}{\partial\theta_{\eta_{1}}}-\frac{\partial\Omega_{\eta}}{\partial\theta_{\eta_{1}}}g_{u_{k}}^{\eta_{1}}\right].
\end{align*}
Plugging this back into \eqref{dgdttmp01-1} immediately gives \eqref{gdef}.

Now we consider $\frac{\partial K_{u_{1}}^{\eta_{2}}}{\partial\theta_{\eta_{1}}}$,
i.e., \eqref{dkdtheta}. Similar to the situation of $\frac{\partial g_{u_{1}}^{\eta_{2}}}{\partial\theta_{\eta_{1}}}$,
the matrix $K_{u_{1}}^{\eta_{2}}$ can be written as a sum along a
lineage. If we differentiate directly from the definition using the
product rule of differentation, $\frac{\partial}{\partial\theta_{\eta_{1}}}F_{u_{k+1}}^{\eta_{2}}$
and $\frac{\partial\Lambda_{u_{k}}}{\partial\theta_{\eta_{1}}}$ arises.
Let
\[
K^{*}=\begin{cases}
K+1 & \textrm{if case (1);}\\
K & \textrm{if case (4).}
\end{cases}
\]
 First, notice that for $k\ge K^{*}$ we have $\frac{\partial}{\partial\theta_{\eta_{1}}}F_{u_{k+1}}^{\eta_{2}}=0$
and on the other hand for $1\le k\le K^{*}-1$, according to the discussion
around \eqref{dfdtu1}, it follows that
\begin{align*}
\left(\frac{\partial}{\partial\theta_{\eta_{1}}}F_{u_{k+1}}^{\eta_{2}}\right)\Lambda_{u_{k}}F_{u_{k+1}}^{\eta_{1}T} & =\left[\begin{aligned}1_{\textrm{Ch}^{*}\left(\eta_{2}\right)}\left(\eta_{1}\right)\left(\frac{\partial F_{\eta_{1}}^{\eta_{2}}}{\partial\theta_{\eta_{1}}}\right)F_{u_{k+1}}^{\eta_{1}}\quad\quad\quad\quad\quad\quad\quad\quad\\
\quad\quad-\left(\sum_{k+1\le\ell\le K}F_{u_{\ell+1}}^{\eta_{2}}\Lambda_{u_{\ell}}F_{u_{\ell+1}}^{\eta_{1}T}\right)\frac{\partial\Omega_{\eta_{1}}}{\partial\theta_{\eta_{1}}}F_{u_{k+1}}^{\eta_{1}}
\end{aligned}
\right]\Lambda_{u_{k}}F_{u_{k+1}}^{\eta_{1}T}\\
 & =A_{k}
\end{align*}
and anologously $F_{u_{k+1}}^{\eta_{2}}\Lambda_{u_{k}}\left(\frac{\partial}{\partial\theta_{\eta_{1}}}F_{u_{k+1}}^{\eta_{1}T}\right)=A_{k}^{T}$
for $1\le k\le K^{*}-1$. On the other hand, by definition $\frac{\partial\Lambda_{u_{k}}}{\partial\theta_{\eta_{1}}}=0$
for $k\ge K^{*}+1$ and for $k\le K^{*}$, according to \eqref{dLambdadany-1}
and \eqref{dchnOmega}, we have
\begin{align*}
F_{u_{k+1}}^{\eta_{2}}\frac{\partial\Lambda_{u_{k}}}{\partial\theta_{\eta_{1}}}F_{u_{k+1}}^{\eta_{1}T} & =-F_{u_{k+1}}^{\eta_{1}}\Lambda_{u_{k}}\frac{\partial\Omega_{u_{k+1}}}{\partial\theta_{\eta_{1}}}\Lambda_{u_{k}}F_{u_{k+1}}^{\eta_{1}T}\\
 & =-F_{u_{k+1}}^{\eta_{1}}\Lambda_{u_{k}}F_{u_{k+1}}^{\eta_{1}T}\frac{\partial\Omega_{\eta_{1}}}{\partial\theta_{\eta_{1}}}F_{u_{k+1}}^{\eta_{1}}\Lambda_{u_{k-1}}F_{u_{k+1}}^{\eta_{1}T}\\
 & =B_{k}.
\end{align*}
 Using these, we arrive at 
\begin{align*}
\frac{\partial K_{u_{1}}^{\eta_{2}}}{\partial\theta_{\eta_{1}}} & =\sum_{1\le k\le K+1}\frac{\partial}{\partial\theta_{\eta_{1}}}\left(F_{u_{k+1}}^{\eta_{2}}\Lambda_{u_{k}}F_{u_{k+1}}^{\eta_{1}T}\right)\\
 & =\sum_{1\le k\le K^{*}}\left(\frac{\partial}{\partial\theta_{\eta_{1}}}F_{u_{k+1}}^{\eta_{2}}\right)\Lambda_{u_{k}}F_{u_{k+1}}^{\eta_{1}T}+\sum_{1\le k\le K^{*}}F_{u_{k+1}}^{\eta_{2}}\Lambda_{u_{k}}\left(\frac{\partial}{\partial\theta_{\eta_{1}}}F_{u_{k+1}}^{\eta_{1}T}\right)\\
 & \quad\quad\quad+\sum_{1\le k\le K+1}F_{u_{k+1}}^{\eta_{2}}\left(\frac{\partial}{\partial\theta_{\eta_{1}}}\Lambda_{u_{k}}\right)F_{u_{k+1}}^{\eta_{1}T}\\
 & =\sum_{1\le k\le K^{*}-1}\left(A_{k}+A_{k}^{T}\right)+\sum_{1\le k\le K^{*}}B_{k}.
\end{align*}
The above expression is just another way of writing the required \eqref{dkdtheta}
because $K^{*}$ here simply means adding one to the upper bounds
of the sums, which is expressed in \eqref{dkdtheta} using a indicator
function.
\end{proof}

\section*{Proof of Proposition 9}

In this section we will make heavy use of directional derivatives
of a multivariate function. We denote the directional derivative with
respect to the direction $V$ as $\D_{V}f(X)$ when it is clear from
the context which variable we are differentiating. When both $X$
and $V$ are $p\times q$ matrix, we treat them as simple vectors
of $pq$ numeric values, so that $\D_{U^{ij}}f(X)=\frac{\partial f\left(x\right)}{\partial X_{ij}}$
(recall that in the manuscript we defined $U^{ij}$=$e_{i}e_{j}^{T}$,
i.e., the matrix whose elements are all zero except that its $\left(i,j\right)$
entry is one). Before we start, notice that we will use the following
identity multiple times (Najfeld and Havel 1995; see manuscript for
full reference):

\begin{equation}
\D_{V}e^{tA}=\int_{0}^{t}e^{\left(t-\tau\right)A}Ve^{\tau A}\d\tau.\label{eq:nhdir}
\end{equation}

First, we derive the expression for $\frac{\partial V_{\eta}}{\partial H_{ij}}.$
If we expand the $V_{\eta}$ into its matrix integral form according
to definition and directly differentiate it, we would arrive at this:

\begin{align}
\frac{\partial V_{\eta}}{\partial H_{ij}} & =\frac{\partial}{\partial H_{ij}}\int_{0}^{t}e^{-vH}\Sigma e^{-vH^{T}}\d v\nonumber \\
 & =\int_{0}^{t}\left(\frac{\partial}{\partial H_{ij}}e^{-vH}\right)\Sigma e^{-vH^{T}}+e^{-vH}\Sigma\left(\frac{\partial}{\partial H_{ij}}e^{-vH}\right)^{T}\d v.\label{eq:p901}
\end{align}
If we treat the partial derivative $\frac{\partial}{\partial H_{ij}}e^{-vH}$
above as an directional derivative with respect to $U^{ij}$ and apply
\eqref{nhdir} we get 
\begin{equation}
\frac{\partial}{\partial H_{ij}}e^{-vH}=-e^{-vH}\int_{0}^{v}e^{\tau H}U^{ij}e^{-\tau H}\d\tau.\label{eq:devdhij01}
\end{equation}
Letting $J_{ij}\left(v\right)=\int_{0}^{v}e^{\tau H}U^{ij}e^{-\tau H}\d\tau$,
i.e., the integral above, we can rearrange \eqref{p901} as
\begin{align}
\frac{\partial V_{\eta}}{\partial H_{ij}} & =\int_{0}^{t}-e^{-vH}J_{ij}\left(v\right)\Sigma e^{-vH^{T}}\d v+\left(\int_{0}^{t}-e^{-vH}J_{ij}\left(v\right)\Sigma e^{-vH^{T}}\d v\right)^{T}\\
 & =-\mathcal{S}\left(\int_{0}^{t}e^{-vH}J_{ij}\left(v\right)\Sigma e^{-vH^{T}}\d v\right),\label{eq:p9011-2}
\end{align}
where the integral in the RHS can be integrated by expressing the
integrand as a Hadamard product involving $H$'s eigen-decomposition
$H=P\Lambda P^{-1}$, as follows
\begin{align}
\int_{0}^{t}e^{-vH}J_{ij}\left(v\right)\Sigma e^{-vH^{T}}\d v & =P\int_{0}^{t}e^{-v\Lambda}P^{-1}J_{ij}\left(v\right)\Sigma P^{-T}e^{-v\Lambda}\d vP^{T}\nonumber \\
 & =P\int_{0}^{t}\left(P^{-1}J_{ij}\left(v\right)\Sigma P^{-T}\right)\odot\left[e^{-\left(\lambda_{k}+\lambda_{\ell}\right)v}\right]_{k\ell}\d vP^{T}.\label{eq:p902}
\end{align}

Since this is a double integral, we first write $J_{ij}\left(v\right)$
explicitly. This is easy $J_{ij}\left(v\right)$ using the same Hadamard
product trick, which gives
\begin{align*}
J_{ij}\left(v\right) & =P\int_{0}^{v}\left(P^{-1}U^{ij}P\right)\odot\left[e^{\left(\lambda_{k}-\lambda_{\ell}\right)\tau}\right]_{k\ell}\d\tau P^{-1}\\
 & =P\left\{ \left(P^{-1}U^{ij}P\right)\odot\left[\int_{0}^{v}e^{\left(\lambda_{k}-\lambda_{\ell}\right)\tau}\d\tau\right]_{k\ell}\right\} P^{-1}\\
 & =P\left\{ \overline{U^{ij}}\odot\left[I_{0}\left(\lambda_{k}-\lambda_{\ell},v\right)\right]_{k\ell}\right\} P^{-1},
\end{align*}
which, if plugged back into \eqref{p902}, its RHS become
\[
P\int_{0}^{t}\left\{ \left\{ \overline{U^{ij}}\odot\left[I_{0}\left(\lambda_{k}-\lambda_{\ell},v\right)\right]_{k\ell}\right\} \overline{\Sigma}\right\} \odot\left[e^{-\left(\lambda_{k}+\lambda_{\ell}\right)v}\right]_{k\ell}\d vP^{T}.
\]
The last step is to evaluate the matrix integral in the above elementwise.
The $\left(m,n\right)$ element of this matrix integral is
\[
\sum_{k}\overline{U^{ij}}_{mk}\overline{\Sigma}_{kn}\int_{0}^{t}I_{0}\left(\lambda_{m}-\lambda_{k},v\right)e^{-\left(\lambda_{m}+\lambda_{n}\right)v}\d v=\sum_{p}\overline{U^{ij}}_{mk}\overline{\Sigma}_{kn}K_{0}\left(\lambda_{m}-\lambda_{k},-\lambda_{m}-\lambda_{n},t\right)=M_{mn}.
\]
Therefore, plugging this back into \eqref{p902} we have 
\[
\int_{0}^{t}e^{-vH}J_{ij}\left(v\right)\Sigma e^{-vH^{T}}\d v=PMP^{T};
\]
so we can conclude from \eqref{p9011-2} that $\frac{\partial V_{\eta}}{\partial H_{ij}}=-\mathcal{S}\left(PMP^{T}\right)$.

On the other hand, to get the expression of $\frac{\partial V_{\eta}}{\partial L_{ij}}$,
notice that
\[
\frac{\partial V_{\eta}}{\partial L_{ij}}=\int_{0}^{t}e^{-Hv}\frac{\partial LL^{T}}{\partial L_{ij}}e^{-H^{T}v},
\]
and 
\begin{align*}
\frac{\partial LL^{T}}{\partial L_{ij}} & =\frac{\partial L}{\partial L_{ij}}L^{T}+L\frac{\partial L^{T}}{\partial L_{ij}}\\
 & =\mathcal{S}\left(\frac{\partial L}{\partial L_{ij}}L^{T}\right)\\
 & =\mathcal{S}\left(U^{ij}L^{T}\right),
\end{align*}
therefore \eqref{dvdlij} follows.

The derivative of $\frac{\partial w_{\eta}}{\partial H_{ij}}$ and
$\frac{\partial w_{\eta}}{\partial\mu_{i}}$ in \eqref{dwdhij} and
\eqref{dwdmui} can be seen trivially from \eqref{wOU}.

Finally $\frac{\partial\Phi_{\eta}}{\partial H_{ij}}$ (\eqref{dphidhij})
is obtained by using \eqref{nhdir} again because when $\Phi_{\eta}$
is seen as a function of $\frac{\partial\Phi_{\eta}}{\partial H_{ij}}=\D_{U^{ij}}e^{-Ht_{\eta}}$.
So
\begin{align*}
\frac{\partial\Phi_{\eta}}{\partial H_{ij}} & =\D_{U^{ij}}e^{-Ht_{\eta}}\\
 & =-\int_{0}^{t_{\eta}}e^{\left(\tau-t_{\eta}\right)H}U^{ij}e^{-\tau H}\d\tau\\
 & =-P\int_{0}^{t_{\eta}}\left(P^{-1}U^{ij}P\right)\odot\left[e^{\left(\tau-t_{\eta}\right)\lambda_{i}-\tau\lambda_{j}}\right]_{ij}\d\tau P^{-1}\\
 & =-P\left\{ \left(P^{-1}U^{ij}P\right)\odot\left[e^{-t_{\eta}\lambda_{i}}\int_{0}^{t_{\eta}}e^{\left(\lambda_{i}-\lambda_{j}\right)\tau}\d\tau\right]_{ij}\right\} P^{-1}\\
 & =-P\left\{ \overline{U^{ij}}\odot\left[e^{-t_{\eta}\lambda_{i}}I_{0}\left(\lambda_{i}-\lambda_{j},t_{\eta}\right)\right]_{ij}\right\} P^{-1},
\end{align*}
as required. 

\section*{Proof of Proposition 10}

First, we introduce a simple ``change-of-basis'' lemma that expresses
a set of double partial derivatives in double directional derivatives
in directions that are easier to work with. Same as before, when we
talk about basis and directional derivatives, we operate in the vector
space of $n\times n$ matrices whose addition and multiplication operators
are simple matrix addition and scalar-matrix multiplication.

We denote such matrix basis by $\left\{ U_{\alpha}\right\} _{\alpha\in\left\{ 1\cdots n\right\} ^{2}}$,
where the square of a set refers to the Cartesian product of the set
with itself. or simply $\left\{ U_{\alpha}\right\} $ when the index
range is clear from the context.
\begin{lem}
\label{lem:basishess}If $\left\{ U_{\alpha}\right\} _{\alpha\in\left\{ 1\cdots n\right\} ^{2}}$
is a basis of the space of $n$-dimensional square matrices, $P$
is an $n\times n$ invertible matrix, and $G_{\alpha}=PU_{\alpha}P^{-1}$,
then $\left\{ G_{\alpha}\right\} _{\alpha\in\left\{ 1\cdots n\right\} ^{2}}$
is also a basis. Furthermore, if $\left\{ U_{\alpha}\right\} $ is
the standard basis, i.e., the set of $U^{ij}$, then the directional
derivative operator $\D_{U_{\alpha}}$ can be expressed as
\[
\D_{U_{\alpha}}=\sum_{\beta\in\left\{ 1\cdots n\right\} ^{2}}\left(P^{-1}U_{\alpha}P\right)_{\beta}\D_{G_{\beta}}
\]
for all $\alpha$.
\end{lem}

\begin{proof}
To see that $\left\{ G_{\alpha}\right\} $ is indeed linearly independent,
let $\left\{ c_{\alpha}\right\} _{\alpha\in\left\{ 1\cdots n\right\} ^{2}}$
be some scalars and notice that if $\sum_{\alpha\in\left\{ 1\cdots n\right\} ^{2}}c_{\alpha}G_{\alpha}=0$
then 
\[
\sum_{\alpha\in\left\{ 1\cdots n\right\} ^{2}}c_{\alpha}U_{\alpha}=\sum_{\alpha\in\left\{ 1\cdots n\right\} ^{2}}c_{\alpha}P^{-1}G_{\alpha}P=P^{-1}\left(\sum_{\alpha\in\left\{ 1\cdots n\right\} ^{2}}c_{\alpha}G_{\alpha}\right)P=0;
\]
which implies $c_{\alpha}=0$ for all $\alpha$ because $\left\{ U_{\alpha}\right\} $
is also linearly independent.

To express $\D_{U_{\alpha}}$ using $\D_{G_{\alpha}}$, notice that
if we can express each $U_{\alpha}$ as a linear combination of $G_{\beta}$,
say, 
\begin{equation}
U_{\alpha}=\sum_{\beta\in\left\{ 1\cdots n\right\} ^{2}}\kappa_{\alpha\beta}G_{\beta},\label{eq:chgbasis001}
\end{equation}
then $\D_{U_{\alpha}}f\left(x\right)=\sum_{\beta\in\left\{ 1\cdots n\right\} ^{2}}\kappa_{\alpha\beta}\D_{G_{\alpha}}f\left(x\right)$
because
\[
\vecc\D_{U_{\alpha}}f\left(x\right)=\left[\vecc\nabla f\left(x\right)\right]^{T}\vecc U_{\alpha}=\sum_{\beta}\kappa_{\beta}\left\{ \left[\vecc\nabla f\left(x\right)\right]^{T}\vecc G_{\beta}\right\} =\sum_{\beta}\kappa_{\beta}\vecc\D_{G_{\alpha}}f\left(x\right).
\]
Therefore all we need is to solve for $\kappa_{\alpha\beta}$.

If we fix $\alpha$ and consider \eqref{chgbasis001} in its vectorised
form 
\begin{equation}
\vecc U_{\alpha}=\sum_{\beta}\kappa_{\alpha\beta}\vecc G_{\beta},\label{eq:chgbasis002}
\end{equation}
we notice that $\kappa_{\alpha\cdot}$, seen as a vector of length
$n^{2}$, must be of form $\kappa_{\alpha\cdot}=M^{-1}\vecc U_{\alpha}$,
where $M$ is the $n^{2}\times n^{2}$ change-of-basis matrix that
changes from the standard basis $\left\{ U^{ij}\right\} _{1\le i,j\le n}$
to $\left\{ G_{\beta}\right\} $; so according standard linear algebra,
\[
M=\left[\vecc G_{11}|\vecc G_{21}|\cdots|\vecc G_{nn}\right]=\left(P^{T}\otimes P^{-1}\right)\left[\vecc U_{11}|\vecc U_{21}|\cdots|\vecc U_{nn}\right].
\]
Therefore 
\[
\kappa_{\alpha\cdot}=M^{-1}\vecc U_{\alpha}=\left[\vecc U_{11}|\vecc U_{21}|\cdots|\vecc U_{nn}\right]^{-1}\left(P^{-T}\otimes P\right)\vecc U_{\alpha}=\vecc\left(PU_{\alpha}P^{-1}\right).
\]
Plugging the above back into \eqref{chgbasis002} gives the required
equation.
\end{proof}
With the above lemma, we can start proving the equation for $\frac{\partial^{2}V_{\eta}}{\partial H_{mn}\partial H_{ij}}$,
i.e., \eqref{hvdhdh}, by taking double directional derivatives. First,
$H=P\Lambda P^{-1}$ be the eigen-decomposition of $H$, and $G_{k\ell}=PU^{k\ell}P^{-1}$.
We start by taking the same steps as in \eqref{p901}-(\ref{eq:p9011-2}),
but instead of applying \eqref{nhdir} in the direction of $U^{ij}$
to obtain $\frac{\partial V_{\eta}}{\partial H_{ij}}$, we use $G_{k\ell}$
as our direction to obtain $\D_{G_{k\ell}}V_{\eta}$ (of course, here
$V_{\eta}$ is treated as a function of $H$); this gives us
\begin{equation}
\D_{G_{k\ell}}V_{\eta}=-\mathcal{S}\left(\int_{0}^{t}e^{-vH}J_{k\ell}\left(v\right)\Sigma e^{-vH^{T}}\d v\right)\label{eq:dgklvn}
\end{equation}
where 
\[
J_{k\ell}\left(v\right)=\int_{0}^{v}e^{\tau H}G_{k\ell}e^{-\tau H}\d v.
\]
Thus, diffentiating again in the direction of $G_{\alpha\beta}$,
we get
\begin{align}
\D_{G_{\alpha\beta}}\D_{G_{k\ell}}V_{\eta} & =-\mathcal{S}\left(\D_{G_{\alpha\beta}}\int_{0}^{t}e^{-vH}J_{k\ell}\left(v\right)\Sigma e^{-vH^{T}}\d v\right)\nonumber \\
 & =-\mathcal{S}\left[\int_{0}^{t}\left(\D_{G_{\alpha\beta}}e^{-vH}\right)J_{k\ell}\left(v\right)\Sigma e^{-vH^{T}}+e^{-vH}J_{k\ell}\left(v\right)\Sigma\left(\D_{G_{\alpha\beta}}e^{-vH^{T}}\right)\d v\right]\nonumber \\
 & \quad\quad\quad-\mathcal{S}\left(\int_{0}^{t}e^{-vH}\left(\D_{G_{\alpha\beta}}J_{k\ell}\left(v\right)\right)\Sigma e^{-vH^{T}}\d v\right).\label{eq:sodarnlong}
\end{align}
We simplify the above two terms separately. The first term can be
written as follows after expanding the $\mathcal{S}$ operator and
regrouping the sums.
\[
-\int_{0}^{t}\left(\D_{G_{\alpha\beta}}e^{-vH}\right)\mathcal{S}\left(J_{k\ell}\left(v\right)\Sigma\right)e^{-vH^{T}}+e^{-vH}\mathcal{S}\left(J_{k\ell}\left(v\right)\Sigma\right)\left(\D_{G_{\alpha\beta}}e^{-vH}\right)^{T}\d v.
\]
According to \eqref{nhdir}, we have $\D_{G_{\alpha\beta}}e^{-vH}=-e^{-vH}\int_{0}^{v}e^{\tau H}G_{\alpha\beta}e^{-\tau H}d\tau=-e^{-vH}J_{\alpha\beta}\left(v\right)$;
therefore after plugging into the above and factoring out $e^{-vH}$
and $e^{-vH^{T}}$, the above is further simplified to
\begin{equation}
\int_{0}^{t}e^{-vH}\mathcal{S}\left[J_{\alpha\beta}\left(v\right)\mathcal{S}\left(J_{k\ell}\left(v\right)\Sigma\right)\right]e^{-vH^{T}}\d v.\label{eq:troublesomeintegral}
\end{equation}
Notice that the integrals in $J_{k\ell}$ and $J_{\alpha\beta}$ takes
a simple form
\begin{align}
J_{k\ell}\left(v\right) & =\int_{0}^{v}e^{\tau H}G_{k\ell}e^{-\tau H}\d v\nonumber \\
 & =P\int_{0}^{v}\left(P^{-1}G_{k\ell}P\right)\odot\left[e^{\left(\lambda_{r}-\lambda_{s}\right)\tau}\right]_{rs}\d\tau P^{-1}\nonumber \\
 & =P\left\{ U^{k\ell}\odot\left[I_{0}\left(\lambda_{r}-\lambda_{s},v\right)\right]_{rs}\right\} P^{-1}\nonumber \\
 & =I_{0}\left(\lambda_{k}-\lambda_{\ell},v\right)PU^{k\ell}P^{-1}\label{eq:jklv-1}
\end{align}
where the last step above holds because all but one entry of $U^{k\ell}$
are zero. Therefore plugging this back into \eqref{troublesomeintegral}
allows us to solve the integral as follows. 
\begin{align*}
 & \int_{0}^{t}e^{-vH}\mathcal{S}\left[J_{\alpha\beta}\left(v\right)\mathcal{S}\left(J_{k\ell}\left(v\right)\Sigma\right)\right]e^{-vH^{T}}\d v\\
= & \int_{0}^{t}e^{-vH}I_{0}\left(\lambda_{\alpha}-\lambda_{\beta},v\right)I_{0}\left(\lambda_{k}-\lambda_{\ell},v\right)\mathcal{S}\left[PU^{\alpha\beta}P^{-1}\mathcal{S}\left(PU^{k\ell}P^{-1}\Sigma\right)\right]e^{-vH^{T}}\d v\\
= & P\int_{0}^{t}\left\{ e^{-v\Lambda}I_{0}\left(\lambda_{\alpha}-\lambda_{\beta},v\right)I_{0}\left(\lambda_{k}-\lambda_{\ell},v\right)P^{-1}\mathcal{S}\left[PU^{\alpha\beta}P^{-1}\mathcal{S}\left(PU^{k\ell}P^{-1}\Sigma\right)\right]P^{-T}e^{-v\Lambda}\right\} \d vP^{T}\\
= & P\int_{0}^{t}\left[I_{0}\left(\lambda_{\alpha}-\lambda_{\beta},v\right)I_{0}\left(\lambda_{k}-\lambda_{\ell},v\right)P^{-1}\mathcal{S}\left[PU^{\alpha\beta}P^{-1}\mathcal{S}\left(PU^{k\ell}P^{-1}\Sigma\right)\right]P^{-T}\right]\odot\left[e^{-\left(\lambda_{r}+\lambda_{s}\right)v}\right]_{rs}\d vP^{T}\\
= & P\left\{ \left\{ P^{-1}\mathcal{S}\left[PU^{\alpha\beta}P^{-1}\mathcal{S}\left(PU^{k\ell}P^{-1}\Sigma\right)\right]P^{-T}\right\} \odot\left[\int_{0}^{t}e^{-\left(\lambda_{r}+\lambda_{s}\right)v}I_{0}\left(\lambda_{\alpha}-\lambda_{\beta},v\right)I_{0}\left(\lambda_{k}-\lambda_{\ell},v\right)\d v\right]_{rs}\right\} P^{T}\\
= & P\left\{ \left\{ P^{-1}\mathcal{S}\left[PU^{\alpha\beta}P^{-1}\mathcal{S}\left(PU^{k\ell}P^{-1}\Sigma\right)\right]P^{-T}\right\} \odot\left[Q\left(-\lambda_{r}-\lambda_{s},\lambda_{\alpha}-\lambda_{\beta},\lambda_{k}-\lambda_{\ell},t\right)\right]_{rs}\right\} P^{T};
\end{align*}
or, in fact, an alternative form of the above integral is
\[
P\left\{ \left\{ P^{-1}\mathcal{S}\left[PU^{\alpha\beta}P^{-1}\mathcal{S}\left(PU^{k\ell}P^{-1}\Sigma\right)\right]P^{-T}\right\} \odot\left[\int_{0}^{t}e^{-\left(\lambda_{r}+\lambda_{s}\right)x}\int_{0}^{x}\int_{0}^{x}e^{\left(\lambda_{\alpha}-\lambda_{\beta}\right)y+\left(\lambda_{k}-\lambda_{\ell}\right)z}\d z\d y\d x\right]_{rs}\right\} P^{T}.
\]
Furthermore, the left side of the Hadamard product above simplifies
to
\begin{align*}
P^{-1}\mathcal{S}\left[PU^{\alpha\beta}P^{-1}\mathcal{S}\left(PU^{k\ell}P^{-1}\Sigma\right)\right]P^{-T} & =P^{-1}\mathcal{S}\left[PU^{\alpha\beta}P^{-1}PU^{k\ell}P^{-1}\Sigma+PU^{\alpha\beta}P^{-1}\Sigma P^{-T}U^{\ell k}P^{T}\right]P^{-T}\\
 & =P^{-1}\mathcal{S}\left[PU^{\alpha\beta}U^{k\ell}P^{-1}\Sigma+PU^{\alpha\beta}P^{-1}\Sigma P^{-T}U^{\ell k}P^{T}\right]P^{-T}\\
 & =U^{\alpha\beta}U^{k\ell}P^{-1}\Sigma P^{-T}+U^{\alpha\beta}P^{-1}\Sigma P^{-T}U^{\ell k}\\
 & \quad\quad+P^{-1}\Sigma P^{-T}U^{\ell k}U^{\beta\alpha}+U^{k\ell}P^{-1}\Sigma P^{-T}U^{\beta\alpha}\\
 & =U^{\alpha\beta}U^{k\ell}\overline{\Sigma}+U^{\alpha\beta}\overline{\Sigma}U^{\ell k}+\overline{\Sigma}U^{\ell k}U^{\beta\alpha}+U^{k\ell}\overline{\Sigma}U^{\beta\alpha}\\
 & =\mathcal{S}\left(U^{\alpha\beta}U^{k\ell}\overline{\Sigma}+U^{\alpha\beta}\overline{\Sigma}U^{\ell k}\right).
\end{align*}
Therefore the first term of $\prettyref{eq:sodarnlong}$ is
\begin{equation}
P\left\{ \mathcal{S}\left(U^{\alpha\beta}U^{k\ell}\overline{\Sigma}+U^{\alpha\beta}\overline{\Sigma}U^{\ell k}\right)\odot\left[Q\left(-\lambda_{r}-\lambda_{s},\lambda_{\alpha}-\lambda_{\beta},\lambda_{k}-\lambda_{\ell},t\right)\right]_{rs}\right\} P^{T}.\label{eq:p10nicet1}
\end{equation}
On the other hand, to evaluate the second term of \eqref{sodarnlong}
we obviously need $\D_{G_{\alpha\beta}}J_{k\ell}\left(v\right)$.
This can be found using \eqref{nhdir} again, which gives
\begin{align}
\D_{G_{\alpha\beta}}J_{k\ell}\left(v\right) & =\D_{G_{\alpha\beta}}\int_{0}^{v}e^{\tau H}G_{k\ell}e^{-\tau H}\d\tau\nonumber \\
 & =\int_{0}^{v}\left(\D_{G_{\alpha\beta}}e^{\tau H}\right)G_{k\ell}e^{-\tau H^{T}}+e^{\tau H}G_{k\ell}\left(\D_{G_{\alpha\beta}}e^{-\tau H}\right)\d\tau\nonumber \\
 & =\int_{0}^{v}e^{\tau H}\int_{0}^{\tau}e^{-yH}G_{\alpha\beta}e^{yH}\d yG_{k\ell}e^{-\tau H^{T}}-e^{\tau H}G_{k\ell}e^{-\tau H}\int_{0}^{\tau}e^{yH}G_{\alpha\beta}e^{-yH}\d\tau\nonumber \\
 & =\int_{0}^{v}e^{\tau H}I_{0}\left(\lambda_{\beta}-\lambda_{\alpha},\tau\right)PU^{\alpha\beta}P^{-1}G_{k\ell}e^{-\tau H}-e^{\tau H}G_{k\ell}e^{-\tau H}I_{0}\left(\lambda_{\alpha}-\lambda_{\beta},\tau\right)PU^{\alpha\beta}P^{-1}\d\tau\nonumber \\
 & =\int_{0}^{v}e^{\tau H}I_{0}\left(\lambda_{\beta}-\lambda_{\alpha},\tau\right)PU^{\alpha\beta}U^{k\ell}P^{-1}e^{-\tau H}\d\tau\nonumber \\
 & \quad\quad-\int_{0}^{v}I_{0}\left(\lambda_{\alpha}-\lambda_{\beta},\tau\right)e^{\tau H}PU^{k\ell}P^{-1}e^{-\tau H}\d\tau PU^{\alpha\beta}P^{-1}\nonumber \\
 & =P\left\{ \left(U^{\alpha\beta}U^{k\ell}\right)\odot\left[\int_{0}^{v}I_{0}\left(\lambda_{\beta}-\lambda_{\alpha},\tau\right)e^{\left(\lambda_{r}-\lambda_{s}\right)\tau}\d\tau\right]_{rs}\right\} P^{-1}\nonumber \\
 & \quad\quad-P\left\{ U^{k\ell}\odot\left[\int_{0}^{v}I_{0}\left(\lambda_{\alpha}-\lambda_{\beta},\tau\right)e^{\left(\lambda_{r}-\lambda_{s}\right)\tau}\d\tau\right]_{rs}\right\} U^{\alpha\beta}P^{-1}.\label{eq:dgabdjkl}
\end{align}
The two terms above seems dissimilar, but in fact, they are of the
same form. First, a matrix of form $\left(U^{\alpha\beta}U^{k\ell}\right)\odot X$
contains one and only one element $X_{\alpha\ell}$ at index $\left(\alpha,\ell\right)$
if and only if $\beta=k$; on the other hand, the matrix $\left(U^{k\ell}\odot X\right)U^{\alpha\beta}$
also contains one and only one element $X_{k\beta}$ at index $\left(k,\beta\right)$
if and only if $\alpha=\ell$. Therefore matching the pattern, the
latter can also be written as $\left(U^{k\ell}\odot X\right)U^{\alpha\beta}=\left(U^{k\ell}U^{\alpha\beta}\right)\odot X$.
Also similarly, it is easy to check that for any matrix $T$ of suitable
size we have $\left[\left(U^{\alpha\beta}U^{k\ell}\right)\odot X\right]T=\left(U^{\alpha\beta}U^{k\ell}T\right)\odot\left[X_{\alpha\ell}\right]_{rs}$;
furthermore, if $X$ were a function of either $\lambda_{\beta}$
or $\lambda_{k}$, we can always exchange the two because we know
that $U^{\alpha\beta}U^{k\ell}$ is a non-zero matrix only if $\beta=k$.
Using this, we can write $\D_{G_{\alpha\beta}}J_{k\ell}\left(v\right)\overline{\Sigma}$
in a more symmetric form:

\begin{equation}
\D_{G_{\alpha\beta}}J_{k\ell}\left(v\right)\overline{\Sigma}=P\left\{ \begin{aligned}\left(U^{\alpha\beta}U^{k\ell}\overline{\Sigma}\right)\odot\left[\int_{0}^{v}I_{0}\left(\lambda_{k}-\lambda_{\alpha},\tau\right)e^{\left(\lambda_{\alpha}-\lambda_{\ell}\right)\tau}\d\tau\right]_{rs}\quad\quad\quad\\
-\left(U^{k\ell}U^{\alpha\beta}\overline{\Sigma}\right)\odot\left[\int_{0}^{v}I_{0}\left(\lambda_{\ell}-\lambda_{\beta},\tau\right)e^{\left(\lambda_{k}-\lambda_{\beta}\right)\tau}\d\tau\right]_{rs}
\end{aligned}
\right\} P^{-1}.\label{eq:dgabdjkl2}
\end{equation}

Therefore the second term of \eqref{sodarnlong} can be derived after
writing the matrix exponentials in their spectral form:
\begin{align*}
 & -\mathcal{S}\left(\int_{0}^{t}e^{-vH}\left(\D_{G_{\alpha\beta}}J_{k\ell}\left(v\right)\right)\Sigma e^{-vH^{T}}\d v\right)\\
= & -\mathcal{S}\left\{ P\int_{0}^{t}e^{-v\Lambda}\left\{ \begin{gathered}\left(U^{\alpha\beta}U^{k\ell}\overline{\Sigma}\right)\odot\left[\int_{0}^{v}I_{0}\left(\lambda_{k}-\lambda_{\alpha},\tau\right)e^{\left(\lambda_{\alpha}-\lambda_{\ell}\right)\tau}\d\tau\right]_{rs}\quad\quad\\
-\left(U^{k\ell}U^{\alpha\beta}\overline{\Sigma}\right)\odot\left[\int_{0}^{v}I_{0}\left(\lambda_{\ell}-\lambda_{\beta},\tau\right)e^{\left(\lambda_{k}-\lambda_{\beta}\right)\tau}\d\tau\right]_{rs}
\end{gathered}
\right\} e^{-v\Lambda}\d vP^{T}\right\} \\
= & -P\left\{ \mathcal{S}\left\{ \begin{gathered}\left(U^{\alpha\beta}U^{k\ell}\overline{\Sigma}\right)\odot\left[\int_{0}^{t}\int_{0}^{x}\int_{0}^{y}e^{-\left(\lambda_{r}+\lambda_{s}\right)x+\left(\lambda_{\alpha}-\lambda_{\ell}\right)y+\left(\lambda_{k}-\lambda_{\alpha}\right)z}\d z\d y\d x\right]_{rs}\quad\quad\\
-\left(U^{k\ell}U^{\alpha\beta}\overline{\Sigma}\right)\odot\left[\int_{0}^{t}\int_{0}^{x}\int_{0}^{y}e^{-\left(\lambda_{r}+\lambda_{s}\right)x+\left(\lambda_{k}-\lambda_{\beta}\right)y+\left(\lambda_{\ell}-\lambda_{\beta}\right)z}\d z\d y\d x\right]_{rs}
\end{gathered}
\right\} \right\} P^{T}\\
= & -P\left\{ \mathcal{S}\left\{ \begin{gathered}\left(U^{\alpha\beta}U^{k\ell}\overline{\Sigma}\right)\odot\left[\int_{0}^{t}e^{-\left(\lambda_{r}+\lambda_{s}\right)x}K\left(\lambda_{\alpha}-\lambda_{\ell},\lambda_{k}-\lambda_{\alpha},x\right)\d x\right]_{rs}\quad\quad\\
-\left(U^{k\ell}U^{\alpha\beta}\overline{\Sigma}\right)\odot\left[\int_{0}^{t}e^{-\left(\lambda_{r}+\lambda_{s}\right)x}K\left(\lambda_{k}-\lambda_{\beta},\lambda_{\ell}-\lambda_{\beta},x\right)\d x\right]_{rs}
\end{gathered}
\right\} \right\} P^{T}
\end{align*}
The integrals above would become $\mathring{L}_{k\ell\alpha\iota}^{\eta}$
and $\mathring{R}_{k\ell\beta\iota}^{\eta}$ defined in the manuscript
if we use \eqref{k0sol} and write $U^{\alpha\beta}U^{k\ell}$ and
$U^{k\ell}U^{\alpha\beta}$ into their non-matrix Iversion bracket
form. Combining this with the first term \eqref{p10nicet1}, we already
have got both terms of $\D_{G_{\alpha\beta}}\D_{G_{k\ell}}V_{\eta}$.
The required equation (\eqref{hvdhdh}) for $\frac{\partial^{2}V_{\eta}}{\partial H_{mn}\partial H_{ij}}$
follows directly from applying \lemref{basishess} twice using $\D_{G_{\alpha\beta}}\D_{G_{k\ell}}V_{\eta}$.

Next, we consider $\frac{\partial^{2}V_{\eta}}{\partial L_{mn}\partial H_{ij}}$.
We directly differentiate \eqref{p9011-2} with respect to $L_{mn}$,
which gives
\[
\frac{\partial^{2}V_{\eta}}{\partial L_{mn}\partial H_{ij}}=-\mathcal{S}\left(\int_{0}^{t}e^{-vH}J_{ij}\left(v\right)\frac{\partial\Sigma}{\partial L_{mn}}e^{-vH^{T}}\d v\right)=-\mathcal{S}\left(\int_{0}^{t}e^{-vH}J_{ij}\left(v\right)\frac{\partial\Sigma}{\partial L_{mn}}e^{-vH^{T}}\d v\right)
\]
where $J_{ij}\left(v\right)=\int_{0}^{v}e^{\tau H}U^{ij}e^{-\tau H}\d\tau$.
Because $\Sigma=LL^{T}$, it is easy to verify that
\begin{equation}
\frac{\partial\Sigma}{\partial L_{mn}}=\frac{\partial LL^{T}}{\partial L_{mn}}=\mathcal{S}\left(U^{mn}L^{T}\right).\label{eq:dsigdlmn}
\end{equation}
As the rest of the proof for $\frac{\partial^{2}V_{\eta}}{\partial H_{ij}}$
makes no assumption about $\Sigma$, we can conclude that $\frac{\partial^{2}V_{\eta}}{\partial L_{mn}\partial H_{ij}}$
is the same as $\frac{\partial V_{\eta}}{\partial H_{ij}}$ except
replacing $\Sigma$ with $\mathcal{S}\left(U^{mn}L^{T}\right)$; hence
\eqref{hvdldh} holds.

Similarly, $\frac{\partial^{2}V_{\eta}}{\partial L_{mn}\partial L_{ij}}$
(\eqref{hvdldl}) is obtained by directly taking double partial derivative
on the definition of $V_{\eta}$ in \eqref{VOU}, where we get
\[
\frac{\partial^{2}V_{\eta}}{\partial L_{mn}\partial L_{ij}}==\intop_{0}^{t_{\eta}}e^{-Hv}\left(\frac{\partial^{2}LL^{T}}{\partial L_{mn}\partial L_{ij}}\right)e^{-H^{T}v}\mathrm{d}v,
\]
and $\frac{\partial^{2}LL^{T}}{\partial L_{mn}\partial L_{ij}}=\mathcal{S}\left(U^{ij}U^{nm}\right)$
can be trivially obtained by differntiating \eqref{dsigdlmn} once
more.

\eqref{hwdhdh}, and (\ref{eq:hwdmudh}), i.e., the formula for $\frac{\partial^{2}w_{\eta}}{\partial H_{mn}\partial H_{ij}}$
and $\frac{\partial^{2}w_{\eta}}{\partial\mu_{m}\partial H_{ij}}$
trivially follows from \eqref{dwdhij}.

Finally, to get the expression of $\frac{\partial^{2}\Phi_{\eta}}{\partial H_{mn}\partial H_{ij}}$
in \eqref{hphidhdh}, we use the same strategy as in $\frac{\partial^{2}V_{\eta}}{\partial H_{mn}\partial H_{ij}}$,
namely, to find $\D_{G_{\alpha\beta}}\D_{G_{k\ell}}\Phi_{\eta}$ for
all $\alpha$, $\beta$, $k$, and $\ell$ instead, and apply \lemref{basishess}
twice to obtain $\frac{\partial^{2}\Phi_{\eta}}{\partial H_{mn}\partial H_{ij}}$.
Applying \eqref{nhdir} we immediately get $\D_{G_{k\ell}}\Phi_{\eta}=\D_{G_{k\ell}}e^{-tH}=-e^{-tH}J_{k\ell}\left(t\right)$.
So using the product rule of differentiation we have 
\[
\D_{G_{\alpha\beta}}\D_{G_{k\ell}}\Phi_{\eta}=e^{-tH}J_{\alpha\beta}\left(t\right)J_{k\ell}\left(t\right)-e^{-tH}\D_{G_{\alpha\beta}}J_{k\ell}\left(t\right).
\]
We have already seen $J_{k\ell}\left(t\right)$ and $\D_{G_{\alpha\beta}}J_{k\ell}\left(t\right)$
in \eqref{jklv-1} and (\ref{eq:dgabdjkl2}), therefore the above
can be written as
\begin{align*}
\D_{G_{\alpha\beta}}\D_{G_{k\ell}}\Phi_{\eta} & =e^{-tH}I_{0}\left(\lambda_{\alpha}-\lambda_{\beta},t\right)I_{0}\left(\lambda_{k}-\lambda_{\ell},t\right)PU^{\alpha\beta}U^{k\ell}P^{-1}\\
 & \quad\quad\quad-e^{-tH}P\left\{ \begin{aligned}\left(U^{\alpha\beta}U^{k\ell}\right)\odot\left[\int_{0}^{t}I_{0}\left(\lambda_{k}-\lambda_{\alpha},\tau\right)e^{\left(\lambda_{r}-\lambda_{s}\right)\tau}\d\tau\right]_{rs}\quad\quad\\
-\left(U^{k\ell}U^{\alpha\beta}\right)\odot\left[\int_{0}^{t}I_{0}\left(\lambda_{\ell}-\lambda_{\beta},\tau\right)e^{\left(\lambda_{r}-\lambda_{s}\right)\tau}\d\tau\right]_{rs}
\end{aligned}
\right\} P^{-1}\\
 & =Pe^{-t\Lambda}\left\{ \begin{aligned}\left[I_{0}\left(\lambda_{\alpha}-\lambda_{\beta},t\right)I_{0}\left(\lambda_{k}-\lambda_{\ell},t\right)-K_{0}\left(\lambda_{\alpha}-\lambda_{\ell},\lambda_{k}-\lambda_{\alpha},t\right)\right]U^{\alpha\beta}U^{k\ell}\quad\quad\\
+K_{0}\left(\lambda_{k}-\lambda_{\beta},\lambda_{\ell}-\lambda_{\beta},t\right)U^{k\ell}U^{\alpha\beta}
\end{aligned}
\right\} P^{-1}.
\end{align*}

\section*{Proof of \eqref{choleskychain-1}}
\begin{proof}
Chain rule for the Cholesky decomposition in \eqref{choleskychain-1}
can be seen as follows.
\begin{align*}
\frac{\partial\Omega}{\partial L_{\ell m}} & =\sum_{i=1}^{k}\sum_{j=1}^{k}\frac{\partial V_{ij}}{\partial L_{\ell m}}\frac{\partial\Omega}{\partial V_{ij}}\\
 & =\sum_{i=1}^{k}\sum_{j=1}^{k}\left(\sum_{q=1}^{k}\frac{\partial L_{iq}L_{jq}}{\partial L_{\ell m}}\right)\frac{\partial\Omega}{\partial V_{ij}}\\
 & =\sum_{i=1}^{k}\sum_{j=1}^{k}\left[2\delta_{ij}\delta_{i\ell}L_{\ell m}+\left(1-\delta_{ij}\right)\left(\delta_{i\ell}L_{jm}+\delta_{j\ell}L_{im}\right)\right]\frac{\partial\Omega}{\partial V_{ij}},\\
 & =2L_{\ell m}\frac{\partial\Omega}{\partial V_{\ell\ell}}+\sum_{\substack{j=1\\
j\ne\ell
}
}^{k}L_{jm}\frac{\partial\Omega}{\partial V_{\ell j}}+\sum_{\substack{i=1\\
i\ne\ell
}
}^{k}L_{im}\frac{\partial\Omega}{\partial V_{i\ell}}\\
 & =\sum_{j=1}^{k}L_{jm}\frac{\partial\Omega}{\partial V_{\ell j}}+\sum_{i=1}^{k}L_{im}\frac{\partial\Omega}{\partial V_{i\ell}}\\
 & =\sum_{i=1}^{k}L_{im}\left(\frac{\partial\Omega}{\partial V_{\ell i}}+\frac{\partial\Omega}{\partial V_{i\ell}}\right).
\end{align*}
\end{proof}

\section*{Explicit Expression for Some Exponential Integrals \label{sec:Integrals-Explicit-Expression}}

In the main manuscript the following definite integrals are defined;

\begin{align}
I_{n}\left(a,t\right) & =\intop_{0}^{t}y^{n}e^{ay}\textrm{d}y,\label{eq:intin}\\
K_{n}\left(a,b,t\right) & =\intop_{0}^{t}e^{ay}I_{n}\left(b,y\right)\textrm{d}y,\label{eq:intkn}\\
Q\left(a,b,c,t\right) & =\intop_{0}^{t}e^{ay}I_{0}\left(b,y\right)I_{0}\left(c,y\right)\textrm{d}y;\label{eq:intq}
\end{align}
and only some of them are used

\begin{align}
I_{0}(a,t) & =\left[a=0\right]t+\left[a\neq0\right]\frac{e^{at}-1}{a}\textrm{, }\nonumber \\
I_{1}(a,t) & =\left[a=0\right]\frac{{t^{2}}}{2}+\left[a\neq0\right]\frac{e^{at}(at-1)+1}{a^{2}}\textrm{, }\nonumber \\
I_{2}\left(a,t\right) & =\left[a=0\right]\frac{t^{3}}{3}+\left[a\neq0\right]e^{at}\left(\frac{t^{2}}{a}-\frac{2t}{a^{2}}+\frac{2}{a^{3}}\right)-\frac{2}{a^{3}},\nonumber \\
K_{0}\left(a,b,t\right) & =\left[b=0\right]I_{1}\left(a,t\right)+\left[b\neq0\right]\frac{I_{0}\left(a+b,t\right)-I_{0}(a,t)}{b},\label{eq:k0sol}\\
K_{1}(a,b,t) & =\left[b=0\right]\frac{I_{2}\left(a,t\right)}{2}+\left[b\neq0\right]\left\{ \frac{I_{1}(a+b,t)}{b}-\frac{I_{0}\left(a+b,t\right)-I_{0}(a,t)}{b^{2}}\right\} ,\nonumber \\
Q\left(a,b,c,t\right) & =\begin{cases}
t^{3}/3 & a=b=c=0\\
\left[I_{0}\left(b,t\right)-t\right]/b & a=0,\,\,b\neq0,\,\,c=0\\
\left[I_{0}\left(c,t\right)-t\right]/c & a=0,\,\,b=0,\,\,c\neq0\\
I_{2}\left(a,t\right) & a\neq0,\,\,b=0,\,\,c=0\\
\left[I_{1}\left(a+b,t\right)-I_{1}\left(a,t\right)\right]/b & a\neq0,\,\,b\neq0,\,\,c=0\\
\left[I_{1}\left(a+c,t\right)-I_{1}\left(a,t\right)\right]/c & a\neq0,\,\,b=0,\,\,c\neq0\\
\left[I_{0}\left(a+b+c,t\right)+I_{0}\left(a,t\right)-I_{0}\left(a+b,t\right)-I_{0}\left(a+c,t\right)\right]/(bc) & \textrm{otherwise.}
\end{cases}\nonumber 
\end{align}

\end{document}